\documentclass[lettersize,onecolumn]{IEEEtran}
\usepackage[utf8]{inputenc}
\usepackage[T1]{fontenc}
\usepackage{url}
\usepackage{ifthen}
\usepackage{amssymb,amsfonts}
\usepackage{amsthm}
\usepackage{mathrsfs}
\usepackage{graphicx}
\usepackage{cite}

\usepackage[cmex10]{amsmath}
\newtheorem{theorem}{Theorem}[]
\newtheorem{example}{Example}[]%
\newtheorem{definition}{Definition}[]
\newtheorem{remark}{Remark}[]%
\newtheorem{lemma}{Lemma}[]
\newtheorem{proposition}{Proposition}[]   
\newtheorem{corollary}{Corollary}[]  
\newtheorem{conjecture}{Conjecture}[]                          

\begin{document}
	\title{On the Hamming Distance and LCD Properties of Binary Polycyclic Codes and Their Duals} 
	\author{
		Sujata Bansal and
		Pramod Kumar Kewat
		\thanks{
			Sujata Bansal and Pramod Kumar Kewat are with the Department of Mathematics and Computing, 
			Indian Institute of Technology (ISM), Dhanbad, India (email: sujatabansal8@gmail.com, pramodk@iitism.ac.in).
		}
	}
	\maketitle
	
	\begin{abstract}
Polycyclic codes offer a natural generalization of cyclic codes and provide a broader algebraic framework for constructing linear codes with good parameters. In this paper, we study binary polycyclic codes associated with powers of irreducible polynomials. We first determine their complete algebraic structure and then develop general results on their minimum Hamming distance, including several exact values and bounds. We also examine the Euclidean duals of these codes and derive corresponding results on the Hamming distance of the dual codes. Furthermore, we study the LCD (linear complementary dual) properties of binary polycyclic codes, establish necessary and sufficient conditions for such codes to be LCD codes, and construct several families of binary LCD codes. Our constructions also yield many optimal and LCD optimal binary linear codes, including codes of larger lengths. We then focus on binary polycyclic codes associated with powers of the self-reciprocal irreducible trinomials $x^{2\cdot3^v}+x^{3^v}+1$, where $v\geq0$. For this class, we determine the exact Hamming distance of all such codes and show that these codes are reversible. Moreover, we show that these codes are LCD codes in certain cases. In addition, we propose a conjecture asserting that all  binary polycyclic codes associated with $\big(x^{2\cdot3^v}+x^{3^v}+1\big)^{2^\mathcal{T}}$, where  $v\geq 0$ and $\mathcal{T}\geq1$, are LCD codes.  These results demonstrate that binary polycyclic codes form a rich source of structured codes with strong distance, duality, reversibility, and LCD properties.
		
	\end{abstract}
\section{Introduction}\label{sec1}
Cyclic codes are one of the most classical and extensively studied families of error-correcting codes because of their elegant algebraic structure and practical efficiency. A cyclic code of length $n$ over a finite field $\mathbb{F}_q$ corresponds to an ideal of the quotient ring $\frac{\mathbb{F}_q[x]}{\langle x^n-1 \rangle}$. This representation enables the analysis of cyclic codes using algebraic methods and provides efficient encoding and decoding procedures. Consequently, cyclic codes have remained a central topic in coding theory for several decades \cite{macwilliams1977theory,huffman2010fundamentals}.
Despite their rich structure and extensive literature, cyclic codes are inherently constrained by the polynomial $x^n-1$. In particular, every generator polynomial of a cyclic code of length $n$ must divide $x^n-1$. Consequently, the search space for good codes within the cyclic family may be limited in certain situations. This has motivated the study of broader classes obtained by replacing $x^n-1$ with a more general polynomial $f(x)$. A polycyclic code can be viewed as an ideal of a quotient ring of the form $\frac{\mathbb{F}_q[x]}{\langle f(x) \rangle}$, where $f(x)\in \mathbb{F}_q[x]$ (see \cite{lopez2009dual}). This framework includes several well-known families, such as cyclic and constacyclic codes, and therefore significantly enlarges the search space for linear codes with desirable parameters. 

For a linear code, the fundamental parameters are its length $n$, dimension $k$, and minimum Hamming distance $d$, usually denoted as $[n,k,d]$. In cyclic and polycyclic codes, the parameters $n$ and $k$ can often be determined directly from the defining polynomial and the generator polynomial of the code. In contrast, the determination of the minimum Hamming distance $d$ is generally far more difficult and remains one of the central problems in coding theory. Consequently, the computation of minimum distance has been widely investigated for several structured families of codes, including cyclic codes, repeated-root cyclic codes, and constacyclic codes \cite{van2003minimum,dinh2008linear,moreno2002minimum,feng2002generalized,luo2024improved,alfarano2021roos,augot1996description}. This difficulty is particularly significant because the error-correcting capability of a code depends on $d$. In the present work, we address this problem for binary polycyclic codes associated with powers of irreducible polynomials by deriving general results, exact values, and bounds for the minimum Hamming distance of all codes in the family under consideration.

An important distinction between cyclic and polycyclic codes arises in the study of duality. The Euclidean dual of a cyclic code is again cyclic. In contrast, the dual of a polycyclic code does not necessarily remain polycyclic. Instead, it is naturally related to another class of linear codes known as sequential codes\cite{lopez2009dual}. Sequential codes were introduced by Hou et al.\cite{hou2009rational} in 2009, where several examples were given to show that, for certain parameters, these codes can achieve an optimal minimum distance that is unattainable within the class of cyclic codes. This has significantly increased interest in the study of sequential codes and in understanding their structural and distance properties. These observations indicate that duality in the polycyclic setting is considerably richer and more subtle than in the classical cyclic case, making the subject both intriguing and worthy of further investigation. During the past decade, polycyclic codes over many alphabets, such as Galois rings, finite fields, and local rings, have been extensively studied \cite{bajalan2022transform,bajalan2024structure,bajalan2025polycyclic,shi2020polycyclic,lopez2013polycyclic}. Recently, Aydin et al. \cite{aydin2022polycyclic}  investigated polycyclic codes over finite fields associated with monic trinomials, analyzed their algebraic properties, and constructed many optimal codes.

A similar contrast arises for LCD codes. A linear code $C$ is called an LCD code if $C \cap C^\perp=\{0\}$. LCD codes, introduced by Massey \cite{massey1992linear}, have attracted considerable attention because of their applications in communication systems, data storage, consumer electronics, and, more recently, cryptography, particularly in resisting certain side-channel and fault injection attacks (see, for example, \cite{carlet2014complementary,alahmadi2020multisecret,mesnager2017complementary}). For cyclic codes, the dual of a cyclic code is again cyclic, and hence the intersection $C \cap C^\perp$ can be characterized through the generator polynomials of $C$ and $C^\perp$. The authors in \cite{yang1994condition} provided the criteria necessary for a cyclic code to have a complementary dual. Therefore, in many situations, verifying whether a cyclic code is an LCD code is relatively simple. However, for polycyclic codes, no equally straightforward general criterion exists, and assessing the LCD property typically requires a deeper analysis of the interaction between the code and its dual. A major breakthrough was obtained by Carlet et al.\cite{carlet2018linear} in 2018, who proved that for $q>3$, every linear code over the finite field $\mathbb{F}_q$ is equivalent to an LCD code. This result narrows the construction problem primarily to the binary and ternary cases. More recently, in \cite{bansal2026binary}, the authors studied binary polycyclic codes and identified several families of binary linear codes that are LCD codes. This indicates that polycyclic codes are promising candidates for constructing LCD codes with good parameters.

Motivated by these observations, in this paper we study binary polycyclic codes associated with powers of irreducible polynomials. Let $P(x)$ be a binary irreducible polynomial of degree $m\geq 2,$ and let $\mathcal{L}\geq 2$ be a positive integer. We consider the quotient ring $\frac{\mathbb{F}_2[x]}{\langle P(x)^\mathcal{L}\rangle}$ and investigate the corresponding family of binary polycyclic codes. We first determine the complete algebraic structure of these binary polycyclic codes by describing all ideals of the ambient quotient ring. We then develop general methods to study their minimum Hamming distance and derive several exact values and bounds. Next, we investigate the Euclidean duals of these codes and obtain corresponding results on the Hamming distance of the dual family.

We also study the LCD property of binary polycyclic codes in this setting. We establish necessary and sufficient conditions for a polycyclic code associated with $P(x)^\mathcal{L}$ to be an LCD code.  Using these criteria, we construct several  families of binary LCD codes.  Another important aspect in this direction is determining the maximum possible Hamming distance, denoted by $LCD(n,k)$, of an LCD code with length $n$ and dimension $k.$ This problem has received considerable attention in the literature. The exact value of $LCD(n,k)$ for $1\leq k \leq 5$ is known from \cite{dougherty2017combinatorics,galvez2018bounds, harada2019binary, araya2021characterization, liu2024minimum}. In addition, several results are available for LCD codes of large dimension \cite{araya2021minimum, araya2024characterizations}. Significant progress has also been made in determining exact values or bounds for $LCD(n,k)$ when $n\leq 50$ \cite{bouyuklieva2021optimal,ishizuka2023construction,li2024several,wang2024new}. However, for lengths beyond $50$, comparatively fewer results are known, and the existing bounds allow estimation of $LCD(n,k)$ mainly when $1\leq k \leq 5$ or $n-7\leq k \leq n$. Our constructions yield many optimal and LCD optimal binary linear codes, including codes of larger lengths. These results demonstrate that binary polycyclic codes form a rich source of structured codes with strong distance and duality properties.

A substantial portion of this work focuses on irreducible trinomials of the form $x^{2n}+x^{n}+1$, where $n$ is a positive integer. It is well known that such trinomials are irreducible over $\mathbb{F}_2$ if and only if $n =3^v$ for some integer $v \geq 0$ \cite{lidl1997finite}. Consequently,  these polynomials characterize all self-reciprocal irreducible trinomials over $\mathbb{F}_2.$ For this family, we determine the Hamming distance of all associated binary polycyclic codes corresponding to arbitrary powers of $x^{2\cdot3^v}+x^{3^v}+1$. Moreover, since these trinomials are self-reciprocal, the corresponding codes are reversible, thereby endowing them with an additional structural property of both theoretical and practical importance. We show that these codes are LCD codes in certain cases. Motivated by computational evidence, we further conjecture that all  binary polycyclic codes associated with $\left(x^{2\cdot3^v}+x^{3^v}+1\right)^{2^\mathcal{T}}$, where  $v\geq 0$, $\mathcal{T}\geq1$, are LCD codes. 

Throughout this paper, we adopt the notation $\mathbb{F}_q$ to denote a finite field of order $q$, $P(x)$ to represent a binary irreducible polynomial of degree $m\geq2$, and $\mathcal{P}$ for the quotient ring $\frac{\mathbb{F}_2[x]}{\left\langle  P(x)^\mathcal{L} \right\rangle}$, where $\mathcal{L}\geq2.$ We also follow the convention that the zero polynomial is assigned a degree of $-1$. The organization of the paper is as follows: Section~\ref{sec2} presents the necessary preliminaries and notations. In Section~\ref{sec3}, we examine the structure and Hamming distance of binary polycyclic codes associated with powers of irreducible polynomials.   Section~\ref{sec4} is devoted to the Hamming distance of the dual codes.  In Section~\ref{sec5}, we investigate the LCD property and present constructions of LCD codes. Section~\ref{sec6} studies codes associated with powers of the self-reciprocal irreducible trinomials $x^{2\cdot3^v}+x^{3^v}+1$, including their Hamming distance, LCD property, and reversibility. Finally, in Section~\ref{sec7}, we conclude the paper and discuss some directions for further research.
\section{Preliminaries}\label{sec2}
\begin{definition}
	Let $\mathfrak{C}$ be a subset of $\mathbb{F}_q^n,$ where $\mathbb{F}_q$ is the finite field of order $q$. Then $\mathfrak{C}$ is called a code of length $n$ over  $\mathbb{F}_q$. Moreover, if $\mathfrak{C}$ is a vector subspace of  $\mathbb{F}_q^n$,  then it is said to be a linear code. The dual code of $\mathfrak{C}$ is defined as 
	\begin{equation*}
		\mathfrak{C}^\perp= \left\{ \mathbf{a} \in \mathbb{F}_q^n \mid \mathbf{a}\cdot\mathbf{c}=0  \text{ for all } \mathbf{c}\in \mathfrak{C}  \right\},
	\end{equation*} where $\mathbf{a}\cdot\mathbf{c}$ is the Euclidean inner product of $\mathbf{a}= (a_0, a_1,\ldots,a_{n-1})$  and $\mathbf{c}= (c_0, c_1,\ldots,c_{n-1})$, defined as $\mathbf{a}\cdot\mathbf{c}= \sum_{i=0}^{n-1} a_ic_i. $
\end{definition}
\begin{definition}
	Let $\mathbf{u}$ and $\mathbf{v}$ be two vectors in $\mathbb{F}_q^n$. Then the Hamming weight of $\mathbf{u}=(u_0, u_1,\ldots, u_{n-1})$ is defined as
	\begin{equation*}
		wt_H(\mathbf{u}) =  | \{i\mid u_i\ne 0\} |.
	\end{equation*} The Hamming distance between $\mathbf{u}=(u_0, u_1,\ldots, u_{n-1})$ and $\mathbf{v}=(v_0, v_1,\ldots, v_{n-1})$ is defined as 
	\begin{equation*}
		d_H(\mathbf{u}, \mathbf{v} )= |\{i\mid u_i \ne v_i\}|.
	\end{equation*} For a code $\mathfrak{C}$ of length $n$ over $\mathbb{F}_q$, the minimum Hamming distance of  $\mathfrak{C}$ is defined as 
	\begin{equation*}
		d_H( \mathfrak{C})= \min \{d_H(\mathbf{u}, \mathbf{v}) \mid \mathbf{u}, \mathbf{v} \in  \mathfrak{C},\ \mathbf{u} \ne \mathbf{v}  \}.
	\end{equation*}
	If $\mathfrak{C}$ is a linear code of length $n$ over $\mathbb{F}_q$  with dimension $k$ and minimum Hamming distance $d$, then $\mathfrak{C}$ is said to have parameters $[n, k, d]_q$, and for such $\mathfrak{C}$, $\mathfrak{C}^\perp$ is a linear code  over $\mathbb{F}_q$ with parameters $[n, n-k, d^\perp]_q$ for some  $d^\perp \geq1$.
\end{definition}
\begin{definition}\cite{lidl1997finite}
	Let $f(x)$ be a non-zero polynomial over $\mathbb{F}_q$ with constant term non-zero. Then order of $f(x)$ is defined as
	\begin{equation*}
		Ord (f(x))= \min \{ m \in \mathbb{N}\mid\ f(x) | x^m-1 \}.
	\end{equation*}
\end{definition}
\begin{definition}\cite{dinh2007complete}
	For $f(x)\in \mathbb{F}_q^n$, the coefficient weight of $f(x)= f_0+f_1x+\ldots+f_{n-1}x^{n-1}$ is defined as
	\begin{equation*}
		cw (f(x))= \begin{cases}
			0, & \text{ if } wt_H(f(x)) \text{ is } 0 \text{ or } 1,\\
			\min\{ |i-j|;\ f_i \ne 0,\ f_j \ne 0, i\ne j \}, & \text{ otherwise.}
		\end{cases}
	\end{equation*}
\end{definition}
\begin{definition}\cite{lopez2009dual}
	Let $\mathfrak{C}$ be a linear code over $\mathbb{F}_q$ with length $n$ and $\mathbf{a} = (a_0, a_1,\ldots, a_{n-1}) \in \mathbb{F}_q^n.$ Then $\mathfrak{C}$ is said to be a right (respectively, left) polycyclic code induced by $\mathbf{a}$ if for every $\mathbf{c}= (c_0, c_1,\ldots, c_{n-1})\in \mathfrak{C}$, $(0,c_0,c_1,\ldots,c_{n-2})+c_{n-1} (a_0, a_1,\ldots,a_{n-1}) \in \mathfrak{C}$ $\left(\text{respectively, } (c_1, c_2,\ldots,c_{n-1},0)+ c_0(a_0, a_1,\ldots,a_{n-1}) \in  \mathfrak{C} \right)$. The vector $\mathbf{a}$ is called an associate vector of $\mathfrak{C}$. Under the natural isomorphism $\sigma$ from $\mathbb{F}_q^n$ to the set of polynomials over $\mathbb{F}_q$ with degree less than $n$, each vector $\mathbf{v}=(v_0, v_1,\ldots, v_{n-1})\in \mathbb{F}_q^n$ can be identified as $v(x)=v_0+v_1x+\ldots+v_{n-1}x^{n-1}\in \mathbb{F}_q[x]$.
	
	Let $f(x)=x^n-a(x)$ and $f'(x)= x^n-a'(x)$, where $a(x)=a_0+a_1x+\ldots+a_{n-1}x^{n-1}$ and $a'(x)= a_{n-1}+a_{n-2}x+\ldots+a_1x^{n-2}+a_0x^{n-1}$. Then  right polycyclic codes of length $n$ induced by $\mathbf{a}$ correspond to the ideals of the ring $\frac{\mathbb{F}_q[x]}{\langle f(x) \rangle}$, and left polycyclic codes of length $n$ induced by $\mathbf{a}$ correspond to the ideals of the ring $\frac{\mathbb{F}_q[x]}{\langle f'(x) \rangle}$.
\end{definition}
\begin{definition}\cite{hou2009rational}
	A code $\mathfrak{C}$ of length $n$ over the finite field $\mathbb{F}_q$ is called right sequential if, for each $(c_0, c_1,\ldots, c_{n-1})\in \mathfrak{C}$, there exists $\alpha\in \mathbb{F}_q$ such that $(c_1, c_2,\ldots, c_{n-1}, \alpha) \in \mathfrak{C}.$ Similarly, $\mathfrak{C}$ is said to be left sequential if, for each $(c_0, c_1,\ldots, c_{n-1})\in \mathfrak{C}$, there exists $\beta\in \mathbb{F}_q$ such that $(\beta, c_0, c_1,\ldots,c_{n-2}) \in \mathfrak{C}. $
\end{definition}
In this work, the term ``polycyclic codes'' refers to right polycyclic codes and ``sequential codes'' refers to right sequential codes. The Euclidean dual of a polycyclic code is not necessarily polycyclic; however, it is always a sequential code. In fact, a linear code is sequential if and only if its dual is polycyclic (see \cite{lopez2009dual}). For convenience, throughout this paper, we use the terms ``weight'' and ``distance'' to refer to Hamming weight and Hamming distance, respectively, and we denote them as $wt(\cdot)$ and $d(\cdot)$ in place of $wt_H(\cdot)$ and $d_H(\cdot)$.
\section{The Hamming distance of binary polycyclic codes}\label{sec3}
Let $P(x)$ be a binary irreducible polynomial of degree $m$, where $m\geq 2.$ In this section, we study the Hamming distance of binary polycyclic codes associated with the polynomial $ P(x)^\mathcal{L}$, where $\mathcal{L} \geq 2 $ is a positive integer.   These codes are precisely the ideals of the quotient ring
\begin{equation*}
	\mathcal{P}= \frac{\mathbb{F}_2[x]}{\left\langle  P(x)^\mathcal{L} \right\rangle}.
\end{equation*}
Note that for each $\mathcal{L}\geq 2$, there exists a unique positive integer $\mathcal{T}$ such that $2^{\mathcal{T}-1} <\mathcal{L} \leq 2^{\mathcal{T}}$.
We retain these notations throughout the paper. We now describe the structure of the ideals of the ring  $\mathcal{P}$. 
\begin{proposition}\label{proposition1}
	The ring	$\mathcal{P}$ is a chain ring, and its ideals are of the form $I_j=\langle  P(x)^j \rangle, $ where $0\leq j \leq \mathcal{L}.$
\end{proposition}
\begin{proof}
	Every element of $\mathcal{P}$ can be uniquely expressed as
	\begin{equation*}
		b(x)=\sum_{i=0}^{m\mathcal{L}-1} b_ix^{i}, \quad \text{where } b_i \in \mathbb{F}_2,\ 0\leq i\leq m\mathcal{L}-1.
	\end{equation*}
	Let $I$ be a non-zero ideal of  $\mathcal{P}$, and let $\mathbf{i}(x) \in I$ be a non-zero  element  in $I$ of the least degree. We first show that such an element is unique. Suppose there exist   two  distinct non-zero polynomials $\mathbf{i}_1(x),\ \mathbf{i}_2(x) \in I$ with the same minimal degree. Then $\mathbf{i}_1(x)-\mathbf{i}_2(x) \in I$, and since the coefficients of $\mathbf{i}_1(x)$ and $\mathbf{i}_2(x)$ lie in $\mathbb{F}_2$, we have
	$
	\deg\left(\mathbf{i}_1(x)-\mathbf{i}_2(x)\right) < \deg\left(\mathbf{i}_1(x)\right)= \deg \left(\mathbf{i}_2(x)\right).
	$ This contradicts the minimality of the degree. Hence, the minimal-degree polynomial  in $I$ must be unique, which we denote  by $\mathbf{i}(x)$. Now, let $a(x)\in I$. By the division algorithm, we can write
	\begin{equation*}
		a(x)= a_Q(x) \mathbf{i}(x)+ a_R(x), \quad \text{where }  a_Q(x), a_R(x) \in \mathbb{F}_2[x],\ \deg \left( a_R(x)\right) < \deg \left( \mathbf{i}(x)\right).
	\end{equation*}
	Since $\mathbf{i}(x)\in I,$ we have $a_Q(x) \mathbf{i}(x)\in I$, and hence $ a_R(x) \in I.$ Since $ \deg \left( a_R(x)\right) < \deg \left( \mathbf{i}(x)\right),$  thus by the minimality of $\deg(\mathbf{i}(x))$, we have $a_R(x)=0$. Therefore, $a(x)\in \langle \mathbf{i}(x) \rangle, $ and hence $I = \langle \mathbf{i}(x) \rangle$.

	Next, consider the zero element of $\mathcal{P}$, namely $P(x)^\mathcal{L}$. Applying the division algorithm again, we write
	\begin{equation*}
		P(x)^\mathcal{L}= f(x) \mathbf{i}(x)+ g(x),\quad \text{where } f(x),\, g(x) \in \mathbb{F}_2[x],\ \deg \left(g(x)\right) < \deg  \left( \mathbf{i}(x)\right).
	\end{equation*}
	Since  $P(x)^\mathcal{L}$ is the zero element of $\mathcal{P}$, we have $	P(x)^\mathcal{L}- f(x) \mathbf{i}(x)=g(x) \in I.$ But the minimality of  $\deg   \left( \mathbf{i}(x)\right)$ gives us $g(x)=0$. Hence, $\mathbf{i}(x) $ divides $P(x)^\mathcal{L}$. Because $P(x)$ is an irreducible polynomial,   all the divisors of $P(x)^\mathcal{L}$ are of the form $P(x)^j$, where $0\leq j\leq \mathcal{L}$. Thus, $\mathbf{i}(x)=  P(x)^j$ for some  $0\leq j \leq \mathcal{L}$.  Therefore, every ideal of  $\mathcal{P}$ is of the form  $I_j=\langle  P(x)^j \rangle$, where $0\leq j \leq \mathcal{L}$. 	Moreover, these  ideals form the chain
	\begin{equation*}
		\{0\}=	I_{\mathcal{L}} \subset I_{\mathcal{L}-1} \subset \ldots \subset I_1 \subset I_0 = \mathcal{P}, 
	\end{equation*}
	where $I_1= \langle P(x) \rangle$ is the maximal ideal of $\mathcal{P}$.
	Therefore, $\mathcal{P}$ is a chain ring.
\end{proof}
The following result describes the size of corresponding family of binary polycyclic codes.
\begin{theorem}\label{theorem1}
	The binary polycyclic codes associated with the polynomial $ P(x)^\mathcal{L}$ are precisely the ideals $C_j= \langle P(x)^j \rangle\subseteq \mathcal{P}$, where $0\leq j \leq\mathcal{L}.$ Moreover, the dimension of $C_j$ is $m\left(\mathcal{L}-j\right)$, and its size is $2^{m\left(\mathcal{L}-j\right)}$.
\end{theorem}
\begin{proof}
	The first statement follows directly from the definition of polycyclic codes together with Proposition~\ref{proposition1}. Any element  $c(x)\in C_j$ can be written as
	\begin{equation*}
		c(x)= P(x)^j a(x),\quad \text{where } a(x)\in \mathbb{F}_2[x],\ \deg \left(a(x)\right)<m \mathcal{L}-mj.
	\end{equation*} Thus, clearly $\dim (C_j)=m\left(\mathcal{L}-j\right)$, and hence  $|C_j|=2^{m\left(\mathcal{L}-j\right)}$.
\end{proof}
We now determine the Hamming distance  of the codes $C_j$, denoted by $d_j$. Since $C_0 = \mathcal{P}$ and $C_{\mathcal{L}} = {0}$, we have $d_0 = 1$ and $d_{\mathcal{L}} = m \mathcal{L}$. We first consider the range $1\leq j \leq 2^{\mathcal{T}-1}$. The next two theorems provide general bounds on $d_j$ depending on the order of $P(x)$. In addition, Theorem~\ref{theorem4} gives exact values in certain cases. In particular, when $P(x)$ is a trinomial, we completely determine $d_j$ for $1\leq j\leq2^{\mathcal{T}-1}$ in Corollary~\ref{corollary1}.
\begin{theorem}\label{theorem2}
	If $ Ord\left(P(x)\right) \geq m \mathcal{L}$, then
	\begin{equation*}
		3 \leq d_1 \leq d_2 \leq \ldots \leq d_{2^{\mathcal{T}-1}} \leq wt\left(P(x)\right).
	\end{equation*}
\end{theorem}
\begin{proof} Since $P(x)$ is an irreducible polynomial of degree $m$, it has a nonzero constant term and can be written as
	\begin{equation*}
		P(x)= 1+p_1x+\ldots+p_{m-1}x^{m-1}+x^m, \quad \text{where } p_i \in \mathbb{F}_2,\ 1\leq  i\leq m-1.
	\end{equation*}
	Clearly, $P(x)^{\mathcal{L}}$ has a non-zero constant term and  degree $m\mathcal{L}$. Since $P(x)^{\mathcal{L}}=0 $ in $\mathcal{P}$, it follows that $x$ is  invertible  in  $\mathcal{P}$.
	
	Consider the code  $C_1= \langle P(x) \rangle.$ Since $C_1 \ne \mathcal{P},$ it  contains no unit, and thus $ d_1 \geq 2.$  If $d_1=2$, then there  exists  $a(x)\in \mathbb{F}_2[x]$ such that
	\begin{equation*}
		P(x) a(x)= x^s \left(1+x^t\right), \quad \text{where } s \geq0,\ t>0, \ s+t < m \mathcal{L}.
	\end{equation*}
	Since  $\gcd \left(x, P(x)  \right)=1$, this implies that $ P(x) \mid  \left(1+x^t\right)$ in $\mathbb{F}_2[x]$. This gives $Ord\left(P(x)\right) \leq t < m \mathcal{L}$, which contradicts the assumption that $Ord\left(P(x)\right) \geq m \mathcal{L}$. Hence, $d_1 \geq 3.$
	
	Now, we  consider $C_{2^{\mathcal{T}-1}}= \langle P(x)^{2^{\mathcal{T}-1}} \rangle$. Since $\mathcal{P}$ is a finite commutative ring of characteristic $2$, thus
	\begin{align*}
		P(x)^{2^{\mathcal{T}-1}}&= \left(1+p_1x+\ldots+p_{m-1}x^{m-1}+x^m\right)^{2^{\mathcal{T}-1}}= 1+p_1 x^{2^{\mathcal{T}-1}}+ p_2 x^{2 \cdot 2^{\mathcal{T}-1}}+\ldots+ p_{m-1} x^{(m-1) 2^{\mathcal{T}-1}}+ x^{m  2^{\mathcal{T}-1}}.
	\end{align*} 
	Thus, the Hamming weight of $P(x)^{2^{\mathcal{T}-1}}$ equals  $wt(P(x))$. Since $P(x)^{2^{\mathcal{T}-1}} \in C_{2^{\mathcal{T}-1}}$, it follows that
	$
	d_{2^{\mathcal{T}-1}} \leq  wt \left(P(x)\right).
	$
	Now $C_{2^{\mathcal{T}-1}} \subset C_{2^{\mathcal{T}-1}-1} \subset \ldots \subset C_2 \subset C_1,$ we obtain
	$
	3 \leq d_1 \leq d_2 \leq \ldots \leq d_{2^{\mathcal{T}-1}} \leq wt\left(P(x)\right).
	$ This completes the proof.
\end{proof}
\begin{theorem}\label{theorem3}
	Let $Ord \left(P(x)\right) = e < m\mathcal{L}$, and let $\mathcal{J}$ be the smallest positive integer such that $e  2^{\mathcal{T}-\mathcal{J}} < m \mathcal{L}$. Then 
	\begin{equation*}
		d_j = 2 \quad \text{for } 1\leq j \leq 2^{\mathcal{T}-\mathcal{J}},
	\end{equation*} and 
	\begin{equation*}
		3\leq d_j \leq wt\left(P(x)\right) \quad \text{for } 2^{\mathcal{T}-\mathcal{J}}+1 \leq j \leq 2^{\mathcal{T}-1}.
	\end{equation*}
\end{theorem}
\begin{proof}
	Since $C_1$ is  a non-trivial ideal of $\mathcal{P}$, as in Theorem~\ref{theorem2}, we have $d_1 \geq 2.$ Given that $Ord \left(P(x)\right) = e < m\mathcal{L}$, there exists a polynomial $U(x) \in \mathbb{F}_2[x]$ such that
	\begin{equation*}
		P(x) U(x) = x^{e}+1.
	\end{equation*}
	Since $\deg \left(P(x)\right)=m$,  it follows that $ e  \mid 2^m-1$ (see \cite{lidl1997finite}), and  hence $e$ must be odd. Consequently, $x^{e}+1$ has no repeated factors over $\mathbb{F}_2$. This implies that $U(x)$ is not a multiple of $P(x)$ and therefore, $U(x)$ is a unit in $\mathcal{P}$.
	
	We first prove that $d_j=2$ for $1\leq j \leq 2^{\mathcal{T}-\mathcal{J}}.$ For such $j$,
	\begin{equation*}
		P(x)^j \Big( P(x)^{2^{\mathcal{T}-\mathcal{J}}-j} U(x)^{2^{\mathcal{T}-\mathcal{J}}} \Big)= P(x)^{2^{\mathcal{T}-\mathcal{J}}} U(x)^{2^{\mathcal{T}-\mathcal{J}}} = \left(x^{e}+1\right)^{2^{\mathcal{T}-\mathcal{J}}}= x^{e  2^{\mathcal{T}-\mathcal{J}}}+1 \in \mathcal{P}.
	\end{equation*}
	Hence, $d_j=2$ for  $1\leq j \leq 2^{\mathcal{T}-\mathcal{J}}$. This proves the first part.
	
	We now prove the second part. If $\mathcal{J}=1,$ then $ d_j=2$ for all $1\leq j \leq 2^{\mathcal{T}-1}$. Now, assume $\mathcal{J} \geq 2$ and consider the code $C_{2^{\mathcal{T}-\mathcal{J}}+1}$. Let $c(x)$ be a non-zero codeword of $C_{2^{\mathcal{T}-\mathcal{J}}+1}$. Then it can be expressed as
	$
	c(x)= P(x)^{2^{\mathcal{T}-\mathcal{J}}} a(x), 
	$ where $a(x)$ is a nilpotent element in $\mathcal{P}$ with $\deg (a(x)) < m \mathcal{L}-m  2^{\mathcal{T}-\mathcal{J}}. $ Applying the division algorithm, we  write 
	\begin{equation*}
		a(x)= U(x)^{2^{\mathcal{T}-\mathcal{J}}} a_Q(x)+ a_R(x), \quad \text{where } \deg \left(a_Q(x)\right) < m\mathcal{L}-e 2^{\mathcal{T}-\mathcal{J}},\  \deg \left(a_R(x)\right) < (e-m)2^{\mathcal{T}-\mathcal{J}}.
	\end{equation*}
	We claim that $ \deg \left(a_Q(x)\right) < e 2^{\mathcal{T}-\mathcal{J}}.$	If $ \deg \left(a_Q(x)\right) \geq e 2^{\mathcal{T}-\mathcal{J}}$, then $ e2^{\mathcal{T}-\mathcal{J}} <  m\mathcal{L}-e 2^{\mathcal{T}-\mathcal{J}}, $ which gives $ e  2^{\mathcal{T}-(\mathcal{J}-1)} < m \mathcal{L}$. This contradicts the minimality of $\mathcal{J}$. Therefore, $ \deg \left(a_Q(x)\right) < e 2^{\mathcal{T}-\mathcal{J}}.$
	Now, we consider the following three cases:\\
	\textbf{Case I:} If $a_R(x)=0$ and $a_Q(x) \ne 0,$ then $a(x)= U(x)^{2^{\mathcal{T}-\mathcal{J}}} a_Q(x),$ and 
	\begin{equation*}
		c(x)= P(x)^{2^{\mathcal{T}-\mathcal{J}}}  U(x)^{2^{\mathcal{T}-\mathcal{J}}} a_Q(x)= \big( x^{e 2^{\mathcal{T}-\mathcal{J}}}+1\big) a_Q(x).
	\end{equation*}
	Here, $a(x)$ is a nilpotent  and $U(x)$ is a unit element in $\mathcal{P}$. Thus, $a_Q(x)$ must be a nilpotent element in $\mathcal{P}$.  Hence, $wt\left(a_Q(x)\right) \geq 2.$ Also, $cw \big( x^{e  2^{\mathcal{T}-\mathcal{J}}}+1\big)= e  2^{\mathcal{T}-\mathcal{J}} > \deg \left(a_Q(x)\right).$ Therefore, no cancellation occurs, and
	\begin{equation*}
		wt(c(x)) = wt \big( x^{e  2^{\mathcal{T}-\mathcal{J}}}+1\big) \cdot wt\left(a_Q(x)\right) = 2   wt\left(a_Q(x)\right) \geq 4.
	\end{equation*}
	\textbf{Case II:} If $a_Q(x)=0$ and $a_R(x) \ne 0$, then $a(x)= a_R(x)$, and 
	$
	c(x)=  P(x)^{2^{\mathcal{T}-\mathcal{J}}} a_R(x).
	$
	Suppose $wt(c(x))=2,$ then 
	\begin{equation*}
		P(x)^{2^{\mathcal{T}-\mathcal{J}}} a_R(x)= x^s \left(1+x^t\right), \quad \text{where } s \geq 0,\ t >0,\ s+t< e 2^{\mathcal{T}-\mathcal{J}}.
	\end{equation*}
	Since $Ord \big( P(x)^{2^{\mathcal{T}-\mathcal{J}}}  \big)= e  2^{\mathcal{T}-\mathcal{J}}$ (see \cite{lidl1997finite}), therefore  $	P(x)^{2^{\mathcal{T}-\mathcal{J}}} \nmid x^s \left(1+x^t\right)$. Hence, $wt(c(x)) \ne 2.$ Thus, in this case $wt(c(x)) \geq 3.$\\
	\textbf{Case III:} If $a_Q(x)\ne 0$ and $a_R(x) \ne 0$, then $a(x)=U(x)^{2^{\mathcal{T}-\mathcal{J}}} a_Q(x)+ a_R(x) $, and 
	\allowdisplaybreaks
	\begin{align*}
		c(x)&= P(x)^{2^{\mathcal{T}-\mathcal{J}}}  \Big( U(x)^{2^{\mathcal{T}-\mathcal{J}}} a_Q(x)+ a_R(x)  \Big)=  \big( x^{e  2^{\mathcal{T}-\mathcal{J}}}+1\big) a_Q(x) + P(x)^{2^{\mathcal{T}-\mathcal{J}}} a_R(x)\\
		&= \Big(a_Q(x)+ P(x)^{2^{\mathcal{T}-\mathcal{J}}} a_R(x)\Big) + x^{e 2^{\mathcal{T}-\mathcal{J}}} a_Q(x).
	\end{align*}
	From the previous arguments, we have $\deg \big(a_Q(x)+ P(x)^{2^{\mathcal{T}-\mathcal{J}}} a_R(x)\big) < e  2^{\mathcal{T}-\mathcal{J}} $, whereas  degree of every  term of $\big(x^{e  2^{\mathcal{T}-\mathcal{J}}} a_Q(x)\big)$ is at least $ e  2^{\mathcal{T}-\mathcal{J}}.$
	Hence, the Hamming weight of $c(x)$ satisfies
	\begin{equation*}
		wt(c(x))= wt \left(a_Q(x)+ P(x)^{2^{\mathcal{T}-\mathcal{J}}} a_R(x)\right)+ wt  \left(a_Q(x)\right).
	\end{equation*} We now consider three subcases based on  $wt  \left(a_Q(x)\right).$
	
	\textbf{Subcase I:} If $wt\left(a_Q(x)\right) \geq 3$,  then clearly $wt(c(x)) \geq 3.$
	
	\textbf{Subcase II:} If $wt\left(a_Q(x)\right) =2$, then by the same argument as in Case II, we have $ wt \big( P(x)^{2^{\mathcal{T}-\mathcal{J}}} a_R(x)\big) \geq 3.$ Hence, 
	$
	wt \big(a_Q(x)+ P(x)^{2^{\mathcal{T}-\mathcal{J}}} a_R(x)\big) \geq 1.
	$
	Therefore, $wt(c(x)) \geq 1+2=3.$
	
	\textbf{Subcase III:} If $wt\left(a_Q(x)\right) =1$, then again as in Case II, we have $ wt \big( P(x)^{2^{\mathcal{T}-\mathcal{J}}} a_R(x)\big) \geq 3.$ Thus, 
	$
	wt \big(a_Q(x)+ P(x)^{2^{\mathcal{T}-\mathcal{J}}} a_R(x)\big) \geq 2,
	$
	which gives $wt(c(x)) \geq 2+1=3.$
	
	Combining all the cases, we conclude that $d_{2^{\mathcal{T}-\mathcal{J}}+1} \geq 3.$ Furthermore, as in Theorem~\ref{theorem2}, we have $d_{2^{\mathcal{T}-1}} \leq wt\left(P(x)\right).$ This completes the proof.
\end{proof}
Theorems~\ref{theorem2} and \ref{theorem3} yield the following corollary.
\begin{corollary}\label{corollary1}
	Let $P(x)$ be a binary irreducible trinomial. Then 
	\begin{enumerate}
		\item If $Ord\left(P(x)\right) \geq m \mathcal{L}$, then \begin{equation*}
			d_1=d_2=\ldots=d_{2^{\mathcal{T}-1}}=3.
		\end{equation*}
		\item If $Ord \left(P(x)\right) = e < m\mathcal{L}$ and  $\mathcal{J}$ is the smallest positive integer such that $e  2^{\mathcal{T}-\mathcal{J}} < m \mathcal{L}$, then  \begin{equation*}
			d_j= \begin{cases}
				2 & \text{for } 1 \leq j \leq 2^{\mathcal{T}-\mathcal{J}},\\
				3 &\text{for } 2^{\mathcal{T}-\mathcal{J}}+1 \leq j \leq 2^{\mathcal{T}-1}.
			\end{cases}
		\end{equation*}
	\end{enumerate}
\end{corollary}
From Corollary~\ref{corollary1}, it is obvious that when $wt\left(P(x)\right)=3$, the exact value of $d_j$ for $1 \leq j \leq 2^{\mathcal{T}-1}$ is completely determined.
Some illustrative values of the Hamming distance for binary irreducible trinomials $P(x)$ and selected values of $\mathcal{L}$ are listed in Table~\ref{tab1}.
\begin{table}[h]
	\caption{The Hamming distance of $C_j$ for $1\leq j \leq 2^{\mathcal{T}-1}$}
	\label{tab1}
	\centering
	\begin{tabular}{|c|c|c|c|c|c|}
		\hline
		$P(x) $ & $ \mathcal{L} $ & $\mathcal{T}$& $Ord \left(P(x)\right)$ & $C_j$ & $d_j$  \\ \hline
		$x^2+x+1$ & $5$ & $3$ & $3$ & $\langle \left(x^2+x+1\right)^j \rangle $ &
		$	d_j= \begin{cases}
			2 & \text{for } 1 \leq j \leq 2,\\
			3 &\text{for } 3 \leq j \leq 4.
		\end{cases}$
		\\\hline
		$x^3+x+1$ & $2$ & $1$ & $7$ & $ \langle \left(x^3+x+1\right)^j \rangle $ &$d_1= 3$.\\ \hline
		$x^4+x+1$ & $7$ & $3$ & $15$ & $ \langle \left(x^4+x+1\right)^j \rangle $ &$	d_j= \begin{cases}
			2 & \text{for }  j = 1,\\
			3 &\text{for } 2 \leq j \leq 4.
		\end{cases}$\\ \hline
		$x^5+x^3+1$ & $6$ & $3$ & $31$ & $ \langle \left(x^5+x^3+1\right)^j \rangle $ &  $d_j=3$ for $1\leq j \leq 4$.\\ \hline
		$x^6+x+1$ & $9$ & $4$ & $63$ & $ \langle \left(x^6+x+1\right)^j \rangle $ & $d_j=3$ for $1\leq j \leq 8$. \\\hline
		$x^7+x^4+1$ & $31$ & $5$ & $127$ & $ \langle \left(x^7+x^4+1\right)^j \rangle $ & $	d_j= \begin{cases}
			2 & \text{for }  j = 1,\\
			3 &\text{for } 2 \leq j \leq 16.
		\end{cases}$\\ \hline
	\end{tabular}
\end{table}
We now turn  to the computation of $d_j$ for the case $j=2^{\mathcal{T}-s}$, where $1 \leq s \leq \mathcal{T}$ and $wt\left(P(x)\right) \geq 4$.
\begin{theorem}\label{theorem4}
	For $1\leq s \leq \mathcal{T}$,	let $\lambda_s$ be the positive integer satisfying $\left(\lambda_s-1\right) 2^{\mathcal{T}-s} < m \left(\mathcal{L}-2^{\mathcal{T}-s} \right) \leq \lambda_s 2^{\mathcal{T}-s}$. Then the Hamming distance of $C_{2^{\mathcal{T}-s} }$ is given by
	\begin{equation*}
		d_{2^{\mathcal{T}-s}}= \begin{cases}
			wt\left( P(x)^{2^{\mathcal{T}-s}} \right), & \text{if } \lambda_s=1,\\
			\min \left\{wt (f(x)) \mid f(x) \in S_{\mathcal{T},s}^{\lambda_s}\right\},
			&\text{if } \lambda_s>1,
		\end{cases}
	\end{equation*}
	where the set $S_{\mathcal{T},s}^{\lambda_s}$ is defined as
	\begin{equation*}
		S_{\mathcal{T},s}^{\lambda_s}= \left\{ a(x)^{2^{\mathcal{T}-s}} P(x)^{2^{\mathcal{T}-s}}\ \big|\ a(x)\in \mathbb{F}_2[x],\ \deg (a(x)) \leq \lambda_s-1,\ a(x) \text{ has constant term  }  1  \right\}.
	\end{equation*}

\end{theorem}
\begin{proof}
	Let $c(x)$ be a  non-zero codeword of $C_{2^{\mathcal{T}-s} }$. Then it can be written as 
	\begin{equation}\label{eq1}
		c(x)= P(x)^{2^{\mathcal{T}-s} } g(x), \quad \text{where } g(x)\in \mathbb{F}_2[x],\ \deg\left(g(x)\right) < m \left(\mathcal{L}-2^{\mathcal{T}-s}\right).
	\end{equation}
	Let $\lambda_s$ be the positive integer such that $\left(\lambda_s-1\right) 2^{\mathcal{T}-s} < m \left(\mathcal{L}-2^{\mathcal{T}-s} \right) \leq \lambda_s 2^{\mathcal{T}-s}$.
	
	If $\lambda_s=1,$ then $\deg (g(x)) < 2^{\mathcal{T}-s}$. Since  $cw\big(P(x)^{2^{\mathcal{T}-s}}\big) \geq 2^{\mathcal{T}-s}> \deg (g(x))$, it follows that no cancellation occurs in the product $P(x)^{2^{\mathcal{T}-s} } g(x)$, and hence 
	\begin{equation*}
		wt \left(c(x)\right)= wt\big(P(x)^{2^{\mathcal{T}-s} } \big) \cdot wt\left(g(x)\right) \geq  wt\big(P(x)^{2^{\mathcal{T}-s} } \big).
	\end{equation*}
	Moreover, since  $P(x)^{2^{\mathcal{T}-s} }  \in C_{2^{\mathcal{T}-s} }$, therefore
	$d_{2^{\mathcal{T}-s}}=	wt\big( P(x)^{2^{\mathcal{T}-s}} \big)$.
	
	Next, we assume that $\lambda_s>1.$ By applying the division algorithm,  the polynomial $g(x)$ can be written as
	\begin{equation*}
		g(x)= \left( 1+ x^{(\lambda_s-1) 2^{\mathcal{T}-s}}\right) g_1(x)+ g_1'(x),  
	\end{equation*} where
	$g_1(x),\ g_1'(x)\in \mathbb{F}_2[x]$ with  $   \deg \left(g_1(x)\right) < m \left(\mathcal{L}- 2^{\mathcal{T}-s}\right)- \left(\lambda_s-1\right)2^{\mathcal{T}-s}\leq 2^{\mathcal{T}-s}$ and $ \deg \left(g_1'(x)\right) < (\lambda_s-1) 2^{\mathcal{T}-s}.$	Applying the division algorithm again, we can write
	\begin{equation*}
		g_1'(x)= \left( 1+ x^{(\lambda_s-2) 2^{\mathcal{T}-s}}\right) g_2(x)+ g_2'(x), 
	\end{equation*} where 
	$ g_2(x),\ g_2'(x)\in \mathbb{F}_2[x]$  with  $\deg \left(g_2(x)\right) < 2^{\mathcal{T}-s}$ and $ \deg \left(g_2'(x)\right) < (\lambda_s-2) 2^{\mathcal{T}-s}.$ Repeating this process iteratively, the polynomial $g(x)$ can be expressed in the following form 
	\begin{equation}\label{eq2}
		g(x)= \left( 1+ x^{(\lambda_s-1) 2^{\mathcal{T}-s}}\right) g_1(x)+ \left( 1+ x^{(\lambda_s-2) 2^{\mathcal{T}-s}}\right) g_2(x)+ \ldots+ \left( 1+ x^{ 2^{\mathcal{T}-s}}\right) g_{\lambda_s-1}(x)+ g_{\lambda_s}(x),
	\end{equation}
	where each $ g_l(x) \in \mathbb{F}_2[x],\  \deg\left(g_1(x)\right)<  m \left(\mathcal{L}- 2^{\mathcal{T}-s}\right)- \left(\lambda_s-1\right)2^{\mathcal{T}-s}\leq 2^{\mathcal{T}-s}$ and $  \deg \left(g_l(x)\right) < 2^{\mathcal{T}-s}$ for $2\leq l \leq \lambda_s$. Furthermore, each  $g_l(x)$ can be expressed as follows:
	\begin{equation}\label{eq3}
		g_l(x)= \sum_{i=0}^{2^{\mathcal{T}-s}-1 } g_{l,i}x^{i}, \quad \text{where } g_{l,i} \in \mathbb{F}_2,\ 0\leq i \leq 2^{\mathcal{T}-s}-1.
	\end{equation}
	Substituting  the values of $g(x)$ and $g_l(x)$ from Eqs.~\eqref{eq2} and \eqref{eq3} into Eq.~\eqref{eq1}, we obtain
	\allowdisplaybreaks
	\begin{align}\label{eq4}
		\nonumber
		c(x)	=& \Big( \left(1+x^{\lambda_s-1}\right) P(x) \Big)^{2^{\mathcal{T}-s}} \sum_{i=0}^{2^{\mathcal{T}-s}-1 } g_{1,i}x^{i}+ \Big( \left(1+x^{\lambda_s-2}\right) P(x) \Big)^{2^{\mathcal{T}-s}} \sum_{i=0}^{2^{\mathcal{T}-s}-1 } g_{2,i}x^{i}+ \ldots\\& + \Big( \left(1+x\right) P(x) \Big)^{2^{\mathcal{T}-s}} \sum_{i=0}^{2^{\mathcal{T}-s}-1 } g_{\lambda_s-1,i}x^{i}+   P(x)^{2^{\mathcal{T}-s}} \sum_{i=0}^{2^{\mathcal{T}-s}-1 } g_{\lambda_s,i}x^{i}.
	\end{align}
	In each of the polynomials $ \big( \left(1+x^{\lambda_s-1}\right) P(x) \big)^{2^{\mathcal{T}-s}}, \big( \left(1+x^{\lambda_s-2}\right) P(x) \big)^{2^{\mathcal{T}-s}},\ldots, \big( \left(1+x\right) P(x) \big)^{2^{\mathcal{T}-s}}$, and $P(x)^{2^{\mathcal{T}-s}}$, all the exponents of $x$ are multiples of $2^{\mathcal{T}-s}$.
	Next, we group the terms of $c(x)$ according to their exponents modulo $2^{\mathcal{T}-s}$.  
	For each $0\leq i \leq 2^{\mathcal{T}-s}-1$,	let $c_i(x)$ denote the sum of all terms of $c(x)$ whose exponents are congruent to $i \bmod (2^{\mathcal{T}-s}).$ 
	Then 
	$
	c(x)= c_0(x)+c_1(x)+\ldots+c_{2^{\mathcal{T}-s}-1} (x).
	$
	Clearly, for $i_1 \ne i_2,$ no term in $c_{i_1}(x)$ coincides with a term in $c_{i_2}(x)$, since $ 2^{\mathcal{T}-s} b+i_1 \ne 2^{\mathcal{T}-s} b'+i_2$ for any $b, b' \geq0$ and $i_1 \ne i_2. $ 	Thus, 
	\begin{equation*}
		wt\left(c(x)\right)= wt\left(c_0(x)\right)+ wt\left(c_1(x)\right)+\ldots+ wt\left(c_{2^{\mathcal{T}-s}-1} (x)\right).
	\end{equation*}  Using Eq.~\eqref{eq4}, each $c_i(x)$ can be written as
	\begin{align}\label{eq5}
		\nonumber	c_i(x)= \Big( &  \left( \left(1+x^{\lambda_s-1}\right) P(x) \right)^{2^{\mathcal{T}-s}} g_{1,i}+ \left( \left(1+x^{\lambda_s-2}\right) P(x) \right)^{2^{\mathcal{T}-s}} g_{2,i}+\ldots + \left( \left(1+x\right) P(x) \right)^{2^{\mathcal{T}-s}} g_{\lambda_s-1,i}\\\nonumber &+\left(  P(x) \right)^{2^{\mathcal{T}-s}}g_{\lambda_s,i} \Big) x^{i}\\
		=& \Big( \left(1+x^{\lambda_s-1}\right)^{2^{\mathcal{T}-s}} g_{1,i}+\left(1+x^{\lambda_s-2}\right)^{2^{\mathcal{T}-s}} g_{2,i}+\ldots+ \left(1+x\right)^{2^{\mathcal{T}-s}} g_{\lambda_s-1,i}+ g_{\lambda_s,i} \Big) P(x)^{2^{\mathcal{T}-s}} x^{i}.
	\end{align}
	Since $c(x) \ne 0,$  we must have $c_j(x) \ne 0$ for some index  $j$.  Suppose that for this $j$,
	\begin{equation*}
		g_{m_1,j}= g_{m_2,j}= \ldots= g_{m_\mu, j}=1,\quad 1\leq m_1<m_2<\ldots<m_\mu \leq \lambda_s,
	\end{equation*}
	and all the remaining $g_{l,j}=0.$ Now, we consider the following cases:\\
	\textbf{Case I:} If $m_\mu <\lambda_s$, then  from Eq. \eqref{eq5}, we have 
	\begin{align*}
		c_j(x)&= \Big(  \left(1+x^{\lambda_s-m_1}\right)^{2^{\mathcal{T}-s}} +  \left(1+x^{\lambda_s-m_2}\right)^{2^{\mathcal{T}-s}}+\ldots+  \left(1+x^{\lambda_s-m_\mu}\right)^{2^{\mathcal{T}-s}} \Big)P(x)^{2^{\mathcal{T}-s}} x^{j}\\
		&= \left( 1+ x^{\lambda_s-m_1}+ 1+x^{\lambda_s-m_2}+\ldots+ 1+x^{\lambda_s-m_\mu}  \right)^{2^{\mathcal{T}-s}} P(x)^{2^{\mathcal{T}-s}} x^{j}.
	\end{align*}
	Now, we consider the following two subcases.
	
	\textbf{Subcase I:} If $\mu$ is odd, then we have
	\begin{align}\label{eq6}
		\nonumber	c_j(x)&= \left( 1+ x^{\lambda_s-m_1}+x^{\lambda_s-m_2}+\ldots+ x^{\lambda_s-m_\mu}  \right)^{2^{\mathcal{T}-s}} P(x)^{2^{\mathcal{T}-s}} x^{j},\\
		wt \left(c_j(x)\right) & = wt \Big(\left( 1+ x^{\lambda_s-m_1}+x^{\lambda_s-m_2}+\ldots+ x^{\lambda_s-m_\mu}  \right)^{2^{\mathcal{T}-s}} P(x)^{2^{\mathcal{T}-s}}\Big).
	\end{align}
	
	\textbf{Subcase II:} If $\mu$ is even, then we have 
	\begin{align}\label{eq7}
		\nonumber	c_j(x)&= \left(  x^{\lambda_s-m_1}+x^{\lambda_s-m_2}+\ldots+ x^{\lambda_s-m_\mu}  \right)^{2^{\mathcal{T}-s}} P(x)^{2^{\mathcal{T}-s}} x^{j}\\
		\nonumber	&=  \left(x^{\lambda_s-m_\mu}\right) ^{2^{\mathcal{T}-s}}\left( 1+ x^{m_\mu-m_{\mu-1}} + x^{m_\mu-m_{\mu-2}} +\ldots+  x^{m_\mu-m_1}  \right)^{2^{\mathcal{T}-s}} P(x)^{2^{\mathcal{T}-s}} x^{j},\\
		wt \left(c_j(x)\right) & = wt \left( \left( 1+ x^{m_\mu-m_{\mu-1}} + x^{m_\mu-m_{\mu-2}} +\ldots+  x^{m_\mu-m_1}  \right)^{2^{\mathcal{T}-s}} P(x)^{2^{\mathcal{T}-s}}  \right).
	\end{align}
	\textbf{Case II:} If $m_\mu =\lambda_s$, then we consider the following three subcases.
	
	\textbf{Subcase I:} If $\mu =1$, then from Eq. \eqref{eq5}, we have 
	\begin{align}\label{eq8}
		\nonumber	c_j(x)= & P(x)^{2^{\mathcal{T}-s}} x^{j},\\
		wt \left(c_j(x)\right) = & wt\left(P(x)^{2^{\mathcal{T}-s}}\right).
	\end{align}

	\textbf{Subcase II:} If $\mu$ is odd and $\mu \geq3,$ then from Eq.~\eqref{eq5}, we have
	\begin{align}\label{eq9}
		\nonumber	c_j(x)&= \Big(   \left(1+x^{\lambda_s-m_1}\right)^{2^{\mathcal{T}-s}} +  \left(1+x^{\lambda_s-m_2}\right)^{2^{\mathcal{T}-s}}+\ldots+  \left(1+x^{\lambda_s-m_{\mu-1}}\right)^{2^{\mathcal{T}-s}}+1 \Big)P(x)^{2^{\mathcal{T}-s}} x^{j}\\
		\nonumber	&= 	 \left( 1+ x^{\lambda_s-m_1} + x^{\lambda_s-m_2} +\ldots+  x^{\lambda_s-m_{\mu-1}}  \right)^{2^{\mathcal{T}-s}} P(x)^{2^{\mathcal{T}-s}} x^{j},\\
		wt \left(c_j(x)\right) &= wt \Big(\left( 1+ x^{\lambda_s-m_1} + x^{\lambda_s-m_2} +\ldots+  x^{\lambda_s-m_{\mu-1}}  \right)^{2^{\mathcal{T}-s}}  P(x)^{2^{\mathcal{T}-s}}\Big).
	\end{align}
	
	\textbf{Subcase III:} If $\mu$ is even, then we have 
	\begin{align}\label{eq10}
		\nonumber	c_j(x)&= \Big(   \left(1+x^{\lambda_s-m_1}\right)^{2^{\mathcal{T}-s}} +  \left(1+x^{\lambda_s-m_2}\right)^{2^{\mathcal{T}-s}}+\ldots+  \left(1+x^{\lambda_s-m_{\mu-1}}\right)^{2^{\mathcal{T}-s}}+1 \Big)P(x)^{2^{\mathcal{T}-s}} x^{j}\\
		\nonumber	&=	\left(  x^{\lambda_s-m_1}+x^{\lambda_s-m_2}+\ldots+ x^{\lambda_s-m_{\mu-1}}  \right)^{2^{\mathcal{T}-s}} P(x)^{2^{\mathcal{T}-s}} x^{j}\\
		\nonumber	&=  \left(x^{\lambda_s-m_{\mu-1}}\right) ^{2^{\mathcal{T}-s}} \left( 1+ x^{m_{\mu-1}-m_{\mu-2}} + x^{m_{\mu-1}-m_{\mu-3}} +\ldots+  x^{m_{\mu-1}-m_1}  \right)^{2^{\mathcal{T}-s}} P(x)^{2^{\mathcal{T}-s}} x^{j},\\
		wt \left(c_j(x)\right) &= wt \left( \left( 1+ x^{m_{\mu-1}-m_{\mu-2}} + x^{m_{\mu-1}-m_{\mu-3}} +\ldots+  x^{m_{\mu-1}-m_1}  \right)^{2^{\mathcal{T}-s}} P(x)^{2^{\mathcal{T}-s}}  \right).
	\end{align}
	From Eqs.~\eqref{eq6}--\eqref{eq10}, we conclude that in all the cases
	$
	wt\left(c_j(x)\right) \geq \min \left\{ wt \left(f(x)\right) \mid f(x) \in S_{\mathcal{T},s}^{\lambda_s} \right\}.
	$
	Therefore,
	\begin{equation}\label{eq11}
		wt\left(c(x)\right) \geq \min \left\{wt (f(x)) \mid f(x) \in S_{\mathcal{T},s}^{\lambda_s}\right\}.
	\end{equation}
	Since $ S_{\mathcal{T},s}^{\lambda_s} \subseteq  C_{2^{\mathcal{T}-s}}$, the reverse inequality  $ d_{2^{\mathcal{T}-s}} \leq \min \left\{wt (f(x)) \mid f(x) \in S_{\mathcal{T},s}^{\lambda_s}\right\}$ also holds. This completes the proof.
\end{proof}
\begin{remark}\label{remark1}
	For $1\leq s \leq \mathcal{T}$,	the size of the code $C_{2^{\mathcal{T}-s}}$  is $2^{m\left(L-2^{\mathcal{T}-s}\right) }$, while the cardinality of the set $S_{\mathcal{T},s}^{\lambda_s}$ is $2^{\lambda_s-1}$, where $\left(\lambda_s-1\right) 2^{\mathcal{T}-s} < m \left(\mathcal{L}-2^{\mathcal{T}-s} \right) \leq \lambda_s 2^{\mathcal{T}-s}$. Clearly, $ \lambda_s-1 < \frac{m \left(\mathcal{L}-2^{\mathcal{T}-s} \right) }{2^{\mathcal{T}-s}}$. Computing the Hamming distance of $C_{2^{\mathcal{T}-s}}$ directly would require examining the weights  of all  $2^{m\left(L-2^{\mathcal{T}-s}\right) }-1$ non-zero codewords and then identifying the minimum weight. Therefore, Theorem~\ref{theorem4} substantially reduces the computational complexity of determining $d_{2^{\mathcal{T}-s}}$, since it suffices to examine only the $2^{\lambda_s-1}$ elements of $S_{\mathcal{T},s}^{\lambda_s}$ instead of all nonzero codewords of the code.
\end{remark}
We illustrate the above theorem for $m=4,\ \mathcal{L}=2^{\mathcal{T}}-1$  and $ s=1$ in the following corollary.
\begin{corollary}\label{corollary2}
	Let $P(x)$ be a binary irreducible polynomial of degree $4$ and $\mathcal{L}=2^\mathcal{T}-1,\ \mathcal{T}\geq 4$. Then the Hamming distance of the code $C_{2^{\mathcal{T}-1}}$ is 
	\begin{equation*}
		d_{2^{\mathcal{T}-1}}=  \min \left\{wt (f(x)) \mid f(x) \in S_{\mathcal{T},1}^4\right\}, 
	\end{equation*}
	where
	\allowdisplaybreaks
	\begin{align*}
		S_{\mathcal{T},1}^4= \biggl\{ &  P(x)^{2^{\mathcal{T}-1}},\ \left(1+x\right)^{2^{\mathcal{T}-1}} P(x)^{2^{\mathcal{T}-1}},\ \left(1+x^{2 }\right)^{2^{\mathcal{T}-1}} P(x)^{2^{\mathcal{T}-1}},\  \left(1+x^{3 }\right)^{2^{\mathcal{T}-1}} P(x)^{2^{\mathcal{T}-1}},\   \left(1+x+x^{2 }\right)^{2^{\mathcal{T}-1}} \\ &  P(x)^{2^{\mathcal{T}-1}},\ \left(1+x+x^{3 }\right)^{2^{\mathcal{T}-1}}  P(x)^{2^{\mathcal{T}-1}},\ \left(1+x^{2 }+x^{3 }\right)^{2^{\mathcal{T}-1}}  P(x)^{2^{\mathcal{T}-1}},\ \left(1+x+x^{2 }+x^{3}\right)^{2^{\mathcal{T}-1}}  P(x)^{2^{\mathcal{T}-1}}\biggr\}.
	\end{align*}
\end{corollary}
\begin{proof}
	Let $c(x)$ be a non-zero element of $C_{2^{\mathcal{T}-1}}$. Then it can be expressed as 
	\begin{equation*}
		c(x)= P(x)^{2^{\mathcal{T}-1}} g(x), \quad \text{where } g(x)\in \mathbb{F}_2[x],\ \deg\left(g(x)\right)< 4 \left(2^\mathcal{T}-1-2^{\mathcal{T}-1}\right).
	\end{equation*}
	Let $\lambda_1$ be the positive integer such that
	\begin{equation*}
		\left(\lambda_1-1\right)2^{\mathcal{T}-1} < 4 \left( 2^{\mathcal{T}}-1-2^{\mathcal{T}-1}\right)= 4 \left(2^{\mathcal{T}-1}-1\right) \leq \lambda_1 2^{\mathcal{T}-1}. 
	\end{equation*} Then $\lambda_1=4.$ Using the division algorithm, $g(x)$ can be represented as 
	\begin{equation*}
		g(x)= \big(1+x^{3 \cdot 2^{\mathcal{T}-1}}\big) g_1(x)+g_1'(x),
	\end{equation*}
	where $g_1(x),\ g_1'(x)\in \mathbb{F}_2[x]$ with $\deg \left(g_1(x)\right) < 2^{\mathcal{T}-1}-4$ and $\deg \left(g_1'(x)\right) < 3 \cdot 2^{\mathcal{T}-1}$. By applying the division algorithm three times, the polynomial $g(x)$ can be expressed in the following form
	\begin{equation*}
		g(x)= \big(1+x^{3\cdot 2^{\mathcal{T}-1} }\big) g_1(x)+ \big(1+x^{2\cdot 2^{\mathcal{T}-1} }\big) g_2(x)+ \big(1+x^{ 2^{\mathcal{T}-1} }\big) g_3(x)+ g_4(x),
	\end{equation*}
	where $g_1(x), g_2(x), g_3(x), g_4(x)\in \mathbb{F}_2[x]$ and $\deg\left(g_1(x)\right) < 2^{\mathcal{T}-1}-4, \ \deg \left(g_l(x)\right)< 2^{\mathcal{T}-1} $ for  $2 \leq l\leq 4.$ Furthermore, each $g_l(x)$ can be expanded as
	\begin{equation*}
		g_l(x)=\sum_{i=0}^{2^{\mathcal{T}-1}-1} g_{l,i} x^{i}, \quad \text{where }  g_{l,i}\in \mathbb{F}_2,\ 0\leq i \leq 2^{\mathcal{T}-1}-1.
	\end{equation*} Substituting this into $c(x)$, we obtain
	\allowdisplaybreaks
	\begin{align}\label{eq12}
		\nonumber c(x)=& \big(1+x^{3\cdot 2^{\mathcal{T}-1} }\big) P(x)^{2^{\mathcal{T}-1}} \sum_{i=0}^{2^{\mathcal{T}-1}-1} g_{1,i} x^{i}+ \big(1+x^{2\cdot 2^{\mathcal{T}-1} }\big) P(x)^{2^{\mathcal{T}-1}} \sum_{i=0}^{2^{\mathcal{T}-1}-1} g_{2,i} x^{i}\\& + \big(1+x^{ 2^{\mathcal{T}-1} }\big) P(x)^{2^{\mathcal{T}-1}} \sum_{i=0}^{2^{\mathcal{T}-1}-1} g_{3,i} x^{i}+ P(x)^{2^{\mathcal{T}-1}} \sum_{i=0}^{2^{\mathcal{T}-1}-1} g_{4,i} x^{i}.
	\end{align}
	It is quite obvious that every term of $c(x)$ has exponent of $x$ of the form $b 2^{\mathcal{T}-1}+i$ with $b \geq 0$ and $0\leq i \leq 2^{\mathcal{T}-1}-1.$ Let $c_i(x)$ be the polynomial containing the terms of $c(x)$ whose exponent of $x$ is of the form $b 2^{\mathcal{T}-1}+i$. From Eq. \eqref{eq12}, we get
	\begin{equation*}
		c_i(x)= \left(  \big(1+x^{3\cdot 2^{\mathcal{T}-1} }\big)  g_{1,i}+ \big(1+x^{2\cdot 2^{\mathcal{T}-1} }\big) g_{2,i}+ \big(1+x^{ 2^{\mathcal{T}-1} }\big) g_{3,i}+g_{4,i} \right)P(x)^{2^{\mathcal{T}-1}} x^{i}
	\end{equation*}
	Since $ c(x)\ne 0,$ thus there exists some $j$ for which $c_j(x)\ne 0.$ For this $j$, we  list all the  possibilities of $\left(g_{1,j},\ g_{2,j},\ g_{3,j},\  g_{4,j}\right)$ in Table~\ref{tab2} and observe the weight of the resulting $c_j(x)$.
	\begin{table}[htbp]
		\caption{All possible choices of $c_j(x)$ and their corresponding weights  (Corollary~\ref{corollary2})}
		\label{tab2}
		\centering
		\begin{tabular}{|c|c|c|}
			\hline
			$\left(g_{1,j},\ g_{2,j},\ g_{3,j},\ g_{4,j}\right)$ & $c_j(x)$ & $wt(c_j(x))$ \\ \hline
			$(1,0,0,0)$ &  $\left(1+x^{3 }\right)^{2^{\mathcal{T}-1}}  P(x)^{2^{\mathcal{T}-1}} x^{i}$  & $wt\left(\left(1+x^{3}\right)^{2^{\mathcal{T}-1}} P(x)^{2^{\mathcal{T}-1}}\right)$
			\\ \hline
			$(0,1,0,0)$ & $\left(1+x^{2}\right)^{2^{\mathcal{T}-1} }  P(x)^{2^{\mathcal{T}-1}} x^{i}$ & $ wt\left( \left(1+x^{2}\right)^{2^{\mathcal{T}-1} }  P(x)^{2^{\mathcal{T}-1}} \right)$\\\hline
			$(0,0,1,0)$ & $\left(1+x\right)^{2^{\mathcal{T}-1} }  P(x)^{2^{\mathcal{T}-1}} x^{i}$	& $ wt\left( \left(1+x\right)^{2^{\mathcal{T}-1} }  P(x)^{2^{\mathcal{T}-1}} \right)$\\\hline
			$(0,0,0,1)$ & $P(x)^{2^{\mathcal{T}-1}} x^{i}$ & $wt\left(   P(x)^{2^{\mathcal{T}-1}} \right)$\\\hline
			$(1,1,0,0)$ & $ x^{2\cdot 2^{\mathcal{T}-1} } \left(1+x\right)^{2^{\mathcal{T}-1} } P(x)^{2^{\mathcal{T}-1}} x^{i}$ & $wt\left( \left(1+x\right)^{2^{\mathcal{T}-1} }  P(x)^{2^{\mathcal{T}-1}} \right)$\\\hline
			$(1,0,1,0)$ & $ x^{ 2^{\mathcal{T}-1} } \left(1+x^{2 }\right)^{2^{\mathcal{T}-1} }  P(x)^{2^{\mathcal{T}-1}} x^{i}$ & $ wt\left( \left(1+x^{ 2 }\right)^{2^{\mathcal{T}-1} }  P(x)^{2^{\mathcal{T}-1}} \right)$\\\hline
			$(1,0,0,1)$ & $x^{3\cdot 2^{\mathcal{T}-1} } P(x)^{2^{\mathcal{T}-1}} x^{i}$ & $wt\left(   P(x)^{2^{\mathcal{T}-1}} \right)$\\\hline
			$(0,1,1,0)$ & $x^{ 2^{\mathcal{T}-1} } \left(1+x\right)^{2^{\mathcal{T}-1} } P(x)^{2^{\mathcal{T}-1}} x^{i}$ & $ wt\left( \left(1+x\right)^{2^{\mathcal{T}-1} }  P(x)^{2^{\mathcal{T}-1}} \right)$\\\hline
			$(0,1,0,1)$ & $x^{2\cdot 2^{\mathcal{T}-1} } P(x)^{2^{\mathcal{T}-1}} x^{i}$ & $ wt\left(   P(x)^{2^{\mathcal{T}-1}} \right)$ \\\hline
			$(0,0,1,1)$ & $x^{2^{\mathcal{T}-1} } P(x)^{2^{\mathcal{T}-1}} x^{i}$ & $wt\left(   P(x)^{2^{\mathcal{T}-1}} \right)$\\\hline
			$(1,1,1,0)$ & $\left(1+ x+x^{2}+x^{3 } \right)^{2^{\mathcal{T}-1} }P(x)^{2^{\mathcal{T}-1}} x^{i}$ & $wt \left( \left(1+ x +x^2+x^{3 } \right)^{2^{\mathcal{T}-1} }P(x)^{2^{\mathcal{T}-1}} \right)$\\\hline
			$(1,1,0,1)$ & $\left(1+x^{2 }+x^{3 } \right)^{2^{\mathcal{T}-1} }P(x)^{2^{\mathcal{T}-1}} x^{i}$ & $wt \left( \left(1+ x^{2 }+x^{3} \right)^{2^{\mathcal{T}-1} }P(x)^{2^{\mathcal{T}-1}} \right)$\\\hline
			$(1,0,1,1)$ & $\left(1+x+x^{3 } \right)^{2^{\mathcal{T}-1} }P(x)^{2^{\mathcal{T}-1}} x^{i}$ & $wt \left( \left(1+ x+x^{3} \right)^{2^{\mathcal{T}-1} }P(x)^{2^{\mathcal{T}-1}} \right)$ \\\hline
			$(0,1,1,1)$ & $\left(1+x+x^{2 } \right)^{2^{\mathcal{T}-1} }P(x)^{2^{\mathcal{T}-1}} x^{i}$ & $wt \left( \left(1+ x+x^{2 } \right)^{2^{\mathcal{T}-1} }P(x)^{2^{\mathcal{T}-1}} \right)$\\\hline
			$(1,1,1,1)$ & $ x^{2^{\mathcal{T}-1} } \left(1+x+x^{2} \right)^{2^{\mathcal{T}-1} }	P(x)^{2^{\mathcal{T}-1}} x^{i}$ & $wt \left( \left(1+ x+x^{2 } \right)^{2^{\mathcal{T}-1} }P(x)^{2^{\mathcal{T}-r}} \right)$
			\\\hline
			
		\end{tabular}
	\end{table}
	Considering all the possibilities, we observe that 
	\begin{equation*}
		wt(c(x)) \geq wt\left(c_j(x)\right) \geq \min \left\{wt(f(x)) \mid f(x) \in S_{\mathcal{T},1}^4\right\}. 
	\end{equation*}
	On the other hand, since $ S_{\mathcal{T},1}^4 \subseteq  C_{2^{\mathcal{T}-1}}$, it follows that
	\begin{equation*}
		d_{2^{\mathcal{T}-1}} \leq \min \left\{wt (f(x)) \mid f(x) \in S_{\mathcal{T},1}^4\right\}.
	\end{equation*}
	Combining both the  inequalities, we conclude the result.
\end{proof}
In the following example, we illustrate the computation of the Hamming distance for  the specific case when $P(x)=x^5+ x^4+x^2+x+1$ and $\mathcal{L}=5$.
\begin{example}\label{ex1}
	The Hamming distance of binary polycyclic codes associated with $\left(x^5+ x^4+x^2+x+1\right)^{5}.$
	
	We have $P(x)=x^5+ x^4+x^2+x+1$ and  $\mathcal{L}=5$, thus $m=5$ and $ \mathcal{T}=3$. The order of $P(x)$ is $e= 31$, which satisfies $e\geq m \mathcal{L}= 25.$ Therefore, by Theorem~\ref{theorem2}, we have 
	\begin{equation}\label{eq13}
		3\leq d_1 \leq d_2 \leq d_3 \leq d_4 \leq 5.
	\end{equation}
	For $1\leq s \leq 3$, let $\lambda_s$ be the positive integer such that $\left( \lambda_s-1\right)2^{3-s} < 5 \left(5-2^{3-s}\right) \leq \lambda_s 2^{3-s}$. This gives 
	$
	\lambda_1=2$, $\lambda_2=8$ and $\lambda_3=20.
	$
	By Theorem~\ref{theorem4},
	$
	d_{2^{3-1}} = d_4= \min\left\{wt(f(x)) \mid f(x)\in S_{3,1}^2 \right\},
	$
	where $
	S_{3,1}^2= \left\{  P(x)^4,\ \left(1+x^4\right) P(x)^4\right\}.
	$
	We have 
	\begin{equation*}
		P(x)^4= x^{20}+x^{16}+x^8+x^4+1,\quad \left(1+x^4\right) P(x)^4= x^{24}+x^{16}+x^{12}+1.
	\end{equation*}
	Therefore,
	\begin{equation}\label{eq14}
		d_4=4.
	\end{equation}
	Again, by Theorem~\ref{theorem4}, 
	$
	d_{2^{3-2}}=d_2= \min\left\{wt(f(x)) \mid f(x)\in S_{3,2}^8 \right\},
	$ where
	\begin{equation*}
		S_{3,2}^{8}= \left\{ a(x)^{2} P(x)^{2} \mid a(x)\in \mathbb{F}_2[x],\ \deg (a(x)) \leq 7,\ a(x) \text{ has constant term  }  1  \right\}.
	\end{equation*}
	Now, $\left(1+x+x^2\right)^2P(x)^2= x^{14}+x^4+1\in 	S_{3,2}^{8}$. Thus, $d_2 \leq 3.$ Therefore, from Eqs.~\eqref{eq13} and \eqref{eq14} we conclude that
	\begin{equation*}
		d_1=d_2=3,\quad 3\leq d_3\leq 4,\quad d_4=4.
	\end{equation*}
\end{example}
We have  determined the Hamming distance $d_j$ for $1\leq j \leq 2^{\mathcal{T}-1}$. We now proceed to compute $d_j$ for $2^{\mathcal{T}-1} < j <\mathcal{L}$. We divide the analysis into three cases according to the value of $\mathcal{L}$.
\subsection{The Hamming distance of $C_j$  for $2^{\mathcal{T}-1} < j <\mathcal{L}$, when $\mathcal{L}= 2^{\mathcal{T}}$}\label{sec3A}
For each integer $j$ satisfying $2^{\mathcal{T}-1} < j < 2^{\mathcal{T}}$, there exists an interval of the form
$
2^{\mathcal{T}}-2^{\mathcal{T}-r}+1 \leq j \leq 2^{\mathcal{T}}-2^{\mathcal{T}-r-1},
$
where $1 \leq r \leq \mathcal{T}-1.$
As a first step, we compute the Hamming distance of the codes $C_{	2^{\mathcal{T}}-2^{\mathcal{T}-r}}$, for $1\leq r \leq \mathcal{T}$. 
\begin{theorem}\label{theorem5}
	For $1\leq r \leq \mathcal{T}$,	the Hamming distance of  $C_{2^{\mathcal{T}}-2^{\mathcal{T}-r}}$ is given by
	\begin{equation*}
		d_{2^{\mathcal{T}}-2^{\mathcal{T}-r}}=  \min \left\{wt (f(x)) \mid f(x) \in S_{\mathcal{T},r}^{'m}\right\},
	\end{equation*} where the set $S_{\mathcal{T},r}^{'m}$ is defined as
	
	\begin{equation*}
		S_{\mathcal{T},r}^{'m}= \left\{ a(x)^{2^{\mathcal{T}-r}} P(x)^{2^{\mathcal{T}}-2^{\mathcal{T}-r}}\ \Big|\ a(x)\in \mathbb{F}_2[x],\ \deg (a(x)) \leq m-1,\ a(x) \text{ has constant term  }  1  \right\}.
	\end{equation*}
	
	\end{theorem}
	\begin{proof}
		Let $c(x)$ be a  non-zero element of $ C_{2^{\mathcal{T}}-2^{\mathcal{T}-r}}$. Then it can be expressed as
		\begin{equation}\label{eq15}
			c(x)= P(x)^{2^{\mathcal{T}}-2^{\mathcal{T}-r}} g(x),\quad \text{where } g(x)\in \mathbb{F}_2[x],\ \deg\left(g(x)\right) < m 2^{\mathcal{T}-r}.
		\end{equation}
		By applying the division algorithm as in Theorem~\ref{theorem4},  the polynomial $g(x)$ can be written as
		\begin{equation}\label{eq16}
			g(x)= \big( 1+ x^{(m-1) 2^{\mathcal{T}-r}}\big) g_1(x)+ \big( 1+ x^{(m-2) 2^{\mathcal{T}-r}}\big) g_2(x)+ \ldots+ \big( 1+ x^{ 2^{\mathcal{T}-r}}\big) g_{m-1}(x)+ g_m(x),
		\end{equation}
		where each $ g_l(x) \in \mathbb{F}_2[x]$ with $\deg \left(g_l(x)\right) < 2^{\mathcal{T}-r}$. Furthermore, each  $g_l(x)$ can be expressed as follows:
		\begin{equation}\label{eq17}
			g_l(x)= \sum_{i=0}^{2^{\mathcal{T}-r}-1 } g_{l,i}x^{i}, \quad \text{where } g_{l,i} \in \mathbb{F}_2,\ 0\leq i \leq 2^{\mathcal{T}-r}-1.
		\end{equation}
		Substituting the values of $g(x)$ and $g_l(x)$ from Eqs.~\eqref{eq16} and \eqref{eq17} into Eq. \eqref{eq15}, we obtain
		\allowdisplaybreaks
		\begin{align*}
			\nonumber	c(x)&= \big( 1+ x^{(m-1) 2^{\mathcal{T}-r}}\big) P(x)^{2^{\mathcal{T}}-2^{\mathcal{T}-r}} \sum_{i=0}^{2^{\mathcal{T}-r}-1 } g_{1,i}x^{i}+ \big( 1+ x^{(m-2) 2^{\mathcal{T}-r}}\big) P(x)^{2^{\mathcal{T}}-2^{\mathcal{T}-r}} \sum_{i=0}^{2^{\mathcal{T}-r}-1 } g_{2,i}x^{i}+  \ldots\\\nonumber & \quad+ \big( 1+ x^{ 2^{\mathcal{T}-r}}\big) P(x)^{2^{\mathcal{T}}-2^{\mathcal{T}-r}} \sum_{i=0}^{2^{\mathcal{T}-r}-1 } g_{m-1,i}x^{i}+ P(x)^{2^{\mathcal{T}}-2^{\mathcal{T}-r}} \sum_{i=0}^{2^{\mathcal{T}-r}-1 } g_{m,i}x^{i}\\\nonumber
			&= \Big( \left(1+x^{m-1}\right) P(x)^{2^r-1} \Big)^{2^{\mathcal{T}-r}} \sum_{i=0}^{2^{\mathcal{T}-r}-1 } g_{1,i}x^{i}+ \Big( \left(1+x^{m-2}\right) P(x)^{2^r-1} \Big)^{2^{\mathcal{T}-r}} \sum_{i=0}^{2^{\mathcal{T}-r}-1 } g_{2,i}x^{i}+ \ldots\\& \quad+ \Big( \left(1+x\right) P(x)^{2^r-1} \Big)^{2^{\mathcal{T}-r}} \sum_{i=0}^{2^{\mathcal{T}-r}-1 } g_{m-1,i}x^{i}+ \Big(  P(x)^{2^r-1} \Big)^{2^{\mathcal{T}-r}} \sum_{i=0}^{2^{\mathcal{T}-r}-1 } g_{m,i}x^{i}.
		\end{align*}
		Following the similar argument as in Theorem~\ref{theorem4}, we conclude the result.
	\end{proof}
	Next, we give the bounds on  the Hamming distance of the codes $C_{2^\mathcal{T}-2^{\mathcal{T}-r}+1}$ for $1\leq r\leq \mathcal{T}-2$. 
	\begin{theorem}\label{theorem6}
		For $1\leq r \leq \mathcal{T}-2,$  the Hamming distance of $C_{2^\mathcal{T}-2^{\mathcal{T}-r}+1}$ satisfies 
		\begin{equation*}
			2 d_{2^\mathcal{T}-2^{\mathcal{T}-r}} \leq	d_{2^\mathcal{T}-2^{\mathcal{T}-r}+1} \leq d_{2^\mathcal{T}-2^{\mathcal{T}-r-1}}.
		\end{equation*}
	\end{theorem}
	\begin{proof}
		Let $c(x)$ be a non-zero codeword in $C_{2^\mathcal{T}-2^{\mathcal{T}-r}+1}$. Then 
		\begin{equation*}
			c(x)= P(x)^{2^\mathcal{T}-2^{\mathcal{T}-r}+1} f(x), \quad \text{where } f(x) \in \mathbb{F}_2[x],\ \deg \left(f(x)\right)  < m \left( 2^{\mathcal{T}-r}-1 \right). 
		\end{equation*}
		Equivalently, we may write
		\begin{equation*}
			c(x)= P(x)^{2^\mathcal{T}-2^{\mathcal{T}-r}} g(x),\quad \text{where } g(x)= P(x)f(x) \in \mathbb{F}_2[x],\ \deg \left(g(x)\right)  < m  2^{\mathcal{T}-r}. 
		\end{equation*}
		Since $g(x)= P(x)f(x)$, it is a nilpotent element of $\mathcal{P}$. Hence, $wt(g(x)) \geq 2.$ Suppose $wt(g(x)) = \omega.$ We group the terms of $g(x)$ according to their exponents modulo $2^{\mathcal{T}-r}$. Then  $g(x)$ can be represented in the following form
		\begin{equation}\label{eq18}
			g(x)= \sum_{i=1}^{s_1} x^{2^{\mathcal{T}-r}l_{1,i}+m_1} +  \sum_{i=1}^{s_2} x^{2^{\mathcal{T}-r}l_{2,i}+m_2}+\ldots+ \sum_{i=1}^{s_t} x^{2^{\mathcal{T}-r}l_{t,i}+m_t},
		\end{equation}
		where $t \geq 1,\ 0\leq m_1<m_2<\ldots, m_t \leq 2^{\mathcal{T}-r}-1,\ s_1, s_2,\ldots, s_t \geq 1$, and for each $1\leq \mu \leq t,$ we have $ 0\leq l_{\mu,1} < l_{\mu,2} < \ldots < l_{\mu,s_{\mu}} \leq m-1$.  Moreover, $s_1+s_2+\ldots+s_t= \omega$.
		
		We first show that $t\geq2.$	Suppose  $t=1$. Then 
		\begin{equation*}
			g(x)= \sum_{i=1}^{s_1} x^{2^{\mathcal{T}-r}l_{1,i}+m_1}=x^{m_1} \left( x^{l_{1,1}}+ x^{l_{1,2}}+\ldots+x^{l_{1,s_1}} \right)^{2^{\mathcal{T}-r}}, 
		\end{equation*} where $0\leq m_1 \leq 2^{\mathcal{T}-r}-1,\ s_1 \geq1,\ 0\leq l_{1,1} < l_{1,2} <\ldots< l_{1,s_1}\leq m-1.$
		Since $\deg \big( x^{l_{1,1}}+ x^{l_{1,2}}+\ldots+x^{l_{1,s_1}} \big) <m,$ and $\mathcal{P}$ is a chain ring with the maximal ideal $\langle P(x) \rangle$, thus $x^{l_{1,1}}+ x^{l_{1,2}}+\ldots+x^{l_{1,s_1}}$ must be a unit in $\mathcal{P}$. This implies that $g(x)$ is also a unit in $\mathcal{P}$,  contradicting the fact that $g(x)$ is nilpotent. Therefore, we must have $t \geq 2.$
		
		Using \eqref{eq18},	 we rewrite $c(x)$ as
		\allowdisplaybreaks
		\begin{align*}
			c(x)&= \sum_{\mu=1}^{t} \Big(x^{m_\mu} \left( x^{l_{\mu,1}}+ x^{l_{\mu,2}}+\ldots+x^{l_{\mu,s_\mu}} \right)^{2^{\mathcal{T}-r}} P(x)^{2^\mathcal{T}-2^{\mathcal{T}-r}}\Big)\\
			&=\sum_{\mu=1}^{t} \Big(x^{m_\mu} x^{l_{\mu,1}2^{\mathcal{T}-r}} \left( 1+ x^{l_{\mu,2}-l_{\mu,1}}+x^{l_{\mu,3}-l_{\mu,1}}+  \ldots+ x^{l_{\mu,s_\mu}-l_{\mu,1}}\right)^{2^{\mathcal{T}-r}} P(x)^{2^\mathcal{T}-2^{\mathcal{T}-r}} \Big). 
		\end{align*}
		Therefore, the Hamming weight of $c(x)$ is given by
		\allowdisplaybreaks
		\begin{align*}
			wt(c(x))&= \sum_{\mu=1}^{t} wt \Big( x^{m_\mu} x^{l_{\mu,1}2^{\mathcal{T}-r}} \left( 1+ x^{l_{\mu,2}-l_{\mu,1}}+x^{l_{\mu,3}-l_{\mu,1}}+  \ldots+ x^{l_{\mu,s_\mu}-l_{\mu,1}}\right)^{2^{\mathcal{T}-r}} P(x)^{2^\mathcal{T}-2^{\mathcal{T}-r}} \Big)\\
			&= \sum_{\mu=1}^{t} wt \Big( \left( 1+ x^{l_{\mu,2}-l_{\mu,1}}+x^{l_{\mu,3}-l_{\mu,1}}+  \ldots+ x^{l_{\mu,s_\mu}-l_{\mu,1}}\right)^{2^{\mathcal{T}-r}} P(x)^{2^\mathcal{T}-2^{\mathcal{T}-r}} \Big).
		\end{align*}
		By Theorem \ref{theorem5}, each polynomial $\left( 1+ x^{l_{\mu,2}-l_{\mu,1}}+x^{l_{\mu,3}-l_{\mu,1}}+  \ldots+ x^{l_{\mu,s_\mu}-l_{\mu,1}}\right)^{2^{\mathcal{T}-r}} P(x)^{2^\mathcal{T}-2^{\mathcal{T}-r}}$ belongs to $S_{\mathcal{T},r}^{'m}$. Therefore, for every $\mu$, 
		\begin{equation*}
			wt \Big( \left( 1+ x^{l_{\mu,2}-l_{\mu,1}}+x^{l_{\mu,3}-l_{\mu,1}}+  \ldots+ x^{l_{\mu,s_\mu}-l_{\mu,1}}\right)^{2^{\mathcal{T}-r}} P(x)^{2^\mathcal{T}-2^{\mathcal{T}-r}} \Big) \geq d_{2^\mathcal{T}-2^{\mathcal{T}-r} }.
		\end{equation*}
		Since $t \geq 2,$ we obtain $ wt(c(x)) \geq 2 d_{2^\mathcal{T}-2^{\mathcal{T}-r} }.$ This proves the lower bound. The upper bound follows clearly, since $C_{2^\mathcal{T}-2^{\mathcal{T}-r-1}} \subseteq C_{2^\mathcal{T}-2^{\mathcal{T}-r}+1}$.  This completes the proof.
	\end{proof}
	\begin{remark}\label{remark2}
		The bounds established in Theorem~\ref{theorem6} cannot be further improved for a general polynomial $P(x)$. This can be verified through the following example.
		
		Consider $P(x)=x^4+x+1$ and $\mathcal{L}=16$. In this case, $m=4$ and $\mathcal{T}=4$. For $r=2$, Theorem~\ref{theorem6} gives
		$
		2 d_{2^{4}-2^{4-2}} \leq d_{2^{4}-2^{4-2}+1} \leq d_{2^{4}-2^{4-3}}.
		$
		Equivalently, 
		$
		2 d_{12} \leq d_{13} \leq d_{14}.
		$
		Using Theorem~\ref{theorem5}, we obtain
		$
		d_{12}=8$ and $ d_{14}=16,
		$ which yields $d_{13}=16$. Therefore, in this case, both lower and upper bounds  are same, demonstrating that the bounds provided in Theorem~\ref{theorem6} are tight and cannot be further improved for general \( P(x) \). 
	\end{remark} 
	Since $C_{2^\mathcal{T}-2^{\mathcal{T}-r-1}} \subset C_{2^\mathcal{T}-2^{\mathcal{T}-r-1}-1} \subset \ldots \subset C_{2^\mathcal{T}-2^{\mathcal{T}-r}+2} \subset C_{2^\mathcal{T}-2^{\mathcal{T}-r}+1},\  1\leq r \leq \mathcal{T}-2$, thus Theorem~\ref{theorem6} immediately gives the following result.
	\begin{theorem}\label{theorem7}
		For $ 1\leq r \leq \mathcal{T}-2$ and $1 \leq i \leq 2^{\mathcal{T}-r-1}$, the Hamming distance of  $C_{2^\mathcal{T}-2^{\mathcal{T}-r}+i}$ satisfies
		\begin{equation*}
			2  d_{2^\mathcal{T}-2^{\mathcal{T}-r} }\leq	d_{2^\mathcal{T}-2^{\mathcal{T}-r}+i}\leq d_{2^\mathcal{T}-2^{\mathcal{T}-r-1} }.
		\end{equation*}
	\end{theorem}
	Using Theorems \ref{theorem2}-\ref{theorem7}, we can determine the Hamming distance of all the binary polycyclic codes associated with the polynomial $P(x)^\mathcal{L}$ for  $\mathcal{L}=2^\mathcal{T},\ \mathcal{T}\geq 1.$ In the following example, we illustrate the computation of the Hamming distance for  the specific case when $P(x)= x^4+x+1$ and $\mathcal{L}=16$.
	\begin{example}\label{ex2}
		The Hamming distance of binary polycyclic codes associated with $\left(x^4+x+1\right)^{16}.$
		
		We have $P(x)= x^4+x+1$ and  $\mathcal{L}=16$, thus $m=4$ and $ \mathcal{T}=4$. The order of $P(x)$ is $E= 15$, which satisfies $E< m \mathcal{L}= 64.$ The smallest positive integer $\mathcal{J}$ such that $ 15 \cdot  2^{4-\mathcal{J}} < 64$ is $\mathcal{J}=2$. Therefore, by Corollary~\ref{corollary1}, we have
		\begin{equation}\label{eq19}
			d_j= \begin{cases}
				2 & \text{for } 1 \leq j \leq 4,\\
				3 &\text{for } 5 \leq j \leq 8.
			\end{cases}
		\end{equation}
		By Theorem~\ref{theorem5}, for $2\leq r \leq 4$, the Hamming distance of the code $C_{16-2^{4-r}}$ is
		\begin{equation*}
			d_{16-2^{4-r}}= \min \left\{wt(f(x)) \mid f(x)\in S_{4,r}^{'4} \right\},
		\end{equation*}
		where $ 	S_{4,r}^{'4}= \Bigl\{  \big(P(x)^{2^r-1}\big)^{2^{4-r}}, \big((1+x) P(x)^{2^r-1}\big)^{2^{4-r}}$, $\big((1+x^2) P(x)^{2^r-1}\big)^{2^{4-r}}$, $\big((1+x^3) P(x)^{2^r-1}\big)^{2^{4-r}}, \big((1+x+x^2) P(x)^{2^r-1}\big)^{2^{4-r}}, \big((1+x+x^3) P(x)^{2^r-1}\big)^{2^{4-r}}, \big((1+x^2+x^3) P(x)^{2^r-1}\big)^{2^{4-r}}, \big((1+x+x^2+x^3) P(x)^{2^r-1}\big)^{2^{4-r}} \Bigr\}. $
		Hence, to compute $	d_{16-2^{4-r}}$, we need to determine the weights of the following polynomials.
		\begin{equation*}
			P(x)^{2^r-1},\  \big(1+x\big) P(x)^{2^r-1},\ \big(1+x^2\big) P(x)^{2^r-1},\ \big(1+x^3\big) P(x)^{2^r-1},\ \big(1+x+x^2\big) P(x)^{2^r-1}, 
		\end{equation*}
		\begin{equation*}
			\big(1+x+x^3\big) P(x)^{2^r-1}, \	\big(1+x^2+x^3\big) P(x)^{2^r-1},\ \big(1+x+x^2+x^3\big) P(x)^{2^r-1}.
		\end{equation*}
		These explicit polynomials and their weights are listed in Table~\ref{tab3} for
		$2\leq r \leq 4.$
		\begin{table}[htbp]
			\caption{Hamming weights of some polynomials for $P(x)= x^4+x+1$ (Example~\ref{ex2})}
			\label{tab3}
			\centering
			\begin{tabular}{|c|c|c|c|}
				\hline
				Polynomial and its weight & $ r=2$ & $r=3$& $r=4$ \\ \hline
			
				$wt \left( P(x)^{2^r-1}\right)$ & $9$ & $17$ & $33$\\\hline

				$wt \left(\left(1+x\right) P(x)^{2^r-1}\right) $ & $8$ & $18$ & $34$\\\hline

				$wt \left(\left(1+x^2\right) P(x)^{2^r-1}\right)$ & $8$ & $16$ & $34$\\\hline

				$wt\left(\left(1+x^3\right) P(x)^{2^r-1}\right) $ & $8$ & $18$ & $34$\\\hline

				$wt \left(\left(1+x+x^2\right) P(x)^{2^r-1}\right)$ & $9$ & $17$ & $35$\\\hline

				$wt \left(\left(1+x+x^3\right) P(x)^{2^r-1}\right)$ & $9$ & $17$ & $35$\\\hline

				$wt \left( \left(1+x^2+x^3\right) P(x)^{2^r-1}  \right)$ & $9$ & $17$ & $35$\\\hline

				$ wt\left(\left(1+x+x^2+x^3\right) P(x)^{2^r-1}\right)$ & $8$ & $16$ & $36$\\\hline
			\end{tabular}
		\end{table}
		
		From Table~\ref{tab3}, we deduce that
		\begin{equation}\label{eq20}
			d_{12}=8,\quad d_{14}=16,\quad d_{15}=33.
		\end{equation}
		Applying Theorem~\ref{theorem7} for $1\leq r \leq 2$, we have
		\begin{equation*}
			2  d_{16-2^{4-r}} \leq d_{16-2^{4-r}+i} \leq d_{16-2^{4-r-1}},\quad 1\leq i \leq 2^{4-r-1}.
		\end{equation*}
		For $r=1$, 
		$
		2  d_8 \leq d_{8+i} \leq d_{12}, \text{ where } 1 \leq i\leq 4.
		$
		Since $d_8=3$ and $d_{12}=8$, we obtain
		\begin{equation}\label{eq21}
			6 \leq d_9 \leq d_{10} \leq d_{11} \leq 8.
		\end{equation}
		For $r=2$, we have
		$
		2  d_{12} \leq d_{12+i} \leq d_{14}, \text{ where }  1\leq i \leq 2.
		$
		Since $d_{12}=8$ and $d_{14}=16,$ thus
		\begin{equation}\label{eq22}
			d_{13}=16.
		\end{equation}
		Collecting all  results~\eqref{eq19}--\eqref{eq22}, the Hamming distances of all binary polycyclic codes associated with $\left(x^4+x+1\right)^{16}$ are 
		\begin{equation*}
			d_0=1, \quad d_1=d_2=d_3=d_4=2,  \quad d_5=d_6=d_7=d_8=3,
		\end{equation*}
		\begin{equation*}
			6\leq d_9\leq d_{10} \leq d_{11} \leq 8,\quad d_{12}=8, \quad d_{13}=d_{14}=16, \quad d_{15}=33, \quad d_{16}=64.
		\end{equation*}
	\end{example}
	Next, we consider the case  $2^{\mathcal{T}-1} < \mathcal{L} < 2^\mathcal{T}$. Then there exists an integer $R$, with $1\leq R \leq \mathcal{T}-1$, such that 
	$
	2^\mathcal{T} -2^{\mathcal{T}-R}+1 \leq \mathcal{L} \leq  2^\mathcal{T} -2^{\mathcal{T}-R-1}.
	$
	We divide the discussion into two subcases. 
	First, when 
	$ 2^{\mathcal{T}-1}+1 \leq \mathcal{L} \leq 2^\mathcal{T}-2^{\mathcal{T}-2}$, and  second, when $ 2^\mathcal{T} -2^{\mathcal{T}-R}+1 \leq \mathcal{L} \leq  2^\mathcal{T} -2^{\mathcal{T}-R-1}$ for   $2\leq R \leq \mathcal{T}-1$.
	
	\subsection{The Hamming distance of $C_j$ for $2^{\mathcal{T}-1} < j < \mathcal{L}$, when $2^{\mathcal{T}-1}+1 < \mathcal{L} \leq 2^\mathcal{T}-2^{\mathcal{T}-2}$}\label{sec3B}
	If $\mathcal{L}=2^{\mathcal{T}-1}+1$, then there exists no $j$ such that $2^{\mathcal{T}-1} < j < \mathcal{L}$. Hence, assume  $2^{\mathcal{T}-1}+1 < \mathcal{L} \leq 2^\mathcal{T}-2^{\mathcal{T}-2}$. Write $\mathcal{L}= 2^{\mathcal{T}-1}+\mathcal{L}'$, where  $1< \mathcal{L}' \leq 2^{\mathcal{T}-2}$. We consider $C_j$ for $2^{\mathcal{T}-1}< j <2^{\mathcal{T}-1}+\mathcal{L}'$.  The Hamming distance of $ C_{2^{\mathcal{T}-1}}$ is obtained from Theorem~\ref{theorem4}.
	Proceeding  as in Theorem~\ref{theorem6}, we obtain
	$
	d_{2^{\mathcal{T}-1}+1} \geq 2  d_{2^{\mathcal{T}-1}}.
	$ Since  $C_{2^{\mathcal{T}-1}+i} \subseteq C_{2^{\mathcal{T}-1}+1}$
	for $ 1 \leq i < \mathcal{L}'$,  we immediately get the following result.
	\begin{theorem}\label{theorem8}
		For $ 1 \leq i < \mathcal{L}',$ the Hamming distance of $C_{2^{\mathcal{T}-1}+i}$ satisfies
		\begin{equation*}
			d_{2^{\mathcal{T}-1}+i} \geq 2  d_{2^{\mathcal{T}-1}}.
		\end{equation*}	
	\end{theorem}
	Using Theorems~\ref{theorem2},~\ref{theorem3},~\ref{theorem4}, and~\ref{theorem8}, we can determine the Hamming distance for all binary polycyclic codes associated with  $P(x)^\mathcal{L}$, when $2^{\mathcal{T}-1}+1 \leq \mathcal{L} \leq 2^\mathcal{T}-2^{\mathcal{T}-2}.$ The following example illustrates the computation.
	
	\begin{example}\label{ex3}
		The Hamming distance of binary polycyclic codes associated with $\left(x^5+x^4+x^2+x+1\right)^{12}.$
		
		Here, $P(x)= x^5+x^4+x^2+x+1,\ \mathcal{L}=12=8+4$, thus $m=5,\ \mathcal{T}=4$, and $\mathcal{L}'=4$. The order of $P(x)$ is $E= 31$, which satisfies $E< m \mathcal{L}= 60.$ The smallest positive integer $\mathcal{J}$ such that $ 31 \cdot  2^{4-\mathcal{J}} < 60$ is $\mathcal{J}=4$. By Theorem~\ref{theorem3},
		\begin{equation}\label{eq23}
			d_1=2,
		\end{equation}
		and 
		\begin{equation}\label{eq24}
			3\leq d_j \leq 5 \quad \text{ for } 2\leq j\leq 8.
		\end{equation}
		The positive integer $\lambda_1$ such that $8 \left(\lambda_1-1\right)  < 20 \leq 8\lambda_1$ is $\lambda_1=3$. By Theorem~\ref{theorem4}, 
		\begin{equation*}
			d_8= \min \left\{wt \left(  a(x)^8 P(x)^8 \right)\ \big|\ a(x)\in \mathbb{F}_2[x],\ \deg (a(x)) \leq 2,\ a(x) \text{ has constant term } 1\right\}.
		\end{equation*}
		Equivalently,
		\begin{align*}
			d_8= \min \Bigl\{ &  wt\left(P(x)^8\right),\ wt\left( \left(1+x^8\right) P(x)^8\right),\ wt\left(\left(1+x^{16}\right)P(x)^8\right),\ wt\left(\left(1+x^8+x^{16}\right)P(x)^8\right) \Bigr\}.
		\end{align*}
		For $P(x)= x^5+x^4+x^2+x+1$, we have 
		\begin{equation*}
			P(x)^8=x^{40} + x^{32} + x^{16} + x^8 + 1,\quad
			\left(1+x^8\right)P(x)^8 = x^{48} + x^{32} + x^{24} + 1,
		\end{equation*}
		\begin{equation*}
			\left(1+x^{16}\right)P(x)^8 = x^{56}+x^{48}+x^{40}+x^{24}+x^8+1,\quad
			\left(1+x^8+x^{16}\right)P(x)^8= x^{56} + x^{16} + 1.
		\end{equation*}
		Thus, $d_8=3.$ Therefore, Eq.~\eqref{eq24} reduces to
		\begin{equation}\label{eq25}
			d_j=3 \quad \text{ for } 2\leq j \leq 8.
		\end{equation}
		Now,	by Theorem~\ref{theorem8}, we have
		$
		d_{8+i} \geq 2 d_8$,  where  $1\leq i <4.
		$
		Consequently, 
		\begin{equation}\label{eq26}
			6\leq 	d_9 \leq d_{10} \leq d_{11}. 
		\end{equation}
		From Eqs.~\eqref{eq23}, \eqref{eq25} and \eqref{eq26}, we conclude that
		\begin{equation*}
			d_0=1,\quad d_1=2,\quad d_2=d_3=\ldots=d_8=3,\quad 6\leq 	d_9 \leq d_{10} \leq d_{11},\quad d_{12}=60.
		\end{equation*}
	\end{example}
	\subsection{The Hamming distance of $C_j$ for $2^{\mathcal{T}-1} < j < \mathcal{L}$, when $2^\mathcal{T}-2^{\mathcal{T}-R}+1 \leq \mathcal{L} \leq 2^\mathcal{T}-2^{\mathcal{T}-R-1}$ for some $2\leq R\leq \mathcal{T}-1$}\label{sec3C}
	Let $2^\mathcal{T}-2^{\mathcal{T}-R}+1 \leq \mathcal{L} \leq 2^\mathcal{T}-2^{\mathcal{T}-R-1}$, where $2\leq R\leq \mathcal{T}-1$. Then we may write $\mathcal{L}=  2^{\mathcal{T}} -2^{\mathcal{T}-R}+\mathcal{L}'$, where $1\leq \mathcal{L}' \leq 2^{\mathcal{T}-R-1}$. First, we compute the Hamming distance of $C_{2^{\mathcal{T}}-2^{\mathcal{T}-r} }$ for $1\leq r \leq R.$ 
	\begin{theorem}\label{theorem9}
		For $1\leq r \leq R$,  the Hamming distance of $ C_{2^{\mathcal{T}}-2^{\mathcal{T}-r}}$ is given by 
		
		\begin{equation*}
			d_{2^{\mathcal{T}}-2^{\mathcal{T}-r}}=
			\begin{cases}
				wt\left( P(x)^{2^{\mathcal{T}}-2^{\mathcal{T}-r}} \right), 
				& \text{if } \lambda_r'=1,\\[6pt]
				\begin{aligned}
					\min \Big\{ & wt \left( a(x)^{2^{\mathcal{T}-r}} P(x)^{2^{\mathcal{T}}-2^{\mathcal{T}-r}} \right) \ \Big|\ 
					a(x)\in \mathbb{F}_2[x],\\
					& \deg(a(x)) \le \lambda_r'-1,\ a(x) \text{ has constant term } 1
					\Big\},
				\end{aligned}
				& \text{if } \lambda_r'>1,
			\end{cases}
		\end{equation*}
		where $\lambda_r'$ is the positive integer satisfying 
		$
		\left(\lambda_r'-1\right)2^{\mathcal{T}-r} < m 2^{\mathcal{T}-r} - m \left(2^{\mathcal{T}-R}-\mathcal{L}'\right) \leq \lambda_r' 2^{\mathcal{T}-r}.
		$
	\end{theorem}
	\begin{proof}
		Let $c(x)$ be a non-zero element of $ C_{2^{\mathcal{T}}-2^{\mathcal{T}-r}}$. Then 
		\begin{equation*}
			c(x)= P(x)^{2^{\mathcal{T}}-2^{\mathcal{T}-r}} g(x),  
		\end{equation*} where $g(x)\in \mathbb{F}_2[x]$ and $ \deg\left(g(x)\right) < m \mathcal{L}- m \left(2^{\mathcal{T}}-2^{\mathcal{T}-r}\right)= m 2^{\mathcal{T}-r}- m \left(2^{\mathcal{T}-R}-\mathcal{L}'\right).$
		Let $\lambda_r'$ be the positive integer satisfying 
		$
		\left(\lambda_r'-1\right)2^{\mathcal{T}-r} < m 2^{\mathcal{T}-r} - m \left(2^{\mathcal{T}-R}-\mathcal{L}'\right) \leq \lambda_r' 2^{\mathcal{T}-r}.
		$
		Now, we consider the following two cases:\\
		\textbf{Case I:} If $\lambda_r'=1,$ then $c(x)=  P(x)^{2^{\mathcal{T}}-2^{\mathcal{T}-r}} g(x)$, where $\deg(g(x)) < 2^{\mathcal{T}-r}$. Since $ cw \left(P(x)^{2^{\mathcal{T}}-2^{\mathcal{T}-r}}\right) \geq 2^{\mathcal{T}-r}$, thus
		\begin{equation*}
			wt \left(c(x)\right)= wt \left(P(x)^{2^{\mathcal{T}}-2^{\mathcal{T}-r}}\right) \cdot wt(g(x)) \geq wt\left(P(x)^{2^{\mathcal{T}}-2^{\mathcal{T}-r}}\right).
		\end{equation*}
		Moreover, since $P(x)^{2^{\mathcal{T}}-2^{\mathcal{T}-r}}  \in C_{2^{\mathcal{T}}-2^{\mathcal{T}-r}} $, thus $ d_{2^{\mathcal{T}}-2^{\mathcal{T}-r}} = wt\left( P(x)^{2^{\mathcal{T}}-2^{\mathcal{T}-r}}  \right)$. \\
		\textbf{Case II:} If $1< \lambda_r' \leq m,$ then $ m 2^{\mathcal{T}-r} - m \left(2^{\mathcal{T}-R}-\mathcal{L}'\right) =\left(\lambda_r'-1\right)2^{\mathcal{T}-r}+ \alpha_r$, where $1\leq \alpha_r \leq 2^{\mathcal{T}-r}$. Using the division algorithm, we can express $g(x)$ as
		\begin{equation*}
			g(x)= \left(1+x^{\left(\lambda_r'-1\right)2^{\mathcal{T}-r}}\right) g_1(x)+    \left(1+x^{\left(\lambda_r'-2\right)2^{\mathcal{T}-r}}\right) g_2(x)+\ldots+  \left(1+x^{2^{\mathcal{T}-r}}\right) g_{\lambda_r'-1}(x)+g_{\lambda_r'}(x),
		\end{equation*} where each $g_l(x) \in \mathbb{F}_2[x],\  \deg \left(g_1(x)\right) < \alpha_r$ and $\deg \left(g_l(x)\right) < 2^{\mathcal{T}-r}$ for $2\leq l\leq
		\lambda_r'.$ Following the same reasoning as in Theorem~\ref{theorem4}, we obtain
		\allowdisplaybreaks
		\begin{equation*}
			d_{2^{\mathcal{T}}-2^{\mathcal{T}-r}} =	\min \Big\{ wt \left( a(x)^{2^{\mathcal{T}-r}} P(x)^{2^{\mathcal{T}}-2^{\mathcal{T}-r}} \right) \ \Big|\ 
			a(x)\in \mathbb{F}_2[x],\
			\deg(a(x)) \le \lambda_r'-1,\ a(x) \text{ has constant term } 1
			\Big\}.
		\end{equation*}
		This completes the proof.
	\end{proof}
	Proceeding as in Theorem~\ref{theorem6}, we get the following result.
	\begin{theorem}\label{theorem10}
		For $ 1\leq r \leq R-1,$ the Hamming distance of  $C_{2^\mathcal{T}-2^{\mathcal{T}-r}+i}$ satisfies
		\begin{equation*}
			2  d_{2^\mathcal{T}-2^{\mathcal{T}-r} }\leq	d_{2^\mathcal{T}-2^{\mathcal{T}-r}+i}\leq d_{2^\mathcal{T}-2^{\mathcal{T}-r-1} }, \quad  1 \leq i \leq 2^{\mathcal{T}-r-1},
		\end{equation*}
		and 
		\begin{equation*}
			d_{2^\mathcal{T}-2^{\mathcal{T}-R} +i} \geq 2  d_{2^\mathcal{T}-2^{\mathcal{T}-R}}, \quad  1\leq i < \mathcal{L}'.
		\end{equation*}
	\end{theorem}
	Using Theorems~\ref{theorem2},~\ref{theorem3},~\ref{theorem4},~\ref{theorem9}, and~\ref{theorem10}, we can determine the Hamming distance of all binary polycyclic codes associated with the polynomial $P(x)^\mathcal{L}$ when $2^\mathcal{T}-2^{\mathcal{T}-R}+1 \leq \mathcal{L} \leq 2^\mathcal{T}-2^{\mathcal{T}-R-1}$ for some $2\leq R\leq \mathcal{T}-1$. To illustrate, we compute the Hamming distance of binary polycyclic codes associated with $\left(x^6+x^5+x^3+x^2+1\right)^{25}$.
	\begin{example}\label{ex4}
		The Hamming distance of binary polycyclic codes associated with $\left(x^6+x^5+x^3+x^2+1\right)^{25}$.
		
		Here, $P(x)=x^6+x^5+x^3+x^2+1$ and $ \mathcal{L}=25$. Since $2^4< \mathcal{L}< 2^5$, we have $m=6$ and $\mathcal{T}=5.$ Moreover, $ \mathcal{L}=2^5-2^3+1,$ thus $R=2$ and $\mathcal{L}'=1.$ The order of $P(x)$ is $E= 63 < m \mathcal{L}= 150.$ The smallest  positive integer $\mathcal{J}$ such that $63 \cdot 2^{5-\mathcal{J}} < 150$ is $\mathcal{J}=4.$ Hence, by Theorem~\ref{theorem3},
		\begin{equation}\label{eq27}
			d_j=2 \quad \text{ for } 1\leq j \leq 2,
		\end{equation}
		\begin{equation}\label{eq28}
			3\leq d_j \leq 5\quad \text{ for } 3 \leq j \leq 16.
		\end{equation}
		For $1\leq r\leq 2,$ let $\lambda_r'$ be the positive integer such that
		\begin{equation*}
			\left(\lambda_r'-1\right)2^{5-r} < 6\cdot 2^{5-r}-6\left(2^3-1\right)= 6\cdot 2^{5-r}- 42 \leq \lambda_r' 2^{5-r}.
		\end{equation*}
		This gives
		$
		\lambda_1'=4\text{ and } \lambda_2'=1.
		$
		By Theorem~\ref{theorem9}, 
		\begin{equation*}
			d_{2^5-2^4}= d_{16}= \min  \left\{wt \left(a(x)^{16} P(x)^{16}\right)\ \big|\ a(x)\in \mathbb{F}_2[x],\ \deg(a(x)) \leq 3,\ a(x) \text{ has constant term } 1 \right\}.
		\end{equation*}
		Explicitly,
		\begin{align*}
			d_{16} =\min \Bigl\{ & wt\left(P(x)^{16}\right),\  wt \left( \left(1+x^{16}\right)P(x)^{16} \right),\  wt \left( \left(1+x^{32}\right)P(x)^{16} \right),\  wt \left( \left(1+x^{48}\right)P(x)^{16} \right),\\&  wt \left( \left(1+x^{16}+x^{32}\right)P(x)^{16} \right),\  wt \left( \left(1+x^{16}+x^{48}\right)P(x)^{16} \right), \ wt \left( \left(1+x^{32}+x^{48}\right)P(x)^{16} \right),\\&  wt \left( \left(1+x^{16}+x^{32}+x^{48}\right)P(x)^{16} \right) \Bigr\}.
		\end{align*}
		These polynomials and their weights are listed in Table~\ref{tab4}.
		
		\begin{table}[htbp]
			\caption{Binary Expansion and Hamming weights of some polynomials  for $P(x)=x^6+x^5+x^3+x^2+1$ (Example~\ref{ex4})}
			\label{tab4}
			\centering
			\begin{tabular}{|c|c|c|}
				\hline
				Polynomial &  Binary Expansion & Weight of the Polynomial \\\hline
				$P(x)^{16}$ &  $x^{96} + x^{80 }+ x^{48} + x^{32} + 1$ & $5$\\\hline
				$ \left(1+x^{16}\right) P(x)^{16}$ & $x^{112 }+ x^{80} + x^{64} + x^{32} + x^{16} + 1$ & $6$ \\\hline
				$ \left(1+x^{32}\right) P(x)^{16}$ & $x^{128} + x^{112} + x^{96} + x^{64} + x^{48} + 1$ & $6$ \\\hline
				$ \left(1+x^{48}\right) P(x)^{16}$ & $x^{144 }+ x^{128 }+ x^{32} + 1$ & $4$\\\hline
				$ \left(1+x^{16}+x^{32}\right) P(x)^{16}$ & $x^{128 }+ x^{16} + 1$ & $3$\\\hline
				$ \left(1+x^{16}+x^{48}\right) P(x)^{16}$ & $x^{144 }+ x^{128 }+ x^{112} + x^{96 }+ x^{64} + x^{48} + x^{32} + x^{16} + 1$ & $9$\\\hline
				$\left(1+x^{32}+x^{48}\right) P(x)^{16} $ & $x^{144 }+ x^{112} + x^{80} + x^{64} + 1$ & $5$ \\\hline
				$\left(1+x^{16}+x^{32}+x^{48}\right) P(x)^{16} $ & $x^{144} + x^{96} + x^{80} + x^{48} + x^{16} + 1
				$ & $6$\\\hline		
			\end{tabular}
		\end{table}
		From Table~\ref{tab4}, we conclude that $d_{16}=3.$ Thus, Eq.~\eqref{eq28} reduces to
		\begin{equation}\label{eq29}
			d_j=3 \quad \text{ for }\ 3\leq j\leq 16.
		\end{equation}
		Next, by Theorem~\ref{theorem9}, we have
		$
		d_{32-8}=d_{24}= wt\left(P(x)^{24}\right).
		$
		Since
		$
		P(x)^{24}=	x^{144} + x^{136 }+ x^{128} + x^{112} + x^{104} + x^{96} + x^{88} + x^{64} + x^{56 }+ x^{48} + x^{40}
		+ x^{32} + x^{24} + x^{16} + 1,
		$
		thus
		\begin{equation}\label{eq30}
			d_{24}=15.
		\end{equation}
		Applying Theorem~\ref{theorem10}, we get
		$
		2  d_{16} \leq d_{16+i} \leq d_{24}$,  where  $1\leq i \leq 8.
		$
		This implies that
		\begin{equation}\label{eq31}
			6 \leq d_{17}\leq d_{18}\leq \ldots \leq d_{23}\leq 15.
		\end{equation}
		Combining Eqs.~\eqref{eq27}, \eqref{eq29}, \eqref{eq30}, and \eqref{eq31}, we obtain
		\begin{equation*}
			d_0=1,\quad d_1=d_2=2,\quad d_3=d_4=\ldots=d_{16}=3,\quad 	6 \leq d_{17}\leq d_{18}\leq \ldots \leq d_{23}\leq 15,\quad 	d_{24}=15,\quad d_{25}=150.
		\end{equation*}
	\end{example}
Therefore, this section develops a general framework for studying the Hamming distances of binary polycyclic codes associated with $P(x)^\mathcal{L}$, yielding exact results in many cases and useful bounds in the remaining ones. We next consider the Euclidean duals of these types of codes and investigate their distance properties.

\section{The Hamming Distance of the dual codes}\label{sec4}
In this section, we investigate the Euclidean dual of binary polycyclic codes associated with $P(x)^\mathcal{L}$, where $P(x)$ is a binary irreducible polynomial of degree $m \geq 2$, and $\mathcal{L} \geq2$ is a positive integer.
By Theorem~\ref{theorem1}, the corresponding binary polycyclic codes are precisely the ideals
\begin{equation*}
	C_j= \langle P(x)^j \rangle \subseteq \mathcal{P}=\frac{\mathbb{F}_2[x]}{\langle P(x)^{\mathcal{L}} \rangle},\quad 0\leq j \leq \mathcal{L}.
\end{equation*} Since $C_0= \mathcal{P}$ and $C_{\mathcal{L}}= \{0\}$, we immediately have
\begin{equation*}
	C_0^\perp= \{0\}, \quad C_{\mathcal{L}}^\perp= \mathcal{P}.
\end{equation*}
We now determine $C_j^\perp$ for $1\leq j \leq \mathcal{L}-1$.

To derive an explicit description of the dual codes, we first introduce the following notations.
Let $e$ be the order of  $P(x)$, and let $\mathcal{T}$ be the unique positive integer satisfying $2^{\mathcal{T}-1} < \mathcal{L} \leq 2^\mathcal{T}$.  Then  $Ord\left(P(x)^\mathcal{L}\right)= e2^\mathcal{T}$(see \cite{lidl1997finite}).
Let $U(x) \in \mathbb{F}_2[x]$ such that 
$
P(x)U(x)= x^{e}+1.
$
Then
\begin{equation*}
	P^*(x) U^*(x)= \left( x^{e}+1\right)^*= x^{e}+1,
\end{equation*} where $^*$ denotes the reciprocal polynomial.
Now, we define
\begin{equation*}
	h(x)= \frac{ x^{e2^\mathcal{T}}+1 }{P(x)^j}=  \frac{\left(x^{e}+1\right)^{2^\mathcal{T}}}{P(x)^j}= \frac{P(x)^{2^\mathcal{T}}  U(x)^{2^\mathcal{T}}}{P(x)^j} = P(x)^{2^\mathcal{T}-j}  U(x)^{2^\mathcal{T}}.
\end{equation*}
By taking reciprocals, we obtain
\begin{equation*}
	h^*(x)= P^*(x)^{2^\mathcal{T}-j} U^*(x)^{2^\mathcal{T}}= \left( P^*(x) U^*(x) \right)^{2^\mathcal{T}-j} U^*(x)^j= \left(x^{e}+1\right)^{2^\mathcal{T}-j} U^*(x)^j.
\end{equation*}
By Theorem~5.3 of \cite{bansal2026binary}, the dual code of $C_j$ is given as follows:
\begin{theorem}\label{theorem11}
	For $1\leq j \leq \mathcal{L}-1$, the dual  of the code $C_j$ is given by
	\begin{equation*}
		C_j^\perp= \left\{  \left(x^{e}+1\right)^{2^\mathcal{T}-j} U^*(x)^j a(x) \bmod \left(x^{m\mathcal{L}}\right) \mid a(x) \in \mathbb{F}_2[x],\ \deg(a(x)) < mj \right\}.
	\end{equation*}
\end{theorem}
We first determine the Hamming distance of $C_j^\perp$ for $j= 2^{\mathcal{T}-s}$, where $1\leq s \leq \mathcal{T}$.
\begin{theorem}\label{theorem12}
		For $1\leq s \leq \mathcal{T}$,
	the Hamming distance of the code $C_{2^{\mathcal{T}-s}}^{\perp}$  is given by
	\begin{equation*}
		d_{2^{\mathcal{T}-s}}^\perp= \min  \left\{  wt \left(b(x)\right)\mid b(x) \in B_{\mathcal{L},\mathcal{T}-s}^m	\right\},
	\end{equation*}
	where the set $B_{\mathcal{L},\mathcal{T}-s}^m$ is defined as
	\begin{equation*}
		B_{\mathcal{L},\mathcal{T}-s}^m = \left\{ \left(\ell(x)V(x)\right)^{2^{\mathcal{T}-s}} x^{2^{\mathcal{T}-s}-1}  \bmod \left(x^{m\mathcal{L}}\right) \Big|\ \ell(x)\in \mathbb{F}_2[x],\ \deg(\ell(x))=m-1 \right\},
	\end{equation*} and $V(x)= \left(x^{e}+1\right)^{2^{s}-1} U^*(x)$.
\end{theorem}
\begin{proof}
		Let $c(x)$ be a non-zero codeword in $C_{2^{\mathcal{T}-s}}^{\perp}$.  By Theorem~\ref{theorem11}, it can be written as
	\begin{equation*}
		c(x)= \left(x^{e}+1\right)^{2^{\mathcal{T}}-2^{\mathcal{T}-s}} U^*(x)^{2^{\mathcal{T}-s}} a(x)  \bmod \left(x^{m\mathcal{L}}\right),\quad \text{where } a(x)\in \mathbb{F}_2[x],\ \deg (a(x)) < m2^{\mathcal{T}-s}.
	\end{equation*}
	Let $V(x)= \left(x^{e}+1\right)^{2^{s}-1} U^*(x). $ Then 
	\begin{equation*}
		c(x)= V(x)^{2^{\mathcal{T}-s}}  a(x)  \bmod \left(x^{m\mathcal{L}}\right).
	\end{equation*}
	Applying the division algorithm as in Theorem~\ref{theorem4}, we can express $a(x)$ as
	\begin{equation*}
		a(x)= \left(1+x^{m-1}\right)^{2^{\mathcal{T}-s}} a_1(x)+  \left(1+x^{m-2}\right)^{2^{\mathcal{T}-s}} a_2(x)+ \ldots+ (1+x)^{2^{\mathcal{T}-s}} a_{m-1}(x)+a_m(x),
	\end{equation*} where  each  $a_i(x)\in \mathbb{F}_2[x]$ and $\deg \left(a_i(x)\right)< 2^{\mathcal{T}-s}$. Thus,  $a_i(x)$ can be written as 
	\begin{equation*}
		a_i(x)= \sum_{j=0}^{2^{\mathcal{T}-s}-1} a_{i,j}x^{j}, \quad  a_{i,j}\in \mathbb{F}_2.
	\end{equation*}
	Substituting the values of $a(x)$ and $a_i(x)$ into $c(x)$, we obtain
	\allowdisplaybreaks
		\allowdisplaybreaks
	\begin{align}\label{eq32}
		\nonumber	c(x)= & \bigg(\left( \left(1+x^{m-1}\right) V(x) \right)^{2^{\mathcal{T}-s}} \sum_{j=0}^{2^{\mathcal{T}-s}-1} a_{1,j}x^{j} +  \left( \left(1+x^{m-2}\right) V(x) \right)^{2^{\mathcal{T}-s}}\sum_{j=0}^{2^{\mathcal{T}-s}-1} a_{2,j}x^{j}+\ldots+\\& \left( \left(1+x\right) V(x)\right)^{2^{\mathcal{T}-s}}\sum_{j=0}^{2^{\mathcal{T}-s}-1} a_{m-1,j}x^{j}+  V(x)^{2^{\mathcal{T}-s}}\sum_{j=0}^{2^{\mathcal{T}-s}-1} a_{m,j}x^{j}\bigg) \bmod \left(x^{m\mathcal{L}}\right).
	\end{align}
	For each $j$, where $0\leq j \leq 2^{\mathcal{T}-s}-1$, let
		\begin{equation}\label{eq33}
		c_j(x)=\Big(\left( \left(1+x^{m-1}\right) V(x) \right)^{2^{\mathcal{T}-s}} a_{1,j}+\left( \left(1+x^{m-2}\right) V(x) \right)^{2^{\mathcal{T}-s}} a_{2,j}+\ldots + \left( \left(1+x\right) V(x) \right)^{2^{\mathcal{T}-s}} a_{m-1,j}+  V(x)^{2^{\mathcal{T}-s}} a_{m,j}\Big) x^j.
	\end{equation}
 Then
\begin{align}\label{eq34}
	c(x)= & \Bigg(\sum_{j=0}^{2^{\mathcal{T}-s}-1} c_j(x)\Bigg)  \bmod \left(x^{m\mathcal{L}}\right)=  \sum_{j=0}^{2^{\mathcal{T}-s}-1} \left(c_j(x) \bmod \left(x^{m\mathcal{L}}\right)\right).
\end{align}
\textbf{Step I: Decomposition of the Hamming weight of $c(x)$.
}\\
It follows from Eq.~\eqref{eq33} that,  for $0 \leq j_1, j_2 \leq 2^{\mathcal{T}-s}-1$ with $ j_1 \ne j_2$, no term of $c_{j_1}(x)$ coincides with any term of $c_{j_2}(x).$ We now show that this remains true after reduction modulo $x^{m\mathcal{L}}$. We may write
\begin{equation*}
	c_j(x)= \sum_{l=0}^{L_j} c_{j,l} x^{2^{\mathcal{T}-s}l+j},\quad \text{where } L_j \geq 0,\ c_{j,l} \in \mathbb{F}_2.
\end{equation*}
If $2^{\mathcal{T}-s} L_j +j < m \mathcal{L}$, then  no term is removed under reduction modulo $x^{m\mathcal{L}}$, and hence 
$
c_j(x) \bmod \left(x^{m\mathcal{L}}\right)= 	c_j(x).
$ Otherwise, the terms of degree at least	${m\mathcal{L}}$ vanish.  Let $N_j$ be the largest integer satisfying   $2^{\mathcal{T}-s} N_j +j < m \mathcal{L}$, then
$
c_j(x) \bmod \left(x^{m\mathcal{L}}\right)= 	\sum_{l=0}^{N_j} c_{j,l} x^{2^{\mathcal{T}-s}l+j}.
$
Therefore,  for $ j_1 \ne j_2$, no term of  $c_{j_1}(x) \bmod \left(x^{m\mathcal{L}}\right)$ coincides with any term of $c_{j_2}(x) \bmod \left(x^{m\mathcal{L}}\right).$ Hence, using Eq.~\eqref{eq34} we obtain
\begin{equation}\label{eq35}
	wt (c(x))= \sum_{j=0}^{2^{\mathcal{T}-s}-1} wt \Big(c_j(x) \bmod \left(x^{m\mathcal{L}}\right)\Big).
\end{equation}
\textbf{Step II:  The weight of $c_j(x) \bmod \left(x^{m\mathcal{L}}\right)$  does not increase under multiplication by $x^r$.}\\
Let $\deg (c_j(x))= S_j \geq0$. Then we may write
\begin{equation*}
	c_j(x)= \sum_{i=0}^{S_j} \mathfrak{c}_{j,i}x^{i}, \quad \text{where } \mathfrak{c}_{j,i}\in \mathbb{F}_2, \ \mathfrak{c}_{j,S_j}=1.
\end{equation*}
Now, we examine the following possible cases for $\deg(c_j(x))$:
\begin{itemize}
	\item If $ S_j < m\mathcal{L}-1$, then 
	\begin{equation*}
		wt \Big(c_j(x) \bmod \left(x^{m\mathcal{L}}\right)\Big)= wt \left(c_j(x)\right)= wt \left(\mathfrak{c}_{j,0},\ \mathfrak{c}_{j,1},\ldots, \mathfrak{c}_{j,S_j}\right),
	\end{equation*}
	\begin{equation*}
		wt \Big(xc_j(x) \bmod \left(x^{m\mathcal{L}}\right)\Big)= wt \left(xc_j(x)\right)= wt \left(\mathfrak{c}_{j,0},\ \mathfrak{c}_{j,1},\ldots, \mathfrak{c}_{j,S_j}\right)=	wt \Big(c_j(x) \bmod \left(x^{m\mathcal{L}}\right)\Big).
	\end{equation*}
	\item  If  $S_j = m\mathcal{L}-1$, then 
	\begin{equation*}
		wt \Big(c_j(x) \bmod \left(x^{m\mathcal{L}}\right)\Big)= wt \left(c_j(x)\right)= wt \left(\mathfrak{c}_{j,0},\ \mathfrak{c}_{j,1},\ldots, \mathfrak{c}_{j,S_j}\right),
	\end{equation*}
	\begin{equation*}
		wt \Big(xc_j(x) \bmod \left(x^{m\mathcal{L}}\right)\Big)=  wt \left(\mathfrak{c}_{j,0},\ \mathfrak{c}_{j,1},\ldots, \mathfrak{c}_{j,S_j-1}\right)= 	wt \Big(c_j(x) \bmod \left(x^{m\mathcal{L}}\right)\Big)-1.
	\end{equation*}
	\item If  $ S_j \geq m\mathcal{L}$, then 
	\begin{equation*}
		wt \Big(c_j(x) \bmod \left(x^{m\mathcal{L}}\right)\Big)= wt \left(\mathfrak{c}_{j,0},\ \mathfrak{c}_{j,1},\ldots, \mathfrak{c}_{j,m\mathcal{L}-1}\right),
	\end{equation*} 
	\begin{equation*}
		wt \Big(xc_j(x) \bmod \left(x^{m\mathcal{L}}\right)\Big)=  wt \left(\mathfrak{c}_{j,0},\ \mathfrak{c}_{j,1},\ldots, \mathfrak{c}_{j,m\mathcal{L}-2}\right) \leq  	wt \Big(c_j(x) \bmod \left(x^{m\mathcal{L}}\right)\Big).
	\end{equation*}
\end{itemize}
Therefore, in all cases 
\begin{equation*}
	wt  \Big(c_j(x) \bmod \left(x^{m\mathcal{L}}\right)\Big) \geq 	wt \Big(xc_j(x) \bmod \left(x^{m\mathcal{L}}\right)\Big).
\end{equation*}
Repeating the same argument inductively, we obtain that for every $r\geq0$,
\begin{equation}\label{eq36}
	wt  \Big(c_j(x) \bmod \left(x^{m\mathcal{L}}\right)\Big) \geq 	wt \Big(x^r c_j(x) \bmod \left(x^{m\mathcal{L}}\right)\Big).
\end{equation}
\textbf{Step III: A lower bound on the weight for each nonzero component $c_j(x) \bmod \left(x^{m\mathcal{L}}\right)$.}\\
Since $c(x) \ne 0$, there exists at least one index $j$ such that $c_j(x)  \bmod \left(x^{m\mathcal{L}}\right) \ne 0$. Fix such a $j$. Suppose that in Eq.~\eqref{eq33}, the non-zero coefficients among $a_{1,j}, a_{2,j},\ldots,a_{m,j}$ occur exactly at the positions $1\leq t_1 < t_2 <\ldots< t_\sigma \leq m,$ that is,
\begin{equation}\label{eq37}
	a_{t_1,j}= a_{t_2,j}=\ldots= a_{t_\sigma,j}=1, \ 
\end{equation} and all  remaining $a_{t,j}=0$.	 Now, we consider the following cases depending on $t_\sigma$ (see Eq.~\eqref{eq37}):\\
	\textbf{Case I:} If $t_\sigma < m$, then from Eqs.~\eqref{eq33} and \eqref{eq37}, we have
\begin{equation*}
	c_j(x)= \Big( \left(1+x^{m-t_1}\right) V(x) + \left(1+x^{m-t_2}\right) V(x) + \ldots + \left(1+x^{m-t_\sigma}\right) V(x) \Big)^{2^{\mathcal{T}-s}} x^{j}.
\end{equation*}
Now, consider the following subcases:

	\textbf{Subcase I:} If $\sigma$ is odd, then
\allowdisplaybreaks
\begin{align*}
	c_j(x)=& \Big(\left(1+x^{m-t_1}+x^{m-t_2}+\ldots+x^{m-t_\sigma}\right)V(x)\Big)^{2^{\mathcal{T}-s}}x^j
	= \Big(\left(1+x^{m-t_\sigma}+x^{m-t_{\sigma-1}}+\ldots+x^{m-t_1}\right)V(x)\Big)^{2^{\mathcal{T}-s}}x^j.
\end{align*}
Consequently,\allowdisplaybreaks
\begin{align*}
	x^{(t_1-1)2^{\mathcal{T}-s}} x^{2^{\mathcal{T}-s}-1-j} c_j(x) = & \Big(\big(x^{t_1-1}+x^{(m-1)-\left(t_\sigma-t_1\right)}+x^{(m-1)- \left(t_{\sigma-1}-t_1\right)}+\ldots+x^{(m-1)-\left(t_2-t_1\right)}+x^{m-1}\big)V(x)\Big)^{2^{\mathcal{T}-s}}\\&x^{2^{\mathcal{T}-s}-1}.
\end{align*}
Using Eq.~\eqref{eq36}, we conclude that \allowdisplaybreaks
\begin{align} \label{eq38}
	\nonumber	wt  \Big(c_j(x) \bmod \left(x^{m\mathcal{L}}\right)\Big) \geq  &  wt \Big(	x^{(t_1-1)2^{\mathcal{T}-s}} x^{2^{\mathcal{T}-s}-1-j} c_j(x) \bmod \left(x^{m\mathcal{L}}\right)\Big)\\\nonumber =& wt \bigg(\Big(\big(x^{t_1-1}+x^{(m-1)-\left(t_\sigma-t_1\right)}+x^{(m-1)- \left(t_{\sigma-1}-t_1\right)}+\ldots+x^{(m-1)-\left(t_2-t_1\right)}+x^{m-1}\big) V(x)\Big)^{2^{\mathcal{T}-s}}\\&\quad\quad x^{2^{\mathcal{T}-s}-1} \bmod \left(x^{m\mathcal{L}}\right)\bigg).
\end{align}

	\textbf{Subcase II:} If $\sigma$ is even, then
\begin{align*}
	c_j(x)= &  \Big(\left(x^{m-t_1}+x^{m-t_2}+\ldots+x^{m-t_\sigma}\right)V(x)\Big)^{2^{\mathcal{T}-s}}x^j= \Big(\left(x^{m-t_\sigma}+x^{m-t_{\sigma-1}}+\ldots+x^{m-t_1}\right)V(x)\Big)^{2^{\mathcal{T}-s}}x^j,
\end{align*}
and 
\begin{align*}
	x^{(t_1-1)2^{\mathcal{T}-s}} x^{2^{\mathcal{T}-s}-1-j} c_j(x) = & \Big(\big(x^{(m-1)-\left(t_\sigma-t_1\right)}+x^{(m-1)- \left(t_{\sigma-1}-t_1\right)}+\ldots+x^{(m-1)-\left(t_2-t_1\right)}+x^{m-1}\big)V(x)\Big)^{2^{\mathcal{T}-s}}x^{2^{\mathcal{T}-s}-1}.
\end{align*}
Thus, by Eq.~\eqref{eq36}, we have 
\allowdisplaybreaks
\begin{align}\label{eq39}
	\nonumber	wt \Big(c_j(x) \bmod \left(x^{m\mathcal{L}}\right)\Big)\geq & wt \bigg(\Big(\Big(x^{(m-1)-\left(t_\sigma-t_1\right)}+x^{(m-1)- \left(t_{\sigma-1}-t_1\right)}+\ldots+x^{(m-1)-\left(t_2-t_1\right)}+x^{m-1}\Big) V(x)\Big)^{2^{\mathcal{T}-s}} \\&\quad\quad x^{2^{\mathcal{T}-s}-1} \bmod \left(x^{m\mathcal{L}}\right)\bigg).
\end{align}
	\textbf{Case II:} If $t_\sigma= m$, then we consider the following subcases:

\textbf{Subcase I:} If $\sigma=1$, then 
\begin{equation*}
	c_j(x)= V(x)^{2^{\mathcal{T}-s}} x^j,
\end{equation*}
\begin{equation}\label{eq40}
	wt \Big(c_j(x) \bmod \left(x^{m\mathcal{L}}\right)\Big)\geq wt \Big( \left(x^{m-1} V(x)\right)^{2^{\mathcal{T}-s}} x^{2^{\mathcal{T}-s}-1} \bmod \left(x^{m\mathcal{L}}\right) \Big).
\end{equation}

\textbf{Subcase II:} If $\sigma$ is odd and $\sigma\geq 3$, then from Eqs.~\eqref{eq33} and \eqref{eq37},
\allowdisplaybreaks
\begin{align*}
	c_j(x)=  &\Big( \left(1+x^{m-t_1}\right) V(x) + \left(1+x^{m-t_2}\right) V(x) + \ldots + \left(1+x^{m-t_{\sigma-1}}\right) V(x)+V(x) \Big)^{2^{\mathcal{T}-s}} x^{j}\\
	=& \Big(\left(1+x^{m-t_{\sigma-1}}+x^{m-t_{\sigma-2}}+\ldots+x^{m-t_{1}}\right)V(x)\Big)^{2^{\mathcal{T}-s}}x^j.
\end{align*}
Thus,  as in Subcase I of Case I, we get
\allowdisplaybreaks
\begin{align}\label{eq41}
	\nonumber	wt  \Big(c_j(x) \bmod \left(x^{m\mathcal{L}}\right)\Big) \geq  &   wt \bigg(\Big(\big(x^{t_1-1}+x^{(m-1)-\left(t_{\sigma-1}-t_1\right)}+x^{(m-1)- \left(t_{\sigma-2}-t_1\right)}+\ldots+x^{(m-1)-\left(t_2-t_1\right)}+x^{m-1}\big) V(x)\Big)^{2^{\mathcal{T}-s}}\\&\quad\quad x^{2^{\mathcal{T}-s}-1} \bmod \left(x^{m\mathcal{L}}\right)\bigg).
\end{align}

\textbf{Subcase III:} If $\sigma$ is even and $\sigma\geq 2$, then 
\begin{equation*}
	c_j(x)= \Big(\left(x^{m-t_1}+x^{m-t_2}+\ldots+x^{m-t_{\sigma-1}}\right)V(x)\Big)^{2^{\mathcal{T}-s}}x^j,
\end{equation*}
\allowdisplaybreaks
\begin{align}\label{eq42}
	\nonumber	wt \Big(c_j(x) \bmod \left(x^{m\mathcal{L}}\right)\Big)\geq & wt \bigg(\Big(\big(x^{(m-1)-\left(t_{\sigma-1}-t_1\right)}+x^{(m-1)- \left(t_{\sigma-2}-t_1\right)}+\ldots+x^{(m-1)-\left(t_2-t_1\right)}
	+x^{m-1}\big) V(x)\Big)^{2^{\mathcal{T}-s}}	\\&\quad\quad x^{2^{\mathcal{T}-s}-1} \bmod \left(x^{m\mathcal{L}}\right)\bigg).
\end{align}
	From Eqs.~\eqref{eq38}-\eqref{eq42}, we conclude that in all the cases, 
\begin{equation*}
	wt \Big(c_j(x) \bmod \left(x^{m\mathcal{L}}\right)\Big)\geq  \min  \left\{  wt \left(b(x)\right)\mid b(x) \in B_{\mathcal{L},\mathcal{T}-s}^m	\right\}. 
\end{equation*}
\textbf{Step IV: A lower bound on the weight of every nonzero codeword}\\
By Eq.~\eqref{eq35}, for every non-zero codeword $c(x)\in 	C_{2^{\mathcal{T}-s}}^\perp$, we have
\begin{equation*}
	wt(c(x)) \geq 	wt \Big(c_j(x) \bmod \left(x^{m\mathcal{L}}\right)\Big)\geq \min  \left\{  wt \left(b(x)\right)\mid b(x) \in B_{\mathcal{L},\mathcal{T}-s}^m	\right\}. 
\end{equation*} Therefore, 
\begin{equation} \label{eq43}
	d_{2^{\mathcal{T}-s}}^\perp  \geq  \min  \left\{  wt \left(b(x)\right)\mid b(x) \in B_{\mathcal{L},\mathcal{T}-s}^m	\right\}. 
\end{equation}
\textbf{Step V: Reverse inequality} \\
Since $B_{\mathcal{L},\mathcal{T}-s}^m \subseteq C_{2^{\mathcal{T}-s}}^\perp$, it follows that
\begin{equation}\label{eq44}
	d_{2^{\mathcal{T}-s}}^\perp \leq \min  \left\{  wt \left(b(x)\right)\mid b(x) \in B_{\mathcal{L},\mathcal{T}-s}^m	\right\}. 
\end{equation}
Combining Inequalities~\eqref{eq43} and \eqref{eq44}, we obtain the desired equality. This completes the proof.
\end{proof}

The following corollary gives the Hamming distance of $C_1^\perp$ when  $P(x)=x^{2\cdot3^v}+x^{3^v}+1,\ v\geq0$ and $\mathcal{L}=2^\mathcal{T}$.
\begin{corollary}\label{corollary3}
	Let $C_1= \langle x^{2\cdot3^v}+x^{3^v}+1 \rangle$ be the binary polycyclic code associated with $\left( x^{2\cdot3^v}+x^{3^v}+1 \right)^{2^\mathcal{T}}$, where $\mathcal{T}\geq 1$ and  $v \geq 0$. Then the Hamming distance of $C_1^\perp$ is given by 
	\begin{equation*}
		d_1^\perp= \begin{cases}
			\frac{2^{\mathcal{T}+2}-2}{3},& \text{ if } \mathcal{T} \text{ is odd,}\\
			\frac{2^{\mathcal{T}+2}-1}{3},& \text{ if } \mathcal{T} \text{ is even.}
		\end{cases}
	\end{equation*}
\end{corollary}
\begin{proof}
		We have $P(x)= x^{2\cdot3^v}+x^{3^v}+1$.  Let $e$ be the order of the polynomial $x^{2\cdot3^v}+x^{3^v}+1$. Since $\left(x^{2\cdot3^v}+x^{3^v}+1\right) \left(1+x^{3^v}\right)= x^{3^{v+1}}+1$, thus $e \leq 3^{v+1}$.  If we assume that $e< 3^{v+1}$, then there exists $g(x)\in \mathbb{F}_2[x]$ with $\deg \left(g(x)\right)< 3^v$ such that 
		$
		\left(x^{2\cdot3^v}+x^{3^v}+1\right) g(x)= x^{e}+1.
		$
		Since $cw 	\left(x^{2\cdot3^v}+x^{3^v}+1\right) =3^v $ and $\deg \left(g(x)\right)< 3^v$, thus
		$
		wt \left(	\left(x^{2\cdot3^v}+x^{3^v}+1\right) g(x)\right)= wt 	\left(x^{2\cdot3^v}+x^{3^v}+1\right)  \cdot wt\left(g(x)\right) \geq 3,
		$ which contradicts the fact that $wt\left(x^{e}+1\right)=2$.  Therefore,
		\begin{equation*}
			Ord \big(x^{2\cdot3^v}+x^{3^v}+1\big) = 3^{v+1}\quad \text{ for every } v \geq0.
		\end{equation*} Thus, here $Ord(P(x))=e=3^{v+1}$, $\mathcal{L}=2^\mathcal{T}$.  Also,  $U(x)=U^*(x)= x^{3^v}+1$. By Theorem~\ref{theorem12},
	\begin{equation*}
		d_1^\perp= \min \left\{wt(b(x))\ \big| \ b(x)\in B_{2^\mathcal{T}, \mathcal{T}-\mathcal{T}}^{2\cdot3^v}  \right\},
	\end{equation*} where 
	\begin{equation*}
		B_{2^\mathcal{T}, \mathcal{T}-\mathcal{T}}^{2\cdot3^v} = \left\{ a(x)V(x)   \bmod \big(x^{3^v 2^{\mathcal{T}+1}}\big) \Big|\ a(x)\in \mathbb{F}_2[x],\ \deg(a(x))=2\cdot3^v-1 \right\},
	\end{equation*} and 
	\begin{equation*}
		V(x)= \big(x^{3^{v+1}}+1\big)^{2^\mathcal{T}-1} \big(x^{3^v}+1\big).
	\end{equation*}
	Let $b(x)$ be an element of $B_{2^\mathcal{T}, \mathcal{T}-\mathcal{T}}^{2\cdot3^v}$. Then
	\begin{equation*}
		b(x)= \big(x^{3^{v+1}}+1\big)^{2^\mathcal{T}-1} \big(x^{3^v}+1\big) a(x) \mod \Big(x^{3^v 2^{\mathcal{T}+1}}\Big),
	\end{equation*} where $a(x)\in \mathbb{F}_2[x]$ and $\deg(a(x))= 2\cdot3^v-1.$ It can be rewritten as 
	\allowdisplaybreaks
	\begin{align*}
		b(x)= & \left(1+x^{3^{v+1}}+x^{2(3^{v+1})}+\ldots+x^{(2^\mathcal{T}-1)3^{v+1}}\right)  \big(x^{3^v}+1\big) a(x) \mod \left(x^{3^v 2^{\mathcal{T}+1}}\right)\\
		=&  \big(x^{3^v}+1\big) a(x)+ x^{3^{v+1}}  \big(x^{3^v}+1\big) a(x)+ \ldots+ x^{(2^\mathcal{T}-1)3^{v+1}} \big(x^{3^v}+1\big) a(x) \mod \left(x^{3^v 2^{\mathcal{T}+1}}\right).
	\end{align*} 
	Let $M$ be the largest positive integer such that $M \leq 2^\mathcal{T}-1$ and  $M 3^{v+1} < 3^v2^{\mathcal{T}+1}$. Then 
		\begin{equation*}
		M= \bigg\lfloor \frac{2^{\mathcal{T}+1}}{3} \bigg\rfloor=\begin{cases}
			\frac{2^{\mathcal{T}+1}-1}{3}, & \text{ if } \mathcal{T} \text{ is odd,}\\
			\frac{2^{\mathcal{T}+1}-2}{3}, & \text{ if } \mathcal{T} \text{ is even.}
		\end{cases}
	\end{equation*} Now, we consider the following two cases:\\
	\textbf{Case I:} If $\mathcal{T}$ is odd, then
	\begin{equation*}
		b(x)=\big(x^{3^v}+1\big) a(x)+ x^{3^{v+1}}  \big(x^{3^v}+1\big) a(x)+\ldots+ x^{ \left(\frac{2^{\mathcal{T}+1}-1}{3}\right)3^{v+1}} \big(x^{3^v}+1\big) a(x) \mod \left(x^{3^v 2^{\mathcal{T}+1}}\right).
	\end{equation*} By using the division algorithm, the polynomial $a(x)$ can be expressed as 
	\begin{equation*}
		a(x)= \big(x^{3^v}+1\big) a_1(x)+a_2(x),\quad \text{where } a_1(x),\ a_2(x)\in \mathbb{F}_2,\ \deg(a_1(x))=3^v-1,\ \deg(a_2(x))<3^v.
	\end{equation*} Thus, \allowdisplaybreaks
	\begin{align*}
		\big(x^{3^v}+1\big) a(x)=& \big(x^{2\cdot3^v}+1\big) a_1(x)+ \big(x^{3^v}+1\big) a_2(x)=  \left(a_1(x)+a_2(x)\right)+x^{3^v}a_2(x)+x^{2\cdot3^v}a_1(x).
	\end{align*} Therefore, we have \allowdisplaybreaks
\begin{align*}
	b(x)= & \big(x^{3^v}+1\big) a(x)+ x^{3^{v+1}}  \big(x^{3^v}+1\big) a(x)+\ldots+ x^{ \left(\frac{2^{\mathcal{T}+1}-4}{3}\right)3^{v+1}} \big(x^{3^v}+1\big) a(x)\\ &+ x^{ \left(\frac{2^{\mathcal{T}+1}-1}{3}\right)3^{v+1}} \left(a_1(x)+a_2(x)\right) \mod \left(x^{3^v 2^{\mathcal{T}+1}}\right).
\end{align*}
 Hence,\allowdisplaybreaks
	\begin{align*}
		wt(b(x))= & \frac{2^{\mathcal{T}+1}-1}{3} 	wt \Big(\big(x^{3^v}+1\big) a(x)\Big)+ wt \big(a_1(x)+a_2(x)\big)\\
		=& \frac{2^{\mathcal{T}+1}-1}{3} \Big( wt \big(a_1(x)+a_2(x)\big)+ wt(a_2(x))+wt(a_1(x))\Big) + wt \big(a_1(x)+a_2(x)\big).
	\end{align*}    
	If $a_1(x)=a_2(x)$, then 
	\begin{equation*}
		wt(b(x))=  \frac{2^{\mathcal{T}+2}-2}{3} wt(a_1(x)) \geq \frac{2^{\mathcal{T}+2}-2}{3}.
	\end{equation*} If $a_1(x)\ne a_2(x)$, then
	\begin{equation*}
		wt(b(x))=  \frac{2^{\mathcal{T}+1}-1}{3} \Big( wt \big(a_1(x)+a_2(x)\big)+ wt(a_2(x))+wt(a_1(x))\Big) + wt \left(a_1(x)+a_2(x)\right) \geq \frac{2^{\mathcal{T}+2}-2}{3}+1.
	\end{equation*} Moreover, for $a(x)= \left(x^{3^v}+1\right)x^{3^v-1}+ x^{3^v-1}$, we have 
	\begin{equation*}
		wt(b(x))= \frac{2^{\mathcal{T}+2}-2}{3}.
	\end{equation*} Therefore, in this case $d_1^\perp= \frac{2^{\mathcal{T}+2}-2}{3}.$\\
	\textbf{Case II:} If $\mathcal{T}$ is even, then 
	\begin{equation*}
		b(x)=\big(x^{3^v}+1\big) a(x)+ x^{3^{v+1}}  \big(x^{3^v}+1\big) a(x)+\ldots+ x^{ (\frac{2^{\mathcal{T}+1}-2}{3})3^{v+1}} \big(x^{3^v}+1\big) a(x) \mod \left(x^{3^v 2^{\mathcal{T}+1}}\right).
	\end{equation*} Expressing $a(x)$ as in Case~I, we have \allowdisplaybreaks
	\begin{align*}
		b(x)=& \big(x^{3^v}+1\big) a(x)+ x^{3^{v+1}}  \big(x^{3^v}+1\big) a(x)+\ldots+ x^{ (\frac{2^{\mathcal{T}+1}-5}{3})3^{v+1}} \big(x^{3^v}+1\big) a(x)\\&+ x^{ (\frac{2^{\mathcal{T}+1}-2}{3})3^{v+1}} \Big(a_1(x)+a_2(x)+x^{3^v}a_2(x)\Big) \mod \left(x^{3^v 2^{\mathcal{T}+1}}\right).
	\end{align*}
	Consequently,
	\begin{equation*}
		wt(b(x))= \frac{2^{\mathcal{T}+1}-2}{3} \Big( wt \big(a_1(x)+a_2(x)\big)+ wt(a_2(x))+wt(a_1(x))\Big) + wt \big(a_1(x)+a_2(x)\big)+wt(a_2(x)).
	\end{equation*} If $a_1(x)=a_2(x)$, then
	\begin{equation*}
		wt(b(x))= \Big( \frac{2^{\mathcal{T}+2}-4}{3}+1\Big) wt(a_1(x)) \geq \frac{2^{\mathcal{T}+2}-1}{3}.
	\end{equation*} If $a_1(x)\ne a_2(x)$, then\allowdisplaybreaks
	\begin{align*}
		wt(b(x))=  & \frac{2^{\mathcal{T}+1}-2}{3} \Big( wt \big(a_1(x)+a_2(x)\big)+ wt(a_2(x))+wt(a_1(x))\Big) + wt \left(a_1(x)+a_2(x)\right)+wt(a_2(x))\\ \geq & \frac{2^{\mathcal{T}+2}-4}{3}+1=\frac{2^{\mathcal{T}+2}-1}{3}.
	\end{align*}
	Moreover, for $a(x)= \left(x^{3^v}+1\right)x^{3^v-1}+ x^{3^v-1}$, we have 
	\begin{equation*}
		wt(b(x))= \frac{2^{\mathcal{T}+2}-1}{3}.
	\end{equation*} Therefore, in this case $d_1^\perp= \frac{2^{\mathcal{T}+2}-1}{3}.$
	This completes the proof.
\end{proof}
In the following example, we illustrate Theorem~\ref{theorem12} for $P(x)=x^3+x+1$ and $\mathcal{L}=9$.
\begin{example}\label{ex5}
	The Hamming distances of  $C_{2^{4-s}}^\perp=\big\langle \left(x^3+x+1\right)^{2^{4-s}} \big\rangle,\ 1\leq s \leq 4$, for  $P(x)=x^3+x+1$ and $\mathcal{L}=9$.
	
	Here, $P(x)=x^3+x+1$ and $\mathcal{L}=9$. Thus, $m=3$ and $\mathcal{T}=4.$ The order of $P(x)$ is $e=7$, and since $\left(x^3+x+1\right) \big(x^4+x^2+x+1\big)=x^7+1$, thus $U(x)= x^4+x^2+x+1$ and  $U^*(x)= x^4+x^3+x^2+1$. Now,
	\begin{equation*}
		V(x)= \left(x^{7}+1\right)^{2^s-1} U^*(x)= \left(x^{7}+1\right)^{2^s-1} \left(x^4+x^3+x^2+1\right).
	\end{equation*}
	Since 	\begin{equation*}
		\left(x^7+1\right)^{2^s-1} = 1+x^7+x^{14}+\ldots+x^{7\left(2^s-1\right)}= \sum_{j=0}^{2^s-1} x^{7j}.
	\end{equation*}
	Thus, for $1\leq s \leq 4$, 
	\begin{equation*}
		V(x)= \sum_{j=0}^{2^s-1} x^{7j} \left(1+x^2+x^3+x^4\right).
	\end{equation*}
	By Theorem~\ref{theorem12},  	for $1\leq s \leq 4$,
	the Hamming distance of the code $C_{2^{4-s}}^{\perp}$  is given by
	\begin{equation*}
		d_{2^{4-s}}^\perp= \min  \left\{  wt \left(b(x)\right)\mid b(x) \in B_{9,4-s}^3	\right\},
	\end{equation*}
	where 	\allowdisplaybreaks
	\begin{align*}
		B_{9,4-s}^3 = & \left\{ \left(a(x)V(x)\right)^{2^{4-s}} x^{2^{4-s}-1}  \bmod \left(x^{27}\right) \Big|\ a(x)\in \mathbb{F}_2[x],\ \deg(a(x))=2 \right\}\\ =& \biggl\{   \big(x^2V(x)\big)^{2^{4-s}} x^{2^{4-s}-1}\bmod  \left(x^{27}\right),\ \Big( \big(1+x^2\big) V(x)\Big)^{2^{4-s}} x^{2^{4-s}-1}  \bmod \left(x^{27}\right), \\& \Big( \big(x+x^2\big) V(x)\Big)^{2^{4-s}} x^{2^{4-s}-1}  \bmod \left(x^{27}\right),
		\Big( \big(1+x+x^2\big) V(x)\Big)^{2^{4-s}} x^{2^{4-s}-1}  \bmod \left(x^{27}\right)	 \biggr\}.
	\end{align*}
The Hamming weights of the elements of	$	B_{9,4-s}^3$  for $1\leq s \leq 4$ are listed in Table~\ref{tab5}.
	\begin{table}[htbp]
		\caption{Hamming weights of elements of   $	B_{9,4-s}^3$ for $P(x)=x^3+x+1$ (Example~\ref{ex5})}
		\label{tab5}
		\centering
		\begin{tabular}{|c|c|c|c|c|}
			\hline
			Polynomial and its weight & $ s=1$ & $s=2$& $s=3$ & $s=4$ \\ \hline

			$wt \Big( \big(x^2V(x)\big)^{2^{4-s}} x^{2^{4-s}-1} \bmod \left(x^{27}\right) \Big)$ 
			& $1$ & $3$ & $7$ & $15$\\\hline

			$wt \Big( \big((1+x^2)V(x)\big)^{2^{4-s}} x^{2^{4-s}-1}  \bmod  \left(x^{27}\right) \Big)$
			& $1$ & $3$ & $7$ & $15$\\\hline

			$wt \Big( \big((x+x^2)V(x)\big)^{2^{4-s}} x^{2^{4-s}-1}  \bmod  \left(x^{27}\right) \Big)$ 
			& $2$ & $3$ & $7$ & $15$\\\hline

			$wt \Big( \big((1+x+x^2)V(x)\big)^{2^{4-s}}   x^{2^{4-s}-1} \bmod  \left(x^{27}\right) \Big)$  
			& $2$ & $3$ & $7$ & $15$\\\hline
			
		\end{tabular} 
	\end{table}
	
	From Table~\ref{tab5}, we obtain
	\begin{equation*}
		d_8^\perp= 1,\quad d_{4}^\perp=3,\quad d_2^\perp=7,\quad d_1^\perp=15.
	\end{equation*}
\end{example}
We next study the Hamming distance of  $C_{2^\mathcal{T}-2^{\mathcal{T}-r}}^\perp$  for  $1\leq r \leq \mathcal{T}$ when $\mathcal{L}=2^\mathcal{T}$, and for $1\leq r \leq R$ when $\mathcal{L}= 2^\mathcal{T}-2^{\mathcal{T}-R}+\mathcal{L}'$, where $2\leq R \leq \mathcal{T}-1$ and $1\leq \mathcal{L}' \leq 2^{\mathcal{T}-R-1}$.
\begin{theorem}\label{theorem13}
	If $\mathcal{L}=2^\mathcal{T}$, then for $1\leq r \leq \mathcal{T}$, the Hamming distance of $C_{2^\mathcal{T}-2^{\mathcal{T}-r}}^\perp$  is given by
	\begin{equation*}
		d_{2^\mathcal{T}-2^{\mathcal{T}-r}}^\perp= 
		\min  \left\{  wt \left(b(x)\right)\mid b(x) \in B_{\mathcal{L},\mathcal{T}-r}^{'m}	\right\},
	\end{equation*}
	where the set $B_{\mathcal{L},\mathcal{T}-r}^{'m}$ is defined as
	\begin{equation*}
		B_{\mathcal{L},\mathcal{T}-r}^{'m} = \left\{ \left(\ell(x)V'(x)\right)^{2^{\mathcal{T}-r}} x^{2^{\mathcal{T}-r}-1}  \bmod \left(x^{m\mathcal{L}}\right)\Big|\ \ell(x)\in \mathbb{F}_2[x],\ \deg(\ell(x))= m(2^r-1)-1 \right\},
	\end{equation*}
	and $V'(x)=\left(x^{e}+1\right)  U^*(x)^{2^r-1}$.	
\end{theorem}
\begin{proof}
	Let $c(x)$ be a non-zero element of $C_{2^\mathcal{T}-2^{\mathcal{T}-r}}^\perp$. Then, by Theorem~\ref{theorem11}, it can be expressed as
	\allowdisplaybreaks
	\begin{align*}
		c(x)= & \left(x^{e}+1\right)^{2^{\mathcal{T}-r}} U^*(x)^{2^\mathcal{T}-2^{\mathcal{T}-r}} a(x) \bmod \left(x^{m\mathcal{L}}\right)
		= \left(\left(x^{e}+1\right) U^*(x)^{2^r-1}\right)^{2^{\mathcal{T}-r}} a(x) \bmod \left(x^{m\mathcal{L}}\right),
	\end{align*}
	where $a(x)\in \mathbb{F}_2[x],\ \deg(a(x)) <m \left(2^\mathcal{T}-2^{\mathcal{T}-r}\right)$.  Let $V'(x)= \left(x^{e}+1\right) U^*(x)^{2^r-1}.  $ Then 
	\begin{equation*}
		c(x)= V'(x)^{2^{\mathcal{T}-r}} a(x) \bmod \left(x^{m\mathcal{L}}\right).
	\end{equation*}
	Since $\deg(a(x)) < m \left(2^\mathcal{T}-2^{\mathcal{T}-r}\right)= m \left(2^r-1\right) 2^{\mathcal{T}-r}$, thus by  using the division algorithm it can be represented as
	\allowdisplaybreaks
	\begin{align*}
		a(x)= & \Big(1+x^{m (2^r-1)-1}\Big)^{2^{\mathcal{T}-r}} a_1(x)+  \Big(1+x^{m (2^r-1)-2}\Big)^{2^{\mathcal{T}-r}} a_2(x)+ \ldots +\left(1+x\right)^{2^{\mathcal{T}-r}} a_{m(2^r-1)-1}(x)+ a_{m(2^r-1)}(x),
	\end{align*} where each $a_i(x)\in \mathbb{F}_2[x]$ and $\deg(a_i(x)) < 2^{\mathcal{T}-r}. $
	Proceeding in a  similar way as in  the proof of Theorem~\ref{theorem12}, we obtain the result.
\end{proof}
The proof of the next theorem follows the same argument as that of Theorem~\ref{theorem13}; therefore, we omit the details.
\begin{theorem}\label{theorem14}
	Let $\mathcal{L}= 2^\mathcal{T}-2^{\mathcal{T}-R}+\mathcal{L}'$, where $2\leq R \leq \mathcal{T}-1$ and $1\leq \mathcal{L}' \leq 2^{\mathcal{T}-R-1}$. Then, for $1\leq r \leq R$,  the Hamming distance of $C_{2^\mathcal{T}-2^{\mathcal{T}-r}}^\perp$  is given by
	\begin{equation*}
		d_{2^\mathcal{T}-2^{\mathcal{T}-r}}^\perp= 
		\min  \left\{  wt \left(b(x)\right)\mid b(x) \in B_{\mathcal{L},\mathcal{T}-r}^{'m}	\right\},
	\end{equation*}
	where the set $B_{\mathcal{L},\mathcal{T}-r}^{'m}$ is   defined as in Theorem~\ref{theorem13}.
\end{theorem}

Table~\ref{tab6} presents several optimal and almost optimal binary polycyclic codes associated with $P(x)^\mathcal{L}$ together with the parameters of their dual codes. These parameters are verified using Magma \cite{bosma1997magma}.
\begin{table}[htbp]
	\caption{Some optimal and almost optimal Binary Polycyclic codes and their duals }
	\label{tab6}
	\centering
	\begin{tabular}{|c|c|c|c|c|}
		\hline
		$P(x) $ & $ \mathcal{L} $ & $C$& Parameters of $C$ & Parameters of $C^\perp$ \\ \hline
		
		$x^3+x+1$ & $9$ & $\langle P(x) \rangle$ & $[27,24,2]^*$ & $[27,3,15]^*$\\ \hline
		$x^4+x+1$ & $22$ &  $\langle P(x) \rangle$ & $[88,84,2]^*$ & $[88,4,46]^*$\\\hline
		$x^4+x+1$ & $26$ &  $\langle P(x) \rangle$ & $[104,100,2]^*$ & $[104,4,55]^*$\\\hline
		$x^4+x+1$ & $45$ &  $\langle P(x) \rangle$ & $[180,176,2]^*$ & $[180,4,96]^*$\\\hline	
		$x^5+x^2+1$ & $13$ &  $\langle P(x) \rangle$ & $[65,60,2]^*$ & $[65,5,32]^*$\\\hline
		$x^5+x^2+1$ & $19$ &  $\langle P(x) \rangle$ & $[95,90,2]^*$ & $[95,5,48]^*$\\\hline			
		$x^5+x^2+1$ & $5$ & $\langle P(x) \rangle$ & $[25,20,3]^*$ &$[25,5,11]^{o}$\\ \hline
		$x^5+x^3+1$ & $6$ & $\langle P(x) \rangle$ & $[30,25,3]^*$ &  $[30,5,15]^*$\\ \hline
		$x^6+x+1$ & $6$ & $\langle P(x) \rangle$ & $[36,30,3]^*$ & $[36,6,15]^{o}$ \\\hline
		$x^6+x^5+x^3+x^2+1$ & $10$ & $\langle P(x) \rangle$ & $[60,54,3]^*$ & $[60,6,29]^{o}$\\ \hline
		$x^6+x^5+1$ & $11$ & $\langle P(x) \rangle$ & $[66,60,2]^*$ & $[66,6,32]^*$ \\\hline
		
		$x^7+x^6+x^3+x+1$ & $4$ & $\langle P(x) \rangle$ & $[28,21,3]^{o}$ & $[28,7,11]^{o}$\\\hline
		$x^7+x^6+x^5+x^4+x^2+x+1$ & $6$ & $\langle P(x) \rangle$ & $[42,35,3]^{o}$ & $[42,7,18]^{o}$ \\\hline
		$x^7+x^4+1$ & $11$ & $\langle P(x) \rangle$ & $[77,70,3]^*$ & $[77,7,34]$\\\hline
		$x^7+x^6+x^3+x^2+1$ & $13$ & $\langle P(x) \rangle$ &  $[91,84,3]^*$ & $[91,7,44]^*$ \\\hline
		$x^7+x^6+1$ & $18$ & $\langle P(x) \rangle$ & $[126,119,3]^*$ & $[126,7,63]^*$ \\\hline
		$x^7+x^6+1$ & $19$ & $\langle P(x) \rangle$  & $[133,126,2]^*$ & $[133,7,64]^*$\\\hline
		$x^8+x^7+x^2+x+1$ & $3$ & $\langle P(x) \rangle$ & $[24,16,4]^*$ & $[24,8,8]^*$\\\hline
		$x^8+x^6+x^5+x+1$ & $5$ & $\langle P(x) \rangle$ & $[40,32,3]^{o}$ & $[40,8,15]^{o}$\\\hline
		$x^8+x^7+x^6+x^5+x^2+x+1$ & $9$ & $\langle P(x) \rangle$ & $[72,64,3]^{o}$ & $[72,8,31]^{o}$\\\hline
		$x^9+x^8+x^7+x^6+x^5+x^3+1$ & $3$ & $\langle P(x) \rangle$ & $[27,18,4]^*$ & $[27,9,9]^{o}$\\\hline
		$x^9+x^8+x^6+x^5+x^4+x^3+x^2+x+1$ & $6$ & $\langle P(x) \rangle$ & $[54,45,4]^*$ & $[54,9,20]$\\\hline
		$x^{11}+x^{10}+x^5+x^4+1$ & $8$ & $\langle P(x) \rangle$ & $[88,77,4]^{*}$ & $[88,11,39]^{*}$\\\hline
	\end{tabular}
	\vspace{2mm}
	\begin{flushleft}
		\footnotesize{$^*$ denotes an optimal code and $^o$ denotes an almost optimal code.}
	\end{flushleft}
\end{table}

Hence, in this section, we have provided general methods and explicit results for the Hamming distances of the Euclidean duals of binary polycyclic codes, together with examples of several optimal and almost optimal codes.
\section{LCD Property of Binary Polycyclic Codes}\label{sec5}
A linear code $\mathcal{C}$ over the finite field $\mathbb{F}_q$ is said to be an LCD code if $\mathcal{C} \cap \mathcal{C}^\perp=\{0\}$. Let $LCD(n,k)$ denote the maximum possible Hamming distance of  a binary LCD code with length $n$ and dimension $k$. A binary LCD code with parameters $[n,k, LCD(n,k)]$ is said to be an LCD optimal code. 
Throughout this section, all notations remain the same as introduced earlier.
\subsection{Condition for a code to be LCD and families of binary LCD codes}
 Let $C_j=\langle P(x)^j\rangle$ be a binary polycyclic code associated with $P(x)^\mathcal{L}$, where $1\leq j \leq \mathcal{L}-1$, $\mathcal{L}\geq2.$ If $Ord(P(x))=e$ and $U(x)\in \mathbb{F}_2[x]$ such that $P(x)U(x)=x^{e}+1.$ Then a general codeword of $C_j$ has the form $P(x)^j \gamma(x)$, where $\gamma(x)\in \mathbb{F}_2[x]$ with $\deg(\gamma(x))< m(\mathcal{L}-j)$, while a general element of $C_j^\perp$ can be written as $(x^{e}+1)^{2^\mathcal{T}-j}U^*(x)^j \delta(x)\bmod (x^{m\mathcal{L}})$, where $\delta(x)\in \mathbb{F}_2[x]$ and $\deg(\delta(x))< mj.$ In this section, we establish necessary and sufficient condition for $C_j$ to be an LCD code.
\begin{theorem}\label{theorem15}
	For $1\leq j \leq 2^{\mathcal{T}-1}$, the code $C_j$ is an LCD code if and only if for every non-zero $\delta(x)\in \mathbb{F}_2[x]$ with $\deg(\delta(x))< mj$, we have 
	\begin{equation*}
		\deg\left( (x^{e}+1)^{2^\mathcal{T}-2j} (U(x)U^*(x))^j \delta(x) \bmod (x^{m\mathcal{L}}) \right) \geq m (\mathcal{L}-j).
	\end{equation*}
\end{theorem}
\begin{proof}
	Suppose $C_j$ is not an LCD code. Then there exists non-zero polynomials $\gamma(x), \delta(x)\in \mathbb{F}_2[x]$ with $\deg(\gamma(x))< m(\mathcal{L}-j)$ and $\deg(\delta(x))< mj$ such that \allowdisplaybreaks
	\begin{align}\label{eq45}
		\nonumber		P(x)^j \gamma(x) \equiv & (x^{e}+1)^{2^\mathcal{T}-j}U^*(x)^j \delta(x)\mod (x^{m\mathcal{L}})\\ \equiv
		& P(x)^{2^\mathcal{T}-j} U(x)^{2^\mathcal{T}-j} U^*(x)^j \delta(x)\mod (x^{m\mathcal{L}}).
	\end{align}
	Equivalently, 
	\begin{equation*}
		P(x)^j \left(\gamma(x)+P(x)^{2^\mathcal{T}-2j}U(x)^{2^\mathcal{T}-j}U^*(x)^j \delta(x) \right) \equiv 0 \mod (x^{m\mathcal{L}}).
	\end{equation*}
	Since $\gcd \left(x, P(x)\right)=1 $ in $\mathbb{F}_2[x]$, it follows that $\gcd \left(x^{m\mathcal{L}}, P(x)\right)=1.$ Therefore, \allowdisplaybreaks
	\begin{align}\label{eq46}
		\nonumber	\gamma(x) \equiv &  P(x)^{2^\mathcal{T}-2j}U(x)^{2^\mathcal{T}-j}U^*(x)^j \delta(x) \mod (x^{m\mathcal{L}})\\
		\equiv & (x^{e}+1)^{2^\mathcal{T}-2j} (U(x)U^*(x))^j \delta(x) \mod (x^{m\mathcal{L}}).
	\end{align} Since $\gamma(x)$ is non-zero and $\deg (\gamma(x))< m (\mathcal{L}-j)$, it follows that 
	\begin{equation}\label{eq47}
		\deg\left( (x^{e}+1)^{2^\mathcal{T}-2j} (U(x)U^*(x))^j \delta(x) \bmod (x^{m\mathcal{L}}) \right) < m (\mathcal{L}-j).
	\end{equation}
	Conversely, suppose there exists a non-zero $\delta(x)\in \mathbb{F}_2[x]$ with $\deg(\delta(x))< mj$ such that $\deg\Big( (x^{e}+1)^{2^\mathcal{T}-2j} (U(x)U^*(x))^j \delta(x)\\ \bmod (x^{m\mathcal{L}}) \Big) < m (\mathcal{L}-j)$. If we take $\gamma(x)=(x^{e}+1)^{2^\mathcal{T}-2j} (U(x)U^*(x))^j \delta(x) \bmod (x^{m\mathcal{L}})$, then $\gamma(x)$ satisfies the congruence~\eqref{eq45}. Hence, $C_j$ is not an LCD code.
	
	Therefore, $C_j$ is not an LCD code if and only if there exists non-zero $\delta(x)$ with $\deg(\delta(x))<mj$ satisfying~\eqref{eq47}. Equivalently, $C_j$ is an LCD code if and only if for every non-zero $\delta(x)$ with $\deg (\delta(x)) < mj$, we have \begin{equation*}
		\deg\left( (x^{e}+1)^{2^\mathcal{T}-2j} (U(x)U^*(x))^j \delta(x) \bmod (x^{m\mathcal{L}}) \right) \geq m (\mathcal{L}-j).
	\end{equation*} This completes the proof.
\end{proof} 
For the remaining case $2^{\mathcal{T}-1}< j < \mathcal{L}$, in a similar fashion, we obtain the following criterion.
\begin{theorem}\label{theorem16}
	For $2^{\mathcal{T}-1}< j < \mathcal{L}$, the code $C_j$ is an LCD code if and only if there does not exist non-zero polynomials $\delta(x), \gamma(x)\in \mathbb{F}_2[x]$ satisfying $\deg(\delta(x))< mj$, $\deg(\gamma(x))< m (\mathcal{L}-j)$, such that 
	\begin{equation*}
		P(x)^{2j-2^\mathcal{T}} \gamma(x) \equiv U(x)^{2^\mathcal{T}-j} U^*(x)^j \delta(x) \mod (x^{m\mathcal{L}}). 
	\end{equation*}
\end{theorem}
The following example illustrates Theorem~\ref{theorem15}.
\begin{example}\label{ex6}
	Consider the code $C_1=\langle x^3+x+1 \rangle$ associated with $(x^3+x+1)^8$. Here, $m=3$, $j=1$, $\mathcal{L}=8$, $\mathcal{T}=3$, and $Ord(x^3+x+1)=7$. Hence, $U(x)= x^4+x^2+x+1$, $U^*(x)= x^4+x^3+x^2+1$, implying $U(x)U^*(x)= (x^7+1)(x+1)$. By Theorem~\ref{theorem15}, $C_1$ is an LCD code if and only if for every non-zero $\delta(x)\in \mathbb{F}_2[x]$ satisfying $\deg (\delta(x))<3$,  we have
	\begin{equation*}
	\deg\left( (x^7+1)^{6} (x^7+1)(x+1) \delta(x) \ \bmod (x^{24})\right) \geq 21.
	\end{equation*}
Now, 
\begin{align*}
	(x^7+1)^{6} (x^7+1)(x+1) \delta(x) \ \bmod (x^{24}) \equiv{}& (x^{28}+1) (x^{14}+1) (x^7+1)(x+1) \delta(x) \ \bmod (x^{24})\\
	\equiv{}& (x^{14}+1) (x^7+1)(x+1) \delta(x) \ \bmod (x^{24})\\
	\equiv{}& \Big((x+1) \delta(x)+x^7 (x+1) \delta(x)+x^{14} (x+1) \delta(x)\Big)+x^{21}(x+1) \delta(x) \ \bmod (x^{24})
\end{align*}
Since for every such $\delta(x)$, $\deg \big(x^{21}(x+1) \delta(x)  \bmod (x^{24})\big) \geq21, $ it follows that $C_1$ is an LCD code. Moreover, by Corollary~\ref{corollary1} and Theorem~\ref{theorem12}, we obtain $d_1=2$ and $d_1^\perp=13.$ Therefore, the parameters of $C_1$ and $C_1^\perp$ are $[24, 21, 2]$ and $[24,3,13]$, respectively.  Both codes are optimal by \cite{grassl_codetables}.
\end{example}
\subsection{Bounds on $LCD(n,k)$ and LCD optimal codes }

In this subsection, we summarize several known results concerning the maximum possible Hamming distance $LCD(n,k)$ of binary LCD codes.
The exact values of $LCD(n,1)$ and $LCD(n,n-1)$ were determined in \cite{dougherty2017combinatorics}. For dimension $k=2$, the values of $LCD(n,2)$ for all $n \geq 2$ were completely characterized in \cite{galvez2018bounds}. In the same work, the authors also computed the exact values of $LCD(n,k)$ for $1 \leq k \leq n \leq 12$ (see Table~1 of \cite{galvez2018bounds}) and established that $LCD(n,n-i)=2$ for all $i \geq 2$ and $n \geq 2^i$. For dimension $k=3$, the exact values of $LCD(n,3)$ for all $n \geq 3$ were determined in \cite{harada2019binary}. For $k=4$, partial results were first obtained in \cite{araya2020minimum}, and later a complete characterization of $LCD(n,4)$ for all $n \geq 4$ was given in \cite{araya2021characterization}. For  $k=5$, bounds and exact values for specific lengths were studied in \cite{araya2020minimum}, while the exact value of $LCD(n,5)$ for all $n \geq 14$ were determined in \cite{liu2024minimum}.

Several results are also available for LCD codes of large dimension. The exact values of $LCD(n,n-5)$ for all $n \geq 6$ were obtained in \cite{araya2021minimum}. Furthermore, in \cite{araya2024characterizations}, the authors determined $LCD(n,n-6)$ for all $n \geq 7$, and $LCD(n,n-7)$ for all $n \geq 8$, except for $19$ exceptional values of $n$. For these remaining cases, the following bound was established:
\begin{equation*}
	3 \leq LCD(n,n-7) \leq 4 \quad \text{for } n \in \{34,36,38,40,42,44,46,48,50,52,54,56,58,59,60,61,62,63,64\}.
\end{equation*}
In addition, bounds for $LCD(n,k)$ with $k \leq 32$ and $n \leq 40$ were investigated in \cite{bouyuklieva2021optimal}. These bounds were later improved in \cite{ishizuka2023construction} for $26 \leq n \leq 40$, and subsequently refined in \cite{li2024several} for $29 \leq n \leq 40$. More recently, \cite{wang2024new} improved several known bounds for $38 \leq n \leq 40$ and $9 \leq k \leq 15$, and also obtained new bounds for $41 \leq n \leq 50$ and $6 \leq k \leq n-6$.

In Table~\ref{tab7}, we present several optimal and LCD optimal binary LCD polycyclic codes $C=\langle P(x) \rangle$ generated by a binary irreducible polynomial $P(x)$  and  associated with $P(x)^\mathcal{L}$ for $\mathcal{L}\geq2$. The parameters of these codes, along with  their duals, have been verified using Magma \cite{bosma1997magma}.
\begin{table}[htbp]
	\caption{Some optimal and LCD optimal binary LCD polycyclic codes $\langle P(x) \rangle$ and their duals }
	\label{tab7}
	\centering
	\begin{tabular}{|c|c|c|c|c|}
		\hline
		$P(x) $ & $ \mathcal{L} $ &  Parameters of $C$ & Parameters of $C^\perp$ & Reference  \\ \hline
		$x^3+x+1$ & $6$ &  $[18,15,2]^{*}$ & $[18,3,9]^\bullet$  &\cite{dougherty2017combinatorics}
		\\\hline
		$x^3+x+1$ & $8$ & $[24,21,2]^*$ & $[24,3,13]^*$& \cite{grassl_codetables}\\ \hline
		$x^3+x+1$ & $13$ & $[39,36,2]^*$ & $[39,3,21]^\bullet$ & \cite{dougherty2017combinatorics}\\ \hline		
		$x^3+x+1$ & $15$ & $[45,42,2]^*$ & $[45,3,25]^*$ & \cite{grassl_codetables} \\\hline
		$x^4+x+1$ & $16$  & $[64,60,2]^*$ & $[64,4,33]^*$ & \cite{grassl_codetables}\\ \hline	
		$x^4+x+1$ & $17$  & $[68,64,2]^*$ & $[68,4,35]^*$&\cite{grassl_codetables}\\ \hline	
		$x^4+x+1$ & $31$  & $[124,120,2]^*$ & $[124,4,65]^*$&\cite{grassl_codetables}\\ \hline
		$x^4+x+1$ & $40$  & $[160,156,2]^*$ & $[160,4,84]^*$&\cite{grassl_codetables} \\ \hline
		$x^4+x+1$ & $46$  & $[184,180,2]^*$ & $[184,4,97]^*$&\cite{grassl_codetables} \\ \hline						
		$x^5+x^2+1$ & $3$  & $[15,10,3]^\bullet$&$[15,5,5]$& \cite{araya2021minimum} \\ \hline
		$x^5+x^2+1$ & $5$  & $[25,20,3]^*$ &$[25,5,11]^{\bullet}$  & \cite{bouyuklieva2021optimal} \\ \hline
		$x^5+x^4+x^2+x+1$ & $8$  & $[40,35,2]^*$ &  $[40,5,19]^\bullet$& \cite{bouyuklieva2021optimal}\\ \hline	
		$x^5+x^2+1$ & $9$ & $[45,40,2]^*$ &$[45,5,21]^{\bullet}$& \cite{liu2024minimum} \\ \hline	
		$x^5+x^2+1$ & $11$  & $[55,50,2]^*$ &$[55,5,27]^{\bullet}$ &\cite{liu2024minimum} \\ \hline		
		$x^5+x^4+x^2+x+1$ & $15$  & $[75,70,2]^*$ &  $[75,5,37]^\bullet$ & \cite{liu2024minimum}\\ \hline
		$x^5+x^2+1$ & $32$  & $[160,155,2]^*$ &  $[160,5,81]^*$ & \cite{grassl_codetables} \\ \hline
		$x^6+x^5+x^4+x+1$ & $5$ &  $[30,24,3]^\bullet$ & $[30,6,12]$ & \cite{araya2024characterizations} \\ \hline
		$x^6+x^5+x^3+x^2+1$ & $7$  & $[42,36,3]^*$ & $[42,6,18]$& \cite{grassl_codetables}\\ \hline
		$x^6+x^5+1$ & $9$  & $[54,48,3]^*$ & $[54,6,25]^o$ & \cite{grassl_codetables} \\\hline
		
		$x^6+x^5+x^3+x^2+1$ & $12$  & $[72,66,2]^*$ & $[72,6,34]^{o}$ & \cite{grassl_codetables}\\ \hline
		$x^6+x^5+x^3+x^2+1$ & $17$ & $[102,96,2]^*$ & $[102,6,49]^{o}$& \cite{grassl_codetables}\\ \hline
		$x^6+x^5+x^3+x^2+1$ & $19$ &  $[114,108,2]^*$ & $[114,6,55]^{o}$ & \cite{grassl_codetables} \\ \hline	
		$x^6+x^5+1$ & $22$ & $[132,126,2]^*$  & $[132,6,65]^*$& \cite{grassl_codetables}\\\hline
		$x^7+x^6+x^3+x+1$ & $10$  & $[70,63,3]^{*}$ & $[70,7,30]$& \cite{grassl_codetables}\\\hline
		$x^7+x^4+1$ & $11$ &  $[77,70,3]^*$ & $[77,7,34]$& \cite{grassl_codetables}\\\hline		
		$x^7+x^6+x^5+x^4+1$ & $16$  & $[112,105,3]^{*}$ & $[112,7,52]$ & \cite{grassl_codetables} \\\hline
		$x^8+x^7+x^2+x+1$ & $3$ & $[24,16,4]^*$ & $[24,8,8]^*$ & \cite{grassl_codetables}\\\hline
		$x^8+x^7+x^2+x+1$ & $7$  & $[56,48,3]^o$ & $[56,8,23]^o$  & \cite{grassl_codetables}\\\hline
		$x^9+x^7+x^2+x+1$ & $2$ & $[18,9,5]^{\bullet}$ & $[18,9,4]$ &\cite{bouyuklieva2021optimal}\\\hline
		$x^9+x^8+x^7+x^6+x^5+x^3+1$ & $3$ &  $[27,18,4]^*$ & $[27,9,9]^{\bullet}$& \cite{bouyuklieva2021optimal}\\\hline
		$x^9+x^6+x^4+x^3+1$ & $4$ & $[36,27,4]^*$ & $[36,9,12]$& \cite{grassl_codetables} \\\hline
		$x^{10}+x^6+x^2+x+1$ & $5$  & $[50,40,4]^*$ & $[50,10,16]$&\cite{grassl_codetables}\\\hline
		$x^{10}+x^9+x^8+x^7+x^5+x^4+1$ & $6$ & $[60,50,3]^{o}$ & $[60,10,22]$& \cite{grassl_codetables}\\\hline
		$x^{11}+x^{10}+x^8+x^6+1$ & $5$ & $[55,44,4]^{bklc}$ & $[55,11,19]$&\cite{grassl_codetables}\\\hline
		$x^{11}+x^{10}+x^5+x^4+1$ & $7$  & $[77,66,4]^{*}$ & $[77,11,30]$&\cite{grassl_codetables}\\\hline
		$x^{11}+x^{10}+x^5+x^4+1$ & $8$  & $[88,77,4]^{*}$ & $[88,11,39]^{*}$&\cite{grassl_codetables}\\\hline
		$x^{12}+x^{11}+x^9+x^7+x^6+x^4+1$ & $5$ & $[60,48,4]^{o}$ & $[60,12,20]$& \cite{grassl_codetables}	\\\hline
		$x^{13}+x^{12}+x^{10}+x^8+x^6+x^4+x^3+x^2+1$ & $4$ & $[52,39,5]^{o} $ & $[52,13,15]$&\cite{grassl_codetables}\\\hline
		$x^{14}+x^{13}+x^{11}+x^6+x^5+x^4+x^2+x+1$ & $4$ & $[56,42,5]^{o}$ & $[56,14,17]$&\cite{grassl_codetables}\\\hline
		$x^{15}+x^7+x^6+x^3+x^2+x+1$ & $4$ & $[60,45,5]^{o}$ & $[60,15,15]$& \cite{grassl_codetables} \\\hline
		$x^{17}+x^8+x^7+x^6+x^4+x^3+1$ & $2$ & $[34,17,7]^{o}$ & $[34,17,5]$&\cite{grassl_codetables}\\\hline
	\end{tabular}
	\vspace{2mm}
	\begin{flushleft}
		\footnotesize{$^*$ denotes an optimal code, $^{o}$ denotes an almost optimal code, $^{bklc}$ denotes a best-known linear code  and $^{\bullet}$ denotes an LCD optimal code.}
	\end{flushleft}
\end{table}

\section{Polycyclic codes associated with powers of self-reciprocal trinomials}\label{sec6}
It is well known from \cite{lidl1997finite} that the trinomial $x^{2n}+x^n+1$ is irreducible over $\mathbb{F}_2$ if and only if $n=3^v$ for some integer $v \geq0$. Equivalently, every self-reciprocal binary irreducible trinomial is of the form $ x^{2\cdot3^v}+x^{3^v}+1$, where $v \geq0.$ Motivated by this characterization, we study polycyclic codes associated with powers of such trinomials. Since these polynomials are self-reciprocal, the corresponding codes possess reversibility properties. This additional structure also allows us to derive explicit formulas for their Hamming distance.

From the general distance results obtained in Section~\ref{sec3}, for an arbitrary binary irreducible polynomial $P(x)$, the Hamming distance of the code $C_j= \langle P(x)^j \rangle$ depends on the Hamming weights of certain polynomial sets such as
$S_{\mathcal{T},s}^n$ and $S_{\mathcal{T},r}^{'n}$.   Therefore, for the family of trinomials considered here, these sets can be analyzed explicitly. This reduces the problem to studying the binary expansion of $\big(x^{2n}+x^n+1\big)^{2^r-1}$, where $r\geq2$ and $n\geq1$. Since this expression is obtained from $\left(x^2+x+1\right)^{2^r-1}$ by replacing $x$ with $x^n$, it is sufficient to compute the expansion of $\left(x^2+x+1\right)^{2^r-1}$.
\begin{lemma}\label{lemma1}
	For $r\geq 2$,	the binary expansion of $\left(x^2+x+1\right)^{2^r-1}$  is given by
	\begin{equation*}
		\left(x^2+x+1\right)^{2^r-1}	= \begin{cases}
			\displaystyle
			\sum_{j=0}^{\frac{2^r-1}{3}-1} \left(x^{3j}+x^{3j+1}\right)+x^{2^r-1}+ \sum_{j=\frac{2^r-1}{3}}^{\frac{2(2^r-1)}{3}-1} \left(x^{3j+2}+x^{3j+3}\right), & \text{if } r \text{ is even},\\
			\displaystyle	\sum_{j=0}^{\frac{2^r-2}{3}-1} \left(x^{3j}+x^{3j+1}\right)+\big(x^{2^r-2}+x^{2^r-1}+x^{2^r}\big)+ \sum_{j=\frac{2^r+1}{3}}^{\frac{2(2^r-2)}{3}} \left(x^{3j+1}+x^{3j+2}\right), &\text{if } r  \text{ is odd}.
		\end{cases}
	\end{equation*}
\end{lemma}
\begin{proof}
		Let 
	\begin{equation*}
		\left(x^2+x+1\right)^{2^r-1}= b_0+b_1x+b_2x^2+\ldots+b_{2\left(2^r-1\right)} x^{2\left(2^r-1\right)}, \quad b_i \in \mathbb{F}_2.
	\end{equation*} 
	Then
	\begin{equation}\label{eq48}
		\left(  b_0+b_1x+b_2x^2+\ldots+b_{2\left(2^r-1\right)} x^{2\left(2^r-1\right)} \right) \left(1+x+x^2\right)= 1+x^{2^r}+x^{2^{r+1}}.
	\end{equation}
	By comparing coefficients of identical powers of $x$ in Eq.~\eqref{eq48}, we obtain the following system of equations:\allowdisplaybreaks
	\begin{align}\label{eq49}
		\nonumber	b_0=1,\quad b_0 + b_1=0,\quad b_0 + b_1 + b_2=0,\quad b_1 + b_2 + b_3=0,\quad 	b_2 + b_3 + b_4=0,&\ldots,\\
		b_{2^r-3} + b_{2^r-2} + b_{2^r-1}=0,\quad 	b_{2^r-2} + b_{2^r-1} + b_{2^r}=1,\quad  	b_{2^r-1} + b_{2^r} + b_{2^r+1}=0,\quad 	b_{2^r} + b_{2^r+1} + b_{2^r+2}=0,&\ldots,\\
		\nonumber	b_{2^r+2^r-4} + b_{2^r+2^r-3} + b_{2^r+2^r-2} =0,\quad 	b_{2^r+2^r-3} + b_{2^r+2^r-2}=0,\quad b_{2^r+2^r-2}=1.&
	\end{align}
Since $x^2+x+1$ is a self-reciprocal polynomial, thus $\left(x^2+x+1\right)^{2^r-1}$ is also self-reciprocal. Therefore,
\begin{equation*}
	b_i= b_{2\left(2^r-1\right)-i},\quad 0\leq i \leq 2\left(2^r-1\right).
\end{equation*}
Thus, it suffices to solve \eqref{eq49} for $b_0,b_1,\ldots,b_{2^r-1}$, and the remaining coefficients  can be computed using the relation  $b_{2^r+l}= b_{2^r-2-l}$, where  $0\leq l \leq 2^r-2$.
Solving the first $2^r$ equations of \eqref{eq49}, we obtain 
\begin{equation*}
	b_0=1,\ \quad b_1=1,\quad b_2=0,\quad b_3=1,\quad b_4=1,\quad b_5=0,\quad b_6=1,\quad b_7=1,\quad b_8=0,\ldots\ldots,
\end{equation*}
and the above sequences continues in the same pattern till $b_{2^r-1}$.
We now consider two cases:\\
\textbf{Case I:} If $r$ is even. Then $3 \mid 2^r-1$. Therefore,  we obtain
\begin{equation*}
	\left(b_0,b_1,b_2,\ldots, b_{2^r-4}, b_{2^r-3}, b_{2^r-2}\right) = \left( 1,1,0,\ldots, 1,1,0  \right) \text{ and }  b_{2^r-1}=1.
\end{equation*}
Hence,
\begin{align*}
	b_0+b_1x+b_2x^2\ldots+b_{2^r-4}x^{2^r-4}+b_{2^r-3}x^{2^r-3}+b_{2^r-2}x^{2^r-2}
	= & \sum_{j=0}^{\frac{2^r-1}{3}-1} \left(x^{3j}+x^{3j+1}\right).
\end{align*}
Since  $b_{2^r+l}= b_{2^r-2-l}$ for  $0\leq l \leq 2^r-2$, thus 
\begin{align*}
	\left(b_{2^r}, b_{2^r+1}, b_{2^r+2}, \ldots, b_{2^r+2^r-4}, b_{2^r+2^r-3}, b_{2^r+2^r-2}\right)&= \left(b_{2^r-2}, b_{2^r-3}, b_{2^r-4},\ldots, b_2,b_1,b_0\right)= \left(0,1,1,\ldots,0,1,1\right).
\end{align*} 
Consequently,
\begin{align*}
	&	b_{2^r}x^{2^r}+ b_{2^r+1}x^{2^r+1}+b_{2^r+2}x^{2^r+2}+\ldots+b_{2(2^r-1)-2} x^{2(2^r-1)-2}+ b_{2(2^r-1)-1}x^{2(2^r-1)-1}+ b_{2(2^r-1)} x^{2(2^r-1)}\\
	=& \sum_{j=\frac{2^r-1}{3}}^{\frac{2(2^r-1)}{3}-1} \left(x^{3j+2}+x^{3j+3}\right).
\end{align*}
Therefore, \allowdisplaybreaks
\begin{align*}
	\left(x^2+x+1\right)^{2^r-1}= & b_0+b_1x+b_2x^2+\ldots+b_{2\left(2^r-1\right)} x^{2\left(2^r-1\right)}\\
	=	& b_0+b_1x+\ldots+b_{2^r-2}x^{2^r-2} + x^{2^r-1}+ b_{2^r}x^{2^r}+b_{2^r+1}x^{2^r+1}+\ldots+ b_{2(2^r-1)}x^{2(2^r-1)}\\
	=& \sum_{j=0}^{\frac{2^r-1}{3}-1} \left(x^{3j}+x^{3j+1}\right)+x^{2^r-1}+ \sum_{j=\frac{2^r-1}{3}}^{\frac{2(2^r-1)}{3}-1} \left(x^{3j+2}+x^{3j+3}\right).
\end{align*}
\textbf{Case II:} If $r$ is odd. Then $3 \mid 2^r-2$. Solving the first $2^r$ equations of \eqref{eq49}, we get
	\begin{equation*}
	\left(b_0,b_1,b_2,\ldots,b_{2^r-5}, b_{2^r-4}, b_{2^r-3}\right) = \left( 1,1,0,\ldots, 1,1,0  \right),\  b_{2^r-2}=1 \text{ and }  b_{2^r-1}=1.
\end{equation*}
Hence,
\begin{align*}
	b_0+b_1x+b_2x^2\ldots+b_{2^r-5}x^{2^r-5}+b_{2^r-4}x^{2^r-4}+b_{2^r-3}x^{2^r-3}+b_{2^r-2}x^{2^r-2}
	= & \sum_{j=0}^{\frac{2^r-2}{3}-1} \left(x^{3j}+x^{3j+1}\right)+x^{2^r-2}.
\end{align*}
Since  $b_{2^r+l}= b_{2^r-2-l}$ for  $0\leq l \leq 2^r-2$, thus 
\begin{align*}
	\left(b_{2^r}, b_{2^r+1}, b_{2^r+2}, \ldots, b_{2^r+2^r-4}, b_{2^r+2^r-3}, b_{2^r+2^r-2}\right)&= \left(b_{2^r-2}, b_{2^r-3}, b_{2^r-4},\ldots, b_2,b_1,b_0\right)= \left(1,0,1,1,\ldots,0,1,1\right).
\end{align*} 
Therefore,
\begin{align*}
	&	b_{2^r}x^{2^r}+ b_{2^r+1}x^{2^r+1}+b_{2^r+2}x^{2^r+2}+\ldots+b_{2(2^r-1)-2} x^{2(2^r-1)-2}+ b_{2(2^r-1)-1}x^{2(2^r-1)-1}+ b_{2(2^r-1)} x^{2(2^r-1)}\\
	=& x^{2^r}+ \sum_{j=\frac{2^r+1}{3}}^{\frac{2(2^r-2)}{3}} \left(x^{3j+1}+x^{3j+2}\right).
\end{align*}
Thus, we have \allowdisplaybreaks
\begin{align*}
	\left(x^2+x+1\right)^{2^r-1}= & b_0+b_1x+b_2x^2+\ldots+b_{2\left(2^r-1\right)} x^{2\left(2^r-1\right)}\\
	=	& b_0+b_1x+\ldots+b_{2^r-2}x^{2^r-2} + x^{2^r-1}+ b_{2^r}x^{2^r}+b_{2^r+1}x^{2^r+1}+\ldots+ b_{2(2^r-1)}x^{2(2^r-1)}\\
	=& \sum_{j=0}^{\frac{2^r-2}{3}-1} \left(x^{3j}+x^{3j+1}\right)+\left(x^{2^r-2}+x^{2^r-1}+x^{2^r}\right)+ \sum_{j=\frac{2^r+1}{3}}^{\frac{2(2^r-2)}{3}} \left(x^{3j+1}+x^{3j+2}\right).
\end{align*}
This completes the proof.
\end{proof}
Replacing $x$ by $x^n$ in Lemma~\ref{lemma1}, we get the following result.
\begin{theorem}\label{theorem17}
	For $n\geq1$ and $ r \geq2$, the binary expansion of $\left(x^{2n}+x^n+1\right)^{2^r-1}$ is given by
	\begin{equation*}
		\left(x^{2n}+x^n+1\right)^{2^r-1}	= \begin{cases}
			\displaystyle
			\sum_{j=0}^{\frac{2^r-1}{3}-1} \big(x^{n(3j)}+x^{n(3j+1)}\big)+x^{n(2^r-1)}+ \sum_{j=\frac{2^r-1}{3}}^{\frac{2(2^r-1)}{3}-1} \big(x^{n(3j+2)}+x^{n(3j+3)}\big), & \text{if } r \text{ is even},\\
			\displaystyle	\sum_{j=0}^{\frac{2^r-2}{3}-1} \big(x^{n(3j)}+x^{n(3j+1)}\big)+\big(x^{n(2^r-2)}+x^{n(2^r-1)}+x^{n(2^r)}\big)+ &\\ \displaystyle \sum_{j=\frac{2^r+1}{3}}^{\frac{2(2^r-2)}{3}} \big(x^{n(3j+1)}+x^{n(3j+2)}\big), & \text{if } r  \text{ is odd}.
		\end{cases}
	\end{equation*}
\end{theorem}
We now determine the Hamming distance of binary polycyclic codes associated with $\big( x^{2\cdot3^v}+x^{3^v}+1  \big)^\mathcal{L}$, where $v \geq 0$ and  $ \mathcal{L} \geq 2$.
Let $C_j= \langle\left( x^{2\cdot3^v}+x^{3^v}+1 \right)^j \rangle$ be the polycyclic code associated with $\big( x^{2\cdot3^v}+x^{3^v}+1  \big)^\mathcal{L}$ for $1 \leq j \leq \mathcal{L}-1$, and let $d_j$ denote its Hamming distance.  
We recall that $\mathcal{T}$ is the unique positive integer satisfying $2^{\mathcal{T}-1} < \mathcal{L} \leq 2^\mathcal{T}$. Also, we recall from Corollary~\ref{corollary3} that $	Ord \big(x^{2\cdot3^v}+x^{3^v}+1\big) = 3^{v+1}$ for every  $ v \geq0.$
Using Corollary~\ref{corollary1}, we obtain the following result.
\begin{theorem}\label{theorem18}
	Let $C_j= \langle \left(  x^{2\cdot3^v}+x^{3^v}+1 \right)^j \rangle$ be the polycyclic code associated with $\big( x^{2\cdot3^v}+x^{3^v}+1  \big)^\mathcal{L}$, where $1 \leq j \leq 2^{\mathcal{T}-1}$. Then the Hamming distance of $C_j$ is given by
	\begin{equation*}
		d_j= \begin{cases}
			2 & \text{for } 1 \leq j \leq 2^{\mathcal{T}-\mathcal{J}},\\
			3 &\text{for } 2^{\mathcal{T}-\mathcal{J}}+1 \leq j \leq 2^{\mathcal{T}-1},
		\end{cases}
	\end{equation*}
	where $\mathcal{J}$ is the smallest positive integer such that $3^{v+1} 2^{\mathcal{T}-\mathcal{J}}< 2\cdot3^v \mathcal{L}$.
	\\
	
	Next, we determine $d_j$ for $ j> 2^{\mathcal{T}-1}$ for different choices of $\mathcal{L}$.
	\subsection{$d_j$ for $2^{\mathcal{T}-1} < j < \mathcal{L}$, when $\mathcal{L}=2^\mathcal{T}$}

	We first consider the codes  $C_{2^{\mathcal{T}}-2^{\mathcal{T}-r}}$ for $2\leq r \leq \mathcal{T}$. 
	By Theorem~\ref{theorem5}, determining their Hamming distance requires the minimum weight of the set  $S_{\mathcal{T},r}^{'2 \cdot3^v}$, where
	\allowdisplaybreaks
	\begin{align*}
		S_{\mathcal{T},r}^{'2\cdot3^v}=  \Bigl\{ &  g(x)^{2^{\mathcal{T}-r}} \big( x^{2\cdot3^v}+x^{3^v}+1\big)^{2^{\mathcal{T}}-2^{\mathcal{T}-r}}\ \big|\ g(x)\in \mathbb{F}_2[x],\ \deg (g(x)) \leq 2\cdot 3^v-1,\  g(x) \text{ has constant term  }  1 \Bigr\}.
	\end{align*}
	Let 
	\begin{equation*}
		F_{\mathcal{T},r}^{'2\cdot3^v}= \left\{ g(x)\big( x^{2\cdot3^v}+x^{3^v}+1\big)^{2^r-1} \ \big|\ g(x)\in \mathbb{F}_2[x],\ \deg (g(x)) \leq 2\cdot 3^v-1,\ g(x) \text{ has constant term  }  1 \right\}.
	\end{equation*}
	Then 
	\begin{equation*}
		S_{\mathcal{T},r}^{'2\cdot3^v}= \left\{ f(x)^{2^{\mathcal{T}-r}} \mid f(x)\in F_{\mathcal{T},r}^{'2\cdot3^v} \right\}.
	\end{equation*}Next, we determine the weights of the polynomials in the set $F_{\mathcal{T},r}^{'2\cdot3^v}$. Since each element of $S_{\mathcal{T},r}^{'2\cdot3^v}$ is simply the $ 2^{\mathcal{T}-r}$-th power of an element in  $F_{\mathcal{T},r}^{'2\cdot3^v}$, the weights of the elements of the set  $F_{\mathcal{T},r}^{'2\cdot3^v}$ fully determine the weights of the elements of the set  $S_{\mathcal{T},r}^{'2\cdot3^v}$. From Theorem~\ref{theorem17}, for $r\geq 2$, we have
\begin{equation}\label{eq50}
	\big( x^{2\cdot3^v}+x^{3^v}+1\big)^{2^r-1}	=	\begin{cases}
		\displaystyle
		\sum_{j=0}^{\frac{2^r-1}{3}-1} \big(x^{3^{v+1}j}+x^{3^v(3j+1)}\big)+x^{3^v(2^r-1)}+ \sum_{j=\frac{2^r-1}{3}}^{\frac{2(2^r-1)}{3}-1} \big(x^{3^v(3j+2)}+x^{3^v(3j+3)}\big), & \text{if } r \text{ is even},\\
		\displaystyle	\sum_{j=0}^{\frac{2^r-2}{3}-1} \big(x^{3^{v+1}j}+x^{3^v(3j+1)}\big)+\big(x^{3^v(2^r-2)}+x^{3^v(2^r-1)}+x^{3^v(2^r)}\big) & \\+ \sum_{j=\frac{2^r+1}{3}}^{\frac{2(2^r-2)}{3}} \big(x^{3^v(3j+1)}+x^{3^v(3j+2)}\big),  & \text{if } r  \text{ is odd}.
	\end{cases}
\end{equation}
Consequently, 
\begin{equation*}
\hspace{-9.5cm}	(1+x^{3^v}) 	\big( x^{2\cdot3^v}+x^{3^v}+1\big)^{2^r-1}
\end{equation*}
\begin{equation}\label{eq51}
	=\begin{cases}
		\displaystyle
		\sum_{j=0}^{\frac{2^r-1}{3}-1} \big(x^{3^{v+1}j}+x^{3^v(3j+2)}\big)+x^{3^v(2^r-1)}+ x^{3^v(2^r)} +   \sum_{j=\frac{2^r-1}{3}}^{\frac{2(2^r-1)}{3}-1} \big(x^{3^v(3j+2)}+x^{3^v(3j+4)}\big), & \text{if } r \text{ is even}, \\
		\displaystyle	\sum_{j=0}^{\frac{2^r-2}{3}-1} \big(x^{3^{v+1}j}+x^{3^v(3j+2)}\big)+x^{3^v(2^r-2)}+x^{3^v(2^r+1)}+\sum_{j=\frac{2^r+1}{3}}^{\frac{2(2^r-2)}{3}} \big(x^{3^v(3j+1)}+x^{3^v(3j+3)}\big), & \text{if } r  \text{ is odd}.
	\end{cases}
\end{equation}
Let $f(x)$ be an element of $F_{\mathcal{T},r}^{'2\cdot3^v}$. Then it can be expressed as
\begin{equation*}
	f(x)= \big( x^{2\cdot3^v}+x^{3^v}+1\big)^{2^r-1} g(x),
\end{equation*}
where $g(x)= \sum_{i=0}^{2 \cdot3^v-1} g_ix^{i}$ with $g_i \in \mathbb{F}_2$ for each $i$ and $g_0=1$. The polynomial $f(x)$ can be rewritten as
\allowdisplaybreaks
\begin{align*}
	f(x) &= \big( x^{2\cdot3^v}+x^{3^v}+1\big)^{2^r-1}  \sum_{i=0}^{2 \cdot3^v-1} g_ix^{i}= \big( x^{2\cdot3^v}+x^{3^v}+1\big)^{2^r-1} \left(\sum_{i=0}^{3^v-1} g_ix^{i}+ \sum_{i=3^v}^{2 \cdot3^v-1}g_ix^{i}\right)\\
	&= \big( x^{2\cdot3^v}+x^{3^v}+1\big)^{2^r-1} \left(\sum_{i=0}^{3^v-1} g_ix^{i}+ x^{3^v} \sum_{i=0}^{3^v-1} g_{3^v+i}x^{i}\right)= \sum_{i=0}^{3^v-1} \left(   \big(g_i+x^{3^v}g_{3^v+i}\big) \big( x^{2\cdot3^v}+x^{3^v}+1\big)^{2^r-1} x^{i}  \right).
\end{align*} From Eqs.~\eqref{eq50} and \eqref{eq51}, it follows that every exponent of $x$ appearing in 	$\left( x^{2\cdot3^v}+x^{3^v}+1\right)^{2^r-1}$ or in $	\left(1+x^{3^v}\right)$$ \big( x^{2\cdot3^v}+x^{3^v}+1\big)^{2^r-1}$, is a multiple of $3^v$. Thus, we have 
\begin{equation*}
wt(f(x))= \sum_{i=0}^{3^v-1} wt \left( \big(g_i+x^{3^v}g_{3^v+i}\big) \big( x^{2\cdot3^v}+x^{3^v}+1\big)^{2^r-1} \right).
\end{equation*}
Now, we have four possibilities of $ \left(g_i, g_{3^v+i}\right)$.
\begin{itemize}
	\item If $g_i=g_{3^v+i}=0$, then $wt \left( \left(g_i+x^{3^v}g_{3^v+i}\right) \left( x^{2\cdot3^v}+x^{3^v}+1\right)^{2^r-1} \right)=0.$
	\item If $g_i=1,\ g_{3^v+i}=0$, then $wt \left( \left(g_i+x^{3^v}g_{3^v+i}\right) \left( x^{2\cdot3^v}+x^{3^v}+1\right)^{2^r-1} \right)= wt \left(  \left( x^{2\cdot3^v}+x^{3^v}+1\right)^{2^r-1} \right).$
	\item If $g_i=0,\ g_{3^v+i}=1$, then $wt \left( \left(g_i+x^{3^v}g_{3^v+i}\right) \left( x^{2\cdot3^v}+x^{3^v}+1\right)^{2^r-1} \right)= wt \left(  \left( x^{2\cdot3^v}+x^{3^v}+1\right)^{2^r-1} \right).$
	\item If $g_i=g_{3^v+i}=1$, then $ wt \Big( \big(g_i+x^{3^v}g_{3^v+i}\big) \big( x^{2\cdot3^v}+x^{3^v}+1\big)^{2^r-1} \Big)= wt \Big( \big(1+x^{3^v}\big) \big( x^{2\cdot3^v}+x^{3^v}+1\big)^{2^r-1} \Big)$.
\end{itemize}
Suppose there are $\alpha$ values of $i$ for which $g_i=1,\ g_{3^v+i}=0$ or $g_i=0,\ g_{3^v+i}=1$ and $\beta$ values of $i$ for which $g_i=g_{3^v+i}=1$. Then
\begin{equation*}
	wt(f(x))= \alpha wt \left(  \big( x^{2\cdot3^v}+x^{3^v}+1\big)^{2^r-1} \right)+ \beta  wt \left( \big(1+x^{3^v}\big) \big( x^{2\cdot3^v}+x^{3^v}+1\big)^{2^r-1} \right).
\end{equation*}Therefore, weights of all the polynomials in the set  $F_{\mathcal{T},r}^{'2\cdot3^v}$ can be expressed as a linear combination of  $wt \Big(  \big( x^{2\cdot3^v}+x^{3^v}+1\big)^{2^r-1} \Big)$ and $wt \Big( \big(1+x^{3^v}\big) \big( x^{2\cdot3^v}+x^{3^v}+1\big)^{2^r-1} \Big).$ Moreover, since $\big( x^{2\cdot3^v}+x^{3^v}+1\big)^{2^r-1},\  \big(1+x^{3^v}\big) \big( x^{2\cdot3^v}+x^{3^v}+1\big)^{2^r-1} \in F_{\mathcal{T},r}^{'2\cdot3^v}$,  we conclude that 
\allowdisplaybreaks
\begin{align}\label{eq52}
	\nonumber	\min \left\{ wt(s(x)) \mid s(x)\in S_{\mathcal{T},r}^{'2\cdot3^v} \right\} &= \min \left\{ wt(f(x)) \mid f(x)\in F_{\mathcal{T},r}^{'2\cdot3^v} \right\}\\
	&= \min \Big\{ wt \Big(  \big( x^{2\cdot3^v}+x^{3^v}+1\big)^{2^r-1} \Big),
	wt \Big( \big(1+x^{3^v}\big) \big( x^{2\cdot3^v}+x^{3^v}+1\big)^{2^r-1} \Big) \Big\}.
\end{align}
Table~\ref{tab8} illustrate the above discussion for $v=1$ and  groups the elements of $F_{\mathcal{T},r}^{'6}$ according to their weights.
\begin{table}[htbp]
	\caption{Weights of polynomials in the set $F_{\mathcal{T},r}^{'6}$}
	\label{tab8}
	\centering
	\begin{tabular}{|c|c|}
		\hline
		Polynomials & Weight\\\hline
		$\left(x^{6}+x^3+1\right)^{2^r-1}$ &	$wt \left(\left(x^{6}+x^3+1\right)^{2^r-1}\right)$   \\\hline
		$\left(1+x^3\right) \left(x^{6}+x^3+1\right)^{2^r-1} $ &	$wt \left( \left(1+x^3\right) \left(x^{6}+x^3+1\right)^{2^r-1} \right)$   \\\hline
		\begin{tabular}{@{}c@{}} $(1+x) \left(x^{6}+x^3+1\right)^{2^r-1},\ \left(1+x^2\right) \left(x^{6}+x^3+1\right)^{2^r-1}, $\\$ \left(1+x^4\right) \left(x^{6}+x^3+1\right)^{2^r-1},\ \left(1+x^5\right) \left(x^{6}+x^3+1\right)^{2^r-1}$
		\end{tabular} & $2 wt \left(\left(x^{6}+x^3+1\right)^{2^r-1}\right)$ \\\hline
		
		\begin{tabular}{@{}c@{}}
			$\left(1+x+x^2\right) \left(x^{6}+x^3+1\right)^{2^r-1},\ \left(1+x+x^5\right) \left(x^{6}+x^3+1\right)^{2^r-1},$ \\$ \left(1+x^2+x^4\right)
			\left(x^{6}+x^3+1\right)^{2^r-1},\  \left(1+x^4+x^5\right) \left(x^{6}+x^3+1\right)^{2^r-1} $
		\end{tabular} & $3  wt \left(\left(x^{6}+x^3+1\right)^{2^r-1}\right)$\\\hline 
		\begin{tabular}{@{}c@{}} 
			$ \left(1+x+x^3\right) \left(x^{6}+x^3+1\right)^{2^r-1},\ \left(1+x+x^4\right) \left(x^{6}+x^3+1\right)^{2^r-1},$ \\ $ \left(1+x^2+x^3\right)\left(x^{6}+x^3+1\right)^{2^r-1},\ \left(1+x^2+x^5\right) \left(x^{6}+x^3+1\right)^{2^r-1},$ \\$ \left(1+x^3+x^4\right) \left(x^{6}+x^3+1\right)^{2^r-1},\ \left(1+x^3+x^5\right) \left(x^{6}+x^3+1\right)^{2^r-1}$
		\end{tabular} & 	\begin{tabular}{@{}c@{}} $wt \left(\left(x^{6}+x^3+1\right)^{2^r-1}\right)+$\\ $ wt \left( \left(1+x^3\right) \left(x^{6}+x^3+1\right)^{2^r-1} \right) $ \end{tabular} \\\hline 
		\begin{tabular}{@{}c@{}} 
			$ \left(1+x+x^2+x^3\right) \left(x^{6}+x^3+1\right)^{2^r-1},\ \left(1+x+x^2+x^4\right) \left(x^{6}+x^3+1\right)^{2^r-1},$ \\ $ \left(1+x+x^2+x^5\right) \left(x^{6}+x^3+1\right)^{2^r-1},\ \left(1+x+x^3+x^5\right) \left(x^{6}+x^3+1\right)^{2^r-1},$ \\ $ \left(1+x+x^4+x^5\right) \left(x^{6}+x^3+1\right)^{2^r-1},\ \left(1+x^2+x^3+x^4\right) \left(x^{6}+x^3+1\right)^{2^r-1},$ \\ $ \left(1+x^2+x^4+x^5\right) \left(x^{6}+x^3+1\right)^{2^r-1},\ \left(1+x^3+x^4+x^5\right) \left(x^{6}+x^3+1\right)^{2^r-1}$ 
		\end{tabular} & 	\begin{tabular}{@{}c@{}} $2 wt \left(\left(x^{6}+x^3+1\right)^{2^r-1}\right)+$\\ $ wt \left( \left(1+x^3\right) \left(x^{6}+x^3+1\right)^{2^r-1} \right) $ \end{tabular} \\\hline  
		$\left(1+x+x^3+x^4\right) \left(x^{6}+x^3+1\right)^{2^r-1},\ \left(1+x^2+x^3+x^5\right) \left(x^{6}+x^3+1\right)^{2^r-1} $ & 	$2 wt \left( \left(1+x^3\right) \left(x^{6}+x^3+1\right)^{2^r-1} \right)$ \\\hline
		\begin{tabular}{@{}c@{}} 
			$\left(1+x+x^2+x^3+x^4\right) \left(x^{6}+x^3+1\right)^{2^r-1},\ \left(1+x+x^2+x^3+x^5\right) \left(x^{6}+x^3+1\right)^{2^r-1}$\\
			$\left(1+x+x^2+x^4+x^5\right) \left(x^{6}+x^3+1\right)^{2^r-1},\ \left(1+x+x^3+x^4+x^5\right) \left(x^{6}+x^3+1\right)^{2^r-1}$\\ $ \left(1+x^2+x^3+x^4+x^5\right) \left(x^{6}+x^3+1\right)^{2^r-1}$
		\end{tabular} & 	\begin{tabular}{@{}c@{}} $ wt \left(\left(x^{6}+x^3+1\right)^{2^r-1}\right)+$\\ $2 wt \left( \left(1+x^3\right) \left(x^{6}+x^3+1\right)^{2^r-1} \right) $ \end{tabular} \\\hline  
		$\left(1+x+x^2+x^3+x^4+x^5\right) \left(x^{6}+x^3+1\right)^{2^r-1} $ & 	$ 3 wt \left( \left(1+x^3\right) \left(x^{6}+x^3+1\right)^{2^r-1} \right)$
		\\\hline
	\end{tabular}
\end{table}

Using Eqs.~\eqref{eq50} and \eqref{eq51}, the weights of the polynomials $\left( x^{2\cdot3^v}+x^{3^v}+1\right)^{2^r-1}$ and  $\left(1+x^{3^v}\right) \big( x^{2\cdot3^v}+x^{3^v}+1\big)^{2^r-1}$ are given as follows:
\begin{equation}\label{eq53}
	wt \left(\big( x^{2\cdot3^v}+x^{3^v}+1\big)^{2^r-1}\right)= \begin{cases}
		\frac{2^{r+2}-1}{3}, & \text{ if } r \text{ is even},\\
		\frac{2^{r+2}+1}{3}, & \text{ if } r \text{ is odd},
	\end{cases}
\end{equation}
and 
\begin{equation}\label{eq54}
	wt  \left( \big(1+x^{3^v}\big) \big( x^{2\cdot3^v}+x^{3^v}+1\big)^{2^r-1} \right)= \begin{cases}
		\frac{2^{r+2}+2}{3}, & \text{ if } r \text{ is even},\\
		\frac{2^{r+2}-2}{3}, & \text{ if } r \text{ is odd}.
	\end{cases}
\end{equation}
\end{theorem}
Now, using Theorem~\ref{theorem5}, we determine the exact Hamming distance of the codes $C_{2^{\mathcal{T}}-2^{\mathcal{T}-r}}$, where $2\leq r \leq \mathcal{T}$, in the following theorem. 
\begin{theorem}\label{theorem19}
	The Hamming distance of the code $C_{2^{\mathcal{T}}-2^{\mathcal{T}-r}}$, where $2\leq r \leq \mathcal{T}$, is given by
	\begin{equation*}
		d_{2^{\mathcal{T}}-2^{\mathcal{T}-r}}= \begin{cases}
			\frac{2^{r+2}-1}{3}, & \text{if } r \text{ is even},\\
			\frac{2^{r+2}-2}{3}, &\text{if } r \text{ is odd}.
		\end{cases}
	\end{equation*}
\end{theorem}
\begin{proof}
	From Theorem~\ref{theorem5}, we have that for $2\leq r \leq \mathcal{T}$, the Hamming distance of the code $C_{2^{\mathcal{T}}-2^{\mathcal{T}-r}}$ is given by 
	\begin{equation*}
		d_{2^{\mathcal{T}}-2^{\mathcal{T}-r}}= \min  \left\{wt(s(x)) \mid s(x) \in S_{\mathcal{T},r}^{'2\cdot3^v}
		\right\}.
	\end{equation*}
	Using Eq.~\eqref{eq52}, we obtain
	\begin{equation*}
		d_{2^{\mathcal{T}}-2^{\mathcal{T}-r}}= \min \left\{  wt \left(\big( x^{2\cdot3^v}+x^{3^v}+1\big)^{2^r-1}\right),\ wt\left(  \big(1+x^{3^v}\big)\big( x^{2\cdot3^v}+x^{3^v}+1\big)^{2^r-1} \right)  \right\}.               
	\end{equation*} The result follows from Eqs.~\eqref{eq53} and \eqref{eq54}.
\end{proof}
In the next theorem, we provide the complete Hamming distance for all the binary polycyclic codes associated with $ \left(x^{2\cdot3^v}+x^{3^v}+1\right)^{2^\mathcal{T}}$.
\begin{theorem}\label{theorem20}
	The Hamming distance of the code $C_j= \langle  \left(x^{2\cdot3^v}+x^{3^v}+1\right)^j \rangle$, where $1\leq j \leq 2^\mathcal{T}-1$, is given by
	\begin{equation*}
		d_{j}= \begin{cases}
			2, & \text{if }1\leq j \leq 2^{\mathcal{T}-1},\\
			\frac{2^{r+2}-1}{3}, &\text{if } j= 2^{\mathcal{T}} - 2^{\mathcal{T}-r}, \text{ where } r \geq 2 \text{ and } r \text{ is even},\\
			\frac{2^{r+2}-2}{3}, &\text{if } j= 2^{\mathcal{T}} - 2^{\mathcal{T}-r}, \text{ where } r \geq 2 \text{ and } r \text{ is odd},\\
			\frac{2^{r+3} - 2}{3}, &\text{if } j= 2^{\mathcal{T}} - 2^{\mathcal{T}-r}+i,
			\text{ where } r \geq 2,\ r \text{ is even and } 1\leq i \leq 2^{\mathcal{T}-r-1}.
		\end{cases}
	\end{equation*}
	Moreover,
	\begin{equation*}
		4\leq d_{ 2^{\mathcal{T}-1}+i} \leq 5,\quad  1\leq i\leq  2^{\mathcal{T}-2},
	\end{equation*}
	and for odd $r\geq2,$ \begin{equation*}
		\frac{2^{r+3}-4}{3} \leq d_{2^\mathcal{T}- 2^{\mathcal{T}-r}+i} \leq 	\frac{2^{r+3}-4}{3}+1,\quad  1\leq i \leq 2^{\mathcal{T}-r-1}.
	\end{equation*}
\end{theorem}
\begin{proof}
		Since $\mathcal{L}=2^\mathcal{T}$ and $e=3^{v+1}$, the smallest positive integer $\mathcal{J}$ satisfying $3^{v+1} 2^{\mathcal{T}-\mathcal{J}}< 2 \cdot3^v \mathcal{L}= 2 \cdot3^v  2^\mathcal{T}= 4 \cdot 3^v 2^{\mathcal{T}-1}$ is $\mathcal{J}=1$. Thus, by Theorem~\ref{theorem18}
		\begin{equation}\label{eq55}
			d_j=2 \quad \text{for } 1\leq j \leq 2^{\mathcal{T}-1}.
		\end{equation}
		From Theorem~\ref{theorem7}, we have 
	\begin{equation}\label{eq56}
		2 d_{2^{\mathcal{T}}-2^{\mathcal{T}-r}}\leq d_{2^{\mathcal{T}}-2^{\mathcal{T}-r}+i} \leq d_{2^{\mathcal{T}}-2^{\mathcal{T}-r-1}},
	\end{equation}
	where $1\leq r \leq \mathcal{T}-2$ and $1\leq i \leq 2^{\mathcal{T}-r-1}$. For $r=1$ in Eq.~\eqref{eq56}, we obtain
	\begin{equation*}
		2 d_{2^{\mathcal{T}-1}} \leq d_{2^{\mathcal{T}-1}+i}  \leq  d_{2^{\mathcal{T}}-2^{\mathcal{T}-2}},\quad  1\leq i \leq 2^{\mathcal{T}-2}.
	\end{equation*}
	Using Theorem~\ref{theorem19} together with Eq.~\eqref{eq55}, we get
\begin{equation}\label{eq57}
	4\leq d_{ 2^{\mathcal{T}-1}+i} \leq 5,\quad  1\leq i\leq  2^{\mathcal{T}-2}.
\end{equation}
Next, for $r\geq2$ in Eq.~\eqref{eq56},  Theorem~\ref{theorem19} gives
\begin{equation}\label{eq58}
	d_{2^{\mathcal{T}}-2^{\mathcal{T}-r}}= \begin{cases}
		\frac{2^{r+2}-1}{3}, & \text{if } r \text{ is even},\\
		\frac{2^{r+2}-2}{3}, &\text{if } r \text{ is odd}.
	\end{cases}
\end{equation}
Now, we consider the following two cases:\\
\textbf{Case I:} If $r$ is even, then by Eq.~\eqref{eq58},  both bounds in  Eq.~\eqref{eq56} coincide, so
\begin{equation*}
	d_{2^{\mathcal{T}}-2^{\mathcal{T}-r}+i}= 	\frac{2^{r+3}-2}{3},\quad  1\leq i\leq  2^{\mathcal{T}-r-1}.
\end{equation*}
	\textbf{Case II:} If $r$ is odd, then from Eqs.~\eqref{eq56} and \eqref{eq58}, we get
		\begin{equation*}
		\frac{2^{r+3}-4}{3} \leq 	d_{2^{\mathcal{T}}-2^{\mathcal{T}-r}+i}\leq \frac{2^{r+3}-4}{3}+1, \quad  1\leq i\leq  2^{\mathcal{T}-r-1}.
	\end{equation*} This completes the proof.
\end{proof}
\subsection{$d_j$ for $2^{\mathcal{T}-1} < j < \mathcal{L}$, when $ \mathcal{L} = 2^{\mathcal{T}-1}+ \mathcal{L}'$ for some $1< \mathcal{L}' \leq 2^{\mathcal{T}-2}$}

We consider the codes $C_{2^{\mathcal{T}-1}+i}$ for $ 1\leq i < \mathcal{L}'$. Let $\lambda_1$ be the positive integer such that 
\begin{equation*}
	\left(\lambda_1-1\right) 2^{\mathcal{T}-1} < 2 \cdot3^v \left(\mathcal{L}-2^{\mathcal{T}-1}\right)= 2 \cdot 3^v\mathcal{L}' \leq \lambda_1 2^{\mathcal{T}-1}.
\end{equation*}
Since $\mathcal{L}'\leq 2^{\mathcal{T}-2}$, it follows that $1\leq\lambda_1\leq 3^v$. From Theorem~\ref{theorem4}, we have
\begin{itemize}
	\item If $\lambda_1=1$, then
	$
	d_{2^{\mathcal{T}-1}}= wt \Big( \left(x^{2\cdot3^v}+x^{3^v}+1\right)^{2^{\mathcal{T}-1}} \Big)=3.
	$ 
	\item  If $2\leq \lambda_1\leq 3^v$, then 
	$
	d_{2^{\mathcal{T}-1}}= \min \left\{ wt(s(x)) \mid s(x)\in S_{\mathcal{T},1}^{\lambda_1} \right\}.
	$ Since  $S_{\mathcal{T},1}^{\lambda_1}=S_{\mathcal{T},1}^{'\lambda_1}$, thus
	\begin{equation*}
		d_{2^{\mathcal{T}-1}}=	\min \left\{ wt(s(x)) \mid s(x)\in S_{\mathcal{T},1}^{'\lambda_1} \right\} =  \min \left\{ wt(f(x)) \mid f(x)\in F_{\mathcal{T},1}^{'\lambda_1} \right\},
	\end{equation*}
	where $F_{\mathcal{T},1}^{'\lambda_1} $ is defined in a similar way as we have earlier defined  $F_{\mathcal{T},r}^{'2 \cdot 3^v} $,
	\begin{equation*}
		F_{\mathcal{T},1}^{'\lambda_1}= \left\{ g(x)\big( x^{2\cdot3^v}+x^{3^v}+1 \big)\ \big| g(x)\in \mathbb{F}_2[x],\ \deg (g(x)) \leq \lambda_1-1,\ g(x) \text{ has constant term  }  1 \right\}.
	\end{equation*}
	Let $f(x) \in F_{\mathcal{T},1}^{'\lambda_1}$. Since $\lambda_1\leq 3^v$, we may write
	\begin{equation*}
		f(x)= \big( x^{2\cdot3^v}+x^{3^v}+1\big)  \sum_{i=0}^{ 3^v-1} g_ix^{i},
	\end{equation*} where each $g_i\in \mathbb{F}_2$ and $g_0=1$. Since  $cw \left(\left( x^{2\cdot3^v}+x^{3^v}+1\right)^{2^r-1} \right)=3^v$ and $\deg \left(\sum_{i=0}^{ 3^v-1} g_ix^{i}\right) < 3^v$, thus  $wt(f(x)) \geq wt \left(\left( x^{2\cdot3^v}+x^{3^v}+1\right) \right)$. Moreover, $x^{2\cdot3^v}+x^{3^v}+1 \in F_{\mathcal{T},1}^{'\lambda_1}$, therefore
	\begin{equation*}
		d_{2^{\mathcal{T}-1}}= wt\big( x^{2\cdot3^v}+x^{3^v}+1 \big)=3. 
	\end{equation*}
\end{itemize}
Using Theorem~\ref{theorem8}, we immediately obtain the following result.
\begin{theorem}\label{theorem21}
	For $ 1\leq i < \mathcal{L}'$, the Hamming distance of $C_{2^{\mathcal{T}-1}+i}$ satisfies 
	\begin{equation*}
		d_{2^{\mathcal{T}-1}+i} \geq6.
	\end{equation*}
\end{theorem}
\subsection{$d_j$ for $2^{\mathcal{T}-1} < j < \mathcal{L}$, when $ \mathcal{L} =2^\mathcal{T}- 2^{\mathcal{T}-R}+ \mathcal{L}'$ with $2\leq R \leq \mathcal{T}-1$ and $1\leq \mathcal{L}' \leq 2^{\mathcal{T}-R-1}$}

Using Theorems~\ref{theorem9} and \ref{theorem10}, the Hamming distance of the codes $C_j$ for $j> 2^{\mathcal{T}-1}$ can now be determined.
First we compute the Hamming distance for the codes $C_{2^\mathcal{T}-2^{\mathcal{T}-r}}$, where $1\leq r \leq R$.
\begin{theorem}\label{theorem22}
	For $1\leq r\leq R$, let $\lambda_r'$ be the positive integer satisfying $ \left(\lambda_r'-1\right) 2^{\mathcal{T}-r} < 2 \cdot 3^v \left(2^{\mathcal{T}-r}\right)-  2 \cdot 3^v \left(2^{\mathcal{T}-R}-\mathcal{L}'\right) \leq \lambda_r'2^{\mathcal{T}-r}$. Then the Hamming distance of $C_{2^\mathcal{T}-2^{\mathcal{T}-r}}$ is given by
	\begin{equation*}
		d_{2^{\mathcal{T}}-2^{\mathcal{T}-r}}=  \begin{cases}
			\frac{2^{r+2}-1}{3}, & \text{if }  r \text{ is even},\\
			\frac{2^{r+2}-2}{3}, & \text{if } 1\leq r \leq R-1, \ r \text{ is odd},\\
			\frac{2^{r+2}+1}{3}, & \text{if } r=R, \ r \text{ is odd}.
		\end{cases}
	\end{equation*}
\end{theorem}
\begin{proof}
	Since 
	\begin{equation*}
		\left(\lambda_r'-1\right) 2^{\mathcal{T}-r} < 2 \cdot 3^v \left(2^{\mathcal{T}-r}\right)-  2 \cdot 3^v \left(2^{\mathcal{T}-R}-\mathcal{L}'\right) \leq \lambda_r'2^{\mathcal{T}-r},
	\end{equation*}
	thus clearly, $1\leq \lambda_r'\leq  2 \cdot 3^v$.
	From Theorem~\ref{theorem9}, we have 
	\begin{equation*}
		d_{2^{\mathcal{T}}-2^{\mathcal{T}-r}}= \begin{cases}
			wt\Big( \left(x^{2\cdot3^v}+x^{3^v}+1\right)^{2^{\mathcal{T}}-2^{\mathcal{T}-r}} \Big), & \text{if }  \lambda_r'=1,\\
			\min\left\{  wt (s(x)) \mid s(x) \in S_{\mathcal{T},r}^{'\lambda_{r}'} \right\}, &\text{if } \lambda_r'\geq2,
		\end{cases}
	\end{equation*}
	where \allowdisplaybreaks
	\begin{align*}
		S_{\mathcal{T},r}^{'\lambda_{r}'}=  \Bigl\{ &  g(x)^{2^{\mathcal{T}-r}} \big( x^{2\cdot3^v}+x^{3^v}+1\big)^{2^{\mathcal{T}}-2^{\mathcal{T}-r}}\ \big|\ g(x)\in \mathbb{F}_2[x],\ \deg (g(x)) \leq \lambda_{r}'-1,\  g(x) \text{ has constant term  }  1 \Bigr\}.
	\end{align*} Let 
\begin{equation*}
F_{\mathcal{T},r}^{'\lambda_{r}'}= \left\{ g(x)\big( x^{2\cdot3^v}+x^{3^v}+1\big)^{2^r-1} \ \big|\ g(x)\in \mathbb{F}_2[x],\ \deg (g(x)) \leq \lambda_{r}'-1,\ g(x) \text{ has constant term  }  1 \right\}.
\end{equation*} 	Then 
\begin{equation*}
S_{\mathcal{T},r}^{'\lambda_{r}'}= \left\{  f(x)^{2^{\mathcal{T}-r}} \mid f(x)\in F_{\mathcal{T},r}^{'\lambda_{r}'} \right\}.
\end{equation*}
Consequently, 
\begin{equation*}
\min\left\{  wt (s(x)) \mid s(x) \in S_{\mathcal{T},r}^{'\lambda_{r}'} \right\} = \min\left\{  wt (f(x)) \mid f(x) \in F_{\mathcal{T},r}^{'\lambda_{r}'} \right\}.
\end{equation*}
Thus, 
\begin{equation}\label{eq59}
d_{2^{\mathcal{T}}-2^{\mathcal{T}-r}}= \begin{cases}
	wt\Big( \left(x^{2\cdot3^v}+x^{3^v}+1\right)^{2^{\mathcal{T}}-2^{\mathcal{T}-r}} \Big), & \text{if }  \lambda_r'=1,\\
	\min\left\{  wt (f(x)) \mid f(x) \in F_{\mathcal{T},r}^{'\lambda_{r}'} \right\}, &\text{if } \lambda_r'\geq2,
\end{cases}
\end{equation} 	The set  $F_{\mathcal{T},r}^{'\lambda_{r}'}$  is clearly a subset of $F_{\mathcal{T},r}^{'2\cdot3^v}$ as $ \lambda_r' \leq 2\cdot3^v$.
Thus, the weights of elements of $F_{\mathcal{T},r}^{'\lambda_{r}'}$ can be written as linear combination  of $wt \left(  \left( x^{2\cdot3^v}+x^{3^v}+1\right)^{2^r-1} \right)$ and $wt \left( \left(1+x^{3^v}\right) \left( x^{2\cdot3^v}+x^{3^v}+1\right)^{2^r-1} \right)$, as discussed earlier.   Now, we consider the following two cases:\\
\textbf{Case I:} If $1\leq r \leq R-1$. Suppose $1\leq \lambda_{r}' \leq 3^v$, then by the definition of $\lambda_{r}'$, we have 
$
2 \cdot 3^v \left(2^{\mathcal{T}-r}\right)-  2 \cdot 3^v \left(2^{\mathcal{T}-R}-\mathcal{L}'\right) \leq 3^v 2^{\mathcal{T}-r}.
$ Consequently,
\begin{equation}\label{eq60}
2^{\mathcal{T}-r-1} \leq 2^{\mathcal{T}-R} - \mathcal{L}'.
\end{equation}
Since $ 1\leq \mathcal{L}' \leq 2^{\mathcal{T}-R-1}$ and $1\leq r \leq R-1$, it follows that
$
2^{\mathcal{T}-R-1} \leq 	2^{\mathcal{T}-R} - \mathcal{L}'  \leq 2^{\mathcal{T}-R}-1 < 2^{\mathcal{T}-R} \leq 2^{\mathcal{T}-r-1}. 
$ This implies that $ 2^{\mathcal{T}-R} - \mathcal{L}' < 2^{\mathcal{T}-r-1}$, which contradicts Inequality~\eqref{eq60}. Therefore, for $1\leq r \leq R-1$, we have $ 3^v+1 \leq \lambda_{r}' \leq 2 \cdot 3^v.$ Hence, in this case 
\begin{equation*}
d_{2^{\mathcal{T}}-2^{\mathcal{T}-r}}= \min\left\{  wt \Big(\big(x^{2\cdot3^v}+x^{3^v}+1\big)^{2^r-1}\Big),\ wt \Big( \big(1+x^{3^v}\big) \big(x^{2\cdot3^v}+x^{3^v}+1\big)^{2^r-1} \Big)  \right\}.
\end{equation*}
	Using Eqs.~\eqref{eq53} and \eqref{eq54}, we obtain
\begin{equation*}
	d_{2^{\mathcal{T}}-2^{\mathcal{T}-r}}= \begin{cases}
		\frac{2^{r+2}-1}{3}, & \text{if } r \text{ is even},\\
		\frac{2^{r+2}-2}{3}, &\text{if } r \text{ is odd}.
	\end{cases}
\end{equation*}
	\textbf{Case II:} If $r=R.$ Since $ \mathcal{L}' \leq 2^{\mathcal{T}-R-1}$,  we have
$
2 \cdot 3^v \left(2^{\mathcal{T}-R}\right)-  2 \cdot 3^v \left(2^{\mathcal{T}-R}-\mathcal{L}'\right)= 2 \cdot 3^v  \mathcal{L}' \leq 3^v 2^{\mathcal{T}-R}.
$
Therefore, $\lambda_{R}' \leq 3^v$.

If $\lambda_{R}'=1$, then by Eq.~\eqref{eq59}, we have 
	\begin{equation*}
	d_{2^{\mathcal{T}}-2^{\mathcal{T}-R}}= wt \Big( \big(x^{2\cdot3^v}+x^{3^v}+1\big)^{2^{\mathcal{T}}-2^{\mathcal{T}-R}} \Big)= wt \Big( \big(x^{2\cdot3^v}+x^{3^v}+1\big)^{2^R-1} \Big).
\end{equation*}

If $2\leq \lambda_{R}' \leq 3^v. $ Then using Eq.~\eqref{eq59}, we get
$
d_{2^{\mathcal{T}}-2^{\mathcal{T}-R}}= \min\left\{  wt (f(x)) \mid f(x) \in F_{\mathcal{T},R}^{'\lambda_{R}'} \right\}.
$
In this case,
\begin{equation*}
	\min\left\{  wt (f(x)) \mid f(x) \in F_{\mathcal{T},R}^{'\lambda_{R}'} \right\}= wt \Big( \big(x^{2\cdot3^v}+x^{3^v}+1\big)^{2^R-1} \Big).
\end{equation*}

Thus,
\begin{equation*}
	d_{2^{\mathcal{T}}-2^{\mathcal{T}-R}}=  wt \Big( \big(x^{2\cdot3^v}+x^{3^v}+1\big)^{2^R-1} \Big).
\end{equation*}
	Therefore, using Eq.~\eqref{eq53}, we obtain
\begin{equation*}
	d_{2^{\mathcal{T}}-2^{\mathcal{T}-R}}=wt \Big( \big(x^{2\cdot3^v}+x^{3^v}+1\big)^{2^R-1} \Big)= \begin{cases}
		\frac{2^{R+2}-1}{3}, & \text{ if } R \text{ is even},\\
		\frac{2^{R+2}+1}{3}, & \text{ if } R \text{ is odd}.
	\end{cases}
\end{equation*}
This completes the proof.
\end{proof}
Using Theorems~\ref{theorem22} and Theorem~\ref{theorem10}, we get the following result.
\begin{theorem}\label{theorem23}
	For $1\leq r \leq R-2$ with $ R \geq3$, and for $1\leq i \leq 2^{\mathcal{T}-r-1}$, the Hamming distance of $C_{2^{\mathcal{T}}-2^{\mathcal{T}-r}+i}$ satisfies the following:
	\begin{enumerate}
		\item If $r$ is even, then
		\begin{equation*}
			d_{2^{\mathcal{T}}-2^{\mathcal{T}-r}+i} = 	\frac{2^{r+3}-2}{3}.
		\end{equation*}
		\item If $r$ is odd, then
		\begin{equation*}
			\frac{2^{r+3}-4}{3} \leq	d_{2^{\mathcal{T}}-2^{\mathcal{T}-r}+i} \leq \frac{2^{r+3}-4}{3}+1.
		\end{equation*}
	\end{enumerate}
	For $r=R-1$ and $1\leq i \leq 2^{\mathcal{T}-R}$,  we have
	\begin{enumerate}
		\item If $R-1$ is even, then
		\begin{equation*}
			\frac{2^{R+2}-2}{3} \leq 	d_{2^{\mathcal{T}}-2^{\mathcal{T}-R+1}+i} \leq \frac{2^{R+2}-2}{3}+1.
		\end{equation*}
		\item  If $R-1$ is odd, then
		\begin{equation*}
			\frac{2^{R+2}-4}{3} \leq 	d_{2^{\mathcal{T}}-2^{\mathcal{T}-R+1}+i} \leq \frac{2^{R+2}-4}{3}+1.
		\end{equation*}
	\end{enumerate}
	Moreover, for $r=R$ and $1\leq i < \mathcal{L}'$,
	\begin{equation*}
		d_{2^{\mathcal{T}}-2^{\mathcal{T}-R}+i} \geq \frac{2^{R+3}-2}{3}, \quad \text{if } R \text{ is even},
	\end{equation*}
	\begin{equation*}
		d_{2^{\mathcal{T}}-2^{\mathcal{T}-R}+i} \geq \frac{2^{R+3}+2}{3}, \quad \text{if }  R \text{ is odd}.
	\end{equation*}
\end{theorem}
Thus, combining Theorems~\ref{theorem18}--\ref{theorem23}, we obtain complete or explicit bounded descriptions for the Hamming distances of all binary polycyclic codes associated with $\left(x^{2\cdot3^v}+x^{3^v}+1\right)^\mathcal{L}$,  $ \mathcal{L} \geq 2$.
\begin{remark}
	Since the polynomial $x^{2\cdot3^v}+x^{3^v}+1$ is self-reciprocal over $\mathbb{F}_2$, it follows that every binary polycyclic code associated with $\left(x^{2\cdot3^v}+x^{3^v}+1\right)^{2^\mathcal{T}}$ is reversible.
\end{remark}
We next turn to the LCD property of binary polycyclic codes associated with $\left(x^{2\cdot3^v}+x^{3^v}+1\right)^{2^\mathcal{T}}$, where  $v\geq 0$ and $\mathcal{T}\geq1$. We begin by identifying several families that are LCD.
\begin{theorem}\label{theorem24}
	For $0\leq r \leq \mathcal{T}-1$, the binary polycyclic code $C_{2^r}= \big\langle \left(x^{2\cdot3^v}+x^{3^v}+1\right)^{2^r} \big\rangle$ associated with $ \left(x^{2\cdot3^v}+x^{3^v}+1\right)^{2^\mathcal{T}}$ is an LCD code.
\end{theorem}
\begin{proof}
		Here, we have $P(x)=P^*(x)=x^{2\cdot3^v}+x^{3^v}+1$,  $Ord(x^{2\cdot3^v}+x^{3^v}+1)=e= 3^{v+1}$, and $U(x)=U^*(x)= x^{3^v}+1.$
	By Theorem~\ref{theorem15}, it is enough to show that for every non-zero $\delta(x)\in \mathbb{F}_2[x]$ with $\deg (\delta(x))< 2^{r+1}3^v$, 
	the polynomial
	\begin{equation*}
		R(x)=	\big(x^{3^{v+1}}+1\big)^{2^\mathcal{T}-2^{r+1}} \big(x^{2 \cdot3^v}+1\big)^{2^r} \delta(x) \bmod \left(x^{3^v2^{\mathcal{T}+1}}\right)
	\end{equation*} satisfies
	\begin{equation*}
		\deg(R(x))\geq 2 \cdot3^v (2^\mathcal{T}-2^r).
	\end{equation*}
	Expanding $R(x)$ as 
	\allowdisplaybreaks
	\begin{align}\label{eq61}
		\nonumber R(x)=	&  \Big(x^{2^{r+1}3^{v+1}}+1\Big)^{2^{\mathcal{T}-r-1}-1} \big(x^{2^{r+1}3^v}+1\big) \delta(x) \mod \left(x^{3^v2^{\mathcal{T}+1}}\right)\\\nonumber
		=& \left(1+  x^{2^{r+1}3^{v+1}}+ x^{2(2^{r+1}3^{v+1})}+x^{3(2^{r+1}3^{v+1})}+\ldots+ x^{(2^{\mathcal{T}-r-1}-1)(2^{r+1}3^{v+1})} \right)\big(x^{2^{r+1}3^v}+1\big) \delta(x) \mod \left(x^{3^v2^{\mathcal{T}+1}}\right)\\\nonumber
		=& \left(1+x^{2^{r+1}3^v} \right) \delta(x)+ x^{2^{r+1}3^{v+1}} \left(1+x^{2^{r+1}3^v} \right) \delta(x)+ x^{2(2^{r+1}3^{v+1})} \left(1+x^{2^{r+1}3^v} \right) \delta(x)+\ldots\\&+ x^{(2^{\mathcal{T}-r-1}-1)(2^{r+1}3^{v+1})} \left(1+x^{2^{r+1}3^v} \right) \delta(x) \mod \left(x^{3^v2^{\mathcal{T}+1}}\right).
	\end{align}
	Now, we consider the following two cases:\\
		\textbf{Case I:} If $(2^{\mathcal{T}-r-1}-1)(2^{r+1}3^{v+1}) < 3^v2^{\mathcal{T}+1}$, then $2^\mathcal{T}< 3 \cdot 2^{r+1}.$ Moreover, since $ 0\leq r \leq \mathcal{T}-1$, thus $ 2^{r+1}\leq 2^\mathcal{T}$. Therefore, we have
	\begin{equation*}
		2^{r+1}\leq 2^\mathcal{T} < 3 \cdot 2^{r+1}< 2^{r+3}.
	\end{equation*}
	This gives us $ 2^\mathcal{T}= 2^{r+1}$, or  $ 2^\mathcal{T}= 2^{r+2}$. Consequently, $r=\mathcal{T}-2$ or $r=\mathcal{T}-1$. Now,  consider the following subcases:
	
	\textbf{Subcase I:} If $r=\mathcal{T}-2$, then by Eq.~\eqref{eq61}, we have \allowdisplaybreaks
		\begin{align}\label{eq62}
		\nonumber	R(x)= & \left(1+x^{2^{\mathcal{T}-1}3^v}\right) \delta(x)+ x^{2^{\mathcal{T}-1}3^{v+1}} \left(1+x^{2^{\mathcal{T}-1}3^v}\right) \delta(x) \mod \left(x^{3^v2^{\mathcal{T}+1}}\right)\\
		=& \delta(x)+  x^{2^{\mathcal{T}-1}3^v} \delta(x)+  x^{2^{\mathcal{T}-1}3^{v+1}}  \delta(x) \mod \left(x^{3^v2^{\mathcal{T}+1}}\right).
	\end{align} Here, $\deg (\delta(x))< 2^{\mathcal{T}-1}3^v.$ By Eq.~\eqref{eq62}, we have $2^{\mathcal{T}-1}3^{v+1} \leq \deg(R(x)) < 3^v2^{\mathcal{T}+1}.$ Clearly, for  $r=\mathcal{T}-2$, $2 \cdot3^v (2^\mathcal{T}-2^r)= 2^{\mathcal{T}-1}3^{v+1}. $ Thus, in this case, $C_{2^r}$ is an LCD code.

	\textbf{Subcase II:} If  $r=\mathcal{T}-1$, then by Eq.~\eqref{eq61}, we have
\allowdisplaybreaks
\begin{align*}
	R(x)= & \left(1+x^{2^\mathcal{T}3^v}\right) \delta(x)  \mod \left(x^{3^v2^{\mathcal{T}+1}}\right)= \delta(x)+ x^{2^\mathcal{T}3^v}  \delta(x)  \mod \left(x^{3^v2^{\mathcal{T}+1}}\right),
\end{align*} where $\deg(\delta(x)) < 2^\mathcal{T}3^v$. Clearly, $ 2^\mathcal{T}3^v \leq \deg(R(x)) < 3^v2^{\mathcal{T}+1}. $ Since, for $r=\mathcal{T}-1$, $2 \cdot3^v (2^\mathcal{T}-2^r)= 2^{\mathcal{T}}3^{v}$, thus $C_{2^r}$ is an LCD code.\\
\textbf{Case II:} If $(2^{\mathcal{T}-r-1}-1)(2^{r+1}3^{v+1}) \geq 3^v2^{\mathcal{T}+1}$. Let $N_r$ be the largest integer satisfying $0\leq N_r \leq 2^{\mathcal{T}-r-1}-2$ such that $N_r(2^{r+1}3^{v+1}) < 3^v2^{\mathcal{T}+1}. $ Then 
\begin{equation}\label{eq63}
	N_r= \Big\lfloor \frac{2^{\mathcal{T}-r}}{3} \Big\rfloor = \begin{cases}
		\frac{2^{\mathcal{T}-r}-1}{3}, & \text{if } \mathcal{T}-r \text{ is even},\\
		\frac{2^{\mathcal{T}-r}-2}{3}, & \text{if } \mathcal{T}-r \text{ is odd}.
	\end{cases}
\end{equation} By Eq.~\eqref{eq61}, we have 
\allowdisplaybreaks
\begin{align}\label{eq64}
\nonumber	R(x)=&  \left(1+x^{2^{r+1}3^v} \right) \delta(x)+ x^{2^{r+1}3^{v+1}} \left(1+x^{2^{r+1}3^v} \right) \delta(x)+ x^{2(2^{r+1}3^{v+1})} \left(1+x^{2^{r+1}3^v} \right) \delta(x)+\ldots\\&+ x^{N_r(2^{r+1}3^{v+1})} \left(1+x^{2^{r+1}3^v} \right) \delta(x) \mod \left(x^{3^v2^{\mathcal{T}+1}}\right),
\end{align} where $\deg (\delta(x))< 2^{r+1}3^v.$
Using Eq.~\eqref{eq63}, we obtain
\begin{equation}\label{eq65}
	N_r(2^{r+1}3^{v+1})= \begin{cases}
		\left(2^{\mathcal{T}+1}-2^{r+1}\right) 3^v, & \text{if } \mathcal{T}-r \text{ is even},\\
		\left(2^{\mathcal{T}+1}-2^{r+2}\right) 3^v, & \text{if } \mathcal{T}-r \text{ is odd}.
	\end{cases}
\end{equation} Now, we consider the following two subcases:

\textbf{Subcase I:} If $\mathcal{T}-r$ is even, then by Eqs.~\eqref{eq64} and~\eqref{eq65}, we have 
\allowdisplaybreaks
\begin{align*}
	2\cdot3^v (2^\mathcal{T}-2^r) \leq \deg(R(x))< 3^v2^{\mathcal{T}+1}.
\end{align*} Thus, $C_{2^r}$ is an LCD code. 

	\textbf{Subcase II:} If $\mathcal{T}-r$ is odd, then by Eqs.~\eqref{eq64} and~\eqref{eq65}, we have 
\begin{equation*}
	2\cdot3^v (2^\mathcal{T}-2^r) \leq \deg(R(x))< 3^v2^{\mathcal{T}+1}.
\end{equation*} Thus, $C_{2^r}$ is an LCD code. 

This completes the proof.
\end{proof}
\begin{theorem}\label{theorem25}
	For $2\leq r \leq \mathcal{T}$,	the binary polycyclic code $C_{2^\mathcal{T}-2^{\mathcal{T}-r}}= \big\langle \left(x^{2\cdot3^v}+x^{3^v}+1\right)^{2^\mathcal{T}-2^{\mathcal{T}-r}} \big\rangle$ associated with $ \left(x^{2\cdot3^v}+x^{3^v}+1\right)^{2^\mathcal{T}}$  is an LCD code.
\end{theorem}
\begin{proof}
		As in Theorem~\ref{theorem24}, we have 	 $P(x)=P^*(x)=x^{2\cdot3^v}+x^{3^v}+1$,  $Ord(x^{2\cdot3^v}+x^{3^v}+1)=e= 3^{v+1}$, and $U(x)=U^*(x)= x^{3^v}+1.$ 	Suppose that $C_{2^\mathcal{T}-2^{\mathcal{T}-r}}$ is not an LCD code. Then, by Theorem~\ref{theorem16}, there exists non-zero polynomials $\gamma(x), \delta(x)\in \mathbb{F}_2[x]$ such that $\deg(\gamma(x))< 2 \cdot3^v2^{\mathcal{T}-r}$, $\deg(\delta(x))< 2 \cdot3^v (2^\mathcal{T}-2^{\mathcal{T}-r})$, and
			\begin{equation*}
			\big(x^{2\cdot3^v}+x^{3^v}+1\big)^{2^\mathcal{T}-2^{\mathcal{T}-r+1}}\gamma(x) \equiv   \big(x^{3^v}+1\big)^{2^\mathcal{T}}\delta(x) \mod (x^{ 2 \cdot 3^v 2^\mathcal{T}}).
		\end{equation*}
		Equivalently, 
		\begin{equation} \label{eq66}
			\Big(x^{2\cdot3^v 2^{\mathcal{T}-r+1}}+x^{3^v2^{\mathcal{T}-r+1}}+1\Big)^{2^{r-1}-1} \gamma(x) \equiv  \big(x^{3^v}+1\big)^{2^\mathcal{T}}\delta(x) \mod (x^{ 2 \cdot 3^v 2^\mathcal{T}}).
		\end{equation} 	Now, since $\deg (\delta(x)) < 2 \cdot3^v (2^\mathcal{T}-2^{\mathcal{T}-r})$, we can write
	\begin{equation*}
	\delta(x)= \sum_{j=0}^{2 \cdot3^v (2^\mathcal{T}-2^{\mathcal{T}-r})-1} \delta_j x^{j}= \sum_{j=0}^{3^v 2^\mathcal{T}-1} \delta_j x^{j}+ \sum_{j=3^v2^\mathcal{T}}^{3^v(2^{\mathcal{T}+1}-2^{\mathcal{T}-r+1})-1}  \delta_j x^{j}.
\end{equation*} Multiplying by $	\left(x^{3^v}+1\right)^{2^\mathcal{T}}$, we have \allowdisplaybreaks
\begin{align*}
\left(x^{3^v}+1\right)^{2^\mathcal{T}}\delta(x) = & \sum_{j=0}^{3^v 2^\mathcal{T}-1} \delta_j x^{j}+ \sum_{j=3^v2^\mathcal{T}}^{3^v(2^{\mathcal{T}+1}-2^{\mathcal{T}-r+1})-1}  \delta_j x^{j}+x^{3^v2^\mathcal{T}} \sum_{j=0}^{3^v 2^\mathcal{T}-1} \delta_j x^{j}+ x^{3^v2^\mathcal{T}} \sum_{j=3^v2^\mathcal{T}}^{3^v(2^{\mathcal{T}+1}-2^{\mathcal{T}-r+1})-1}  \delta_j x^{j}.
\end{align*} Therefore, Eq.~\eqref{eq66} reduces to \allowdisplaybreaks
\begin{align}\label{eq67}
\nonumber\left(x^{2\cdot3^v 2^{\mathcal{T}-r+1}}+x^{3^v2^{\mathcal{T}-r+1}}+1\right)^{2^{r-1}-1} \gamma(x)=& \sum_{j=0}^{3^v 2^\mathcal{T}-1} \delta_j x^{j}+ \sum_{j=3^v2^\mathcal{T}}^{3^v(2^{\mathcal{T}+1}-2^{\mathcal{T}-r+1})-1}  \delta_j x^{j}+x^{3^v2^\mathcal{T}} \sum_{j=0}^{3^v 2^\mathcal{T}-1} \delta_j x^{j}\\\nonumber
=& \sum_{j=0}^{3^v 2^\mathcal{T}-1} \delta_j x^{j}+ \sum_{j=3^v2^\mathcal{T}}^{3^v(2^{\mathcal{T}+1}-2^{\mathcal{T}-r+1})-1}  \delta_j x^{j}\\&+x^{3^v2^\mathcal{T}} \sum_{j=0}^{3^v(2^\mathcal{T}-2^{\mathcal{T}-r+1})-1} \delta_j x^{j}+ x^{3^v2^\mathcal{T}} \sum_{j=3^v(2^\mathcal{T}-2^{\mathcal{T}-r+1})}^{3^v 2^\mathcal{T}-1} \delta_j x^{j}.
\end{align} The polynomial $\left(x^{2\cdot3^v 2^{\mathcal{T}-r+1}}+x^{3^v2^{\mathcal{T}-r+1}}+1\right)^{2^{r-1}-1}$ can be expressed in the  form $  a_0 +a_1 x^{3^v2^{\mathcal{T}-r+1}}+ a_2 x^{2( 3^v2^{\mathcal{T}-r+1})}+ a_3 x^{3( 3^v2^{\mathcal{T}-r+1})}+\ldots + a_{2^r-2} x^{(2^r-2)( 3^v2^{\mathcal{T}-r+1})},
$ where  $a_i\in \mathbb{F}_2$. Moreover, by Lemma~\ref{lemma1}, we obtain 
$
a_{2^{r-1}-1}=1.
$ Since $\deg(\gamma(x))< 3^v2^{\mathcal{T}-r+1}$, it can be written as $\gamma(x)= \sum_{l=0}^{3^v 2^{\mathcal{T}-r+1}-1 } \gamma_l x^{l}$ for $\gamma_l \in \mathbb{F}_2$. and hence the LHS of Eq.~\eqref{eq67} can be represented as \allowdisplaybreaks
\begin{align*}
	a_0 \sum_{l=0}^{3^v 2^{\mathcal{T}-r+1}-1 } \gamma_l x^{l}+ a_1 x^{3^v2^{\mathcal{T}-r+1}} \sum_{l=0}^{3^v 2^{\mathcal{T}-r+1}-1 } \gamma_l x^{l}+\ldots+a_{2^{r-1}-2} x^{(2^{r-1}-2)(3^v2^{\mathcal{T}-r+1})}\sum_{l=0}^{3^v 2^{\mathcal{T}-r+1}-1 } \gamma_l x^{l}+\\  x^{(2^{r-1}-1)(3^v2^{\mathcal{T}-r+1})}\sum_{l=0}^{3^v 2^{\mathcal{T}-r+1}-1 } \gamma_l x^{l}+ a_{2^{r-1}} x^{2^{r-1}(3^v2^{\mathcal{T}-r+1})}\sum_{l=0}^{3^v 2^{\mathcal{T}-r+1}-1 } \gamma_l x^{l}+\ldots+ a_{2^r-2} x^{(2^r-2)( 3^v2^{\mathcal{T}-r+1})} \sum_{l=0}^{3^v 2^{\mathcal{T}-r+1}-1 } \gamma_l x^{l}.&
\end{align*} 	Since $\gamma(x)\ne 0$, thus in the LHS of Eq.~\eqref{eq67}, the coefficient of $x^{3^v(2^\mathcal{T}- 2^{\mathcal{T}-r+1})+l}$ is $1$ for some $0\leq l \leq 3^v 2^{\mathcal{T}-r+1}-1.$ Hence, by Eq.~\eqref{eq67}, we have 
$
\delta_{3^v(2^\mathcal{T}- 2^{\mathcal{T}-r+1})+l}= 1
$ for some  $0\leq l \leq 3^v 2^{\mathcal{T}-r+1}-1.$  Consequently, the polynomial
\begin{equation*}
x^{3^v2^\mathcal{T}} \sum_{j=3^v(2^\mathcal{T}-2^{\mathcal{T}-r+1})}^{3^v 2^\mathcal{T}-1} \delta_j x^{j} \ne 0.
\end{equation*} This implies that the degree of the RHS of  Eq.~\eqref{eq67} is $\geq 2 \cdot 3^v \left(2^\mathcal{T}-2^{\mathcal{T}-r}\right)$, while the degree of the LHS of Eq.~\eqref{eq67} is strictly less than  $ 2 \cdot 3^v \left(2^\mathcal{T}-2^{\mathcal{T}-r}\right).$ This contradiction shows that the  Equality~\eqref{eq67} cannot hold. Hence, the code  $C_{2^\mathcal{T}-2^{\mathcal{T}-r}}$ is  an LCD code.  This completes the proof.
\end{proof}
The previous two theorems provide large classes of LCD codes corresponding to some particular powers of the polynomial $x^{2\cdot3^v}+x^{3^v}+1$. We now consider an additional case that does not belong to the above families.

\begin{theorem}\label{theorem26}
	The binary polycyclic code $C_{3}= \big\langle \left(x^{2\cdot3^v}+x^{3^v}+1\right)^{3} \big\rangle$ associated with $ \left(x^{2\cdot3^v}+x^{3^v}+1\right)^{2^\mathcal{T}}$, where  $\mathcal{T}\geq 3$ and $v\geq 0$, is an LCD code.
\end{theorem}
\begin{proof}
		Here, we have $P(x)=P^*(x)=x^{2\cdot3^v}+x^{3^v}+1$,  $Ord(x^{2\cdot3^v}+x^{3^v}+1)=e= 3^{v+1}$, and $U(x)=U^*(x)= x^{3^v}+1.$
	By Theorem~\ref{theorem15}, it suffices to prove that for every non-zero $\delta(x)\in \mathbb{F}_2[x]$ with $\deg (\delta(x))< 6 \cdot3^v$, 
	\begin{equation*}
		\deg(R(x))\geq 2 \cdot3^v (2^\mathcal{T}-3),
	\end{equation*}  where 
	\begin{equation*}
		R(x)=	\left(x^{3^{v+1}}+1\right)^{2^\mathcal{T}-6} \left(x^{2 \cdot3^v}+1\right)^{3} \delta(x) \bmod \left(x^{3^v2^{\mathcal{T}+1}}\right).
	\end{equation*} We expand $R(x)$ as \allowdisplaybreaks
	\begin{align*}
		R(x)=& \left(x^{3^{v+1}}+1\right)^{2^\mathcal{T}-8} \left(x^{3^{v+1}}+1\right)^{2}  \left(x^{2 \cdot3^v}+1\right)^{3} \delta(x) \bmod \left(x^{3^v2^{\mathcal{T}+1}}\right)\\
		=& \left(x^{24 \cdot 3^v}+1\right)^{2^{\mathcal{T}-3}-1} \left(x^{6 \cdot3^v}+1\right) \left(x^{2 \cdot3^v}+1\right)^{3} \delta(x) \bmod \left(x^{3^v2^{\mathcal{T}+1}}\right)\\
		=& \left(x^{24 \cdot 3^v}+1\right)^{2^{\mathcal{T}-3}-1}   \left(x^{2 \cdot3^v}+1\right) \left(x^{4\cdot3^v}+x^{2\cdot 3^v}+1\right) \left(x^{2 \cdot3^v}+1\right)^{3} \delta(x) \bmod \left(x^{3^v2^{\mathcal{T}+1}}\right)\\
		= & \left(\sum_{j=0}^{2^{\mathcal{T}-3}-1} x^{24 \cdot 3^v j}\right) \left(x^{8 \cdot3^v}+1\right) \left(x^{4\cdot3^v}+x^{2\cdot 3^v}+1\right)  \delta(x) \bmod \left(x^{3^v2^{\mathcal{T}+1}}\right).
	\end{align*}
	Now, consider the following two cases:\\
\textbf{Case I:} If $24 \cdot 3^v\left(2^{\mathcal{T}-3}-1\right) < 3^v2^{\mathcal{T}+1}$, then $2^\mathcal{T}< 24.$ Moreover, since $ 3 \leq \mathcal{T}$, thus $ 3\leq \mathcal{T} \leq 4$.  Next, we consider the following subcases:

\textbf{Subcase I:} If $\mathcal{T}=3$, then  we have \allowdisplaybreaks
\begin{align*}
	\nonumber	R(x)= & \left(x^{8 \cdot3^v}+1\right) \left(x^{4\cdot3^v}+x^{2\cdot 3^v}+1\right)  \delta(x) \bmod \left(x^{16 \cdot 3^v}\right).
\end{align*} Here, $\deg (\delta(x))< 6\cdot3^v.$ Our aim is to show that $\deg (R(x)) \geq 10\cdot 3^v.$  
If $\deg (\delta(x)) < 4 \cdot 3^v$, then $\deg (R(x)) \geq 12 \cdot 3^v$, and we are done. Suppose $4 \cdot 3^v \leq \deg (\delta(x)) < 6 \cdot 3^v.$  Then it can be expressed as 
\begin{equation*}
	\delta(x)= \sum_{j=0}^{6\cdot 3^v -1} \delta_j x^j,\quad \delta_j \in \mathbb{F}_2.
\end{equation*} Moreover, $\delta_j =1 \text{ for some } 4 \cdot 3^v \leq j < 6 \cdot 3^v$. Consequently, there exists 
$0\leq l \leq 2\cdot3^v-1$ such that $\delta_{4\cdot3^v+l}=1.$ For such $l$, the coefficient of $x^{10\cdot3^v+l}$ in $R(x)$ is $\delta_{2\cdot3^v+l}+\delta_l$, and the coefficient of $x^{12\cdot3^v+l}$ in $R(x)$ is $\delta_{4\cdot3^v+l}+\delta_{2\cdot3^v+l}+\delta_l$. Thus, either $\delta_{2\cdot3^v+l}+\delta_l=1$ or $\delta_{4\cdot3^v+l}+\delta_{2\cdot3^v+l}+\delta_l=1.$ Hence, $\deg (R(x)) \geq 10\cdot 3^v.$  
Therefore, in this case, $C_{3}$ is an LCD code.

\textbf{Subcase II:} If  $\mathcal{T}=4$, then 
\allowdisplaybreaks
\begin{align*}
	R(x)= & \left(x^{8 \cdot3^v}+1\right) \left(x^{4\cdot3^v}+x^{2\cdot 3^v}+1\right)  \delta(x) + x^{24\cdot 3^v} \left(x^{8 \cdot3^v}+1\right) \left(x^{4\cdot3^v}+x^{2\cdot 3^v}+1\right)  \delta(x) \bmod \left(x^{32 \cdot 3^v}\right),
\end{align*} where $\deg(\delta(x)) < 6 \cdot3^v$. Our aim is to show that $\deg (R(x)) \geq 26\cdot 3^v.$   As in Subcase~I, in this case also, we get that $C_{3}$ is an LCD code.\\
\textbf{Case II:} If $24 \cdot 3^v\left(2^{\mathcal{T}-3}-1\right) \geq 3^v2^{\mathcal{T}+1}$, then $2^\mathcal{T}\geq 24.$ Let $\mathcal{N}$ be the largest integer satisfying $0\leq \mathcal{N} \leq 2^{\mathcal{T}-3}-2$ such that $24 \cdot 3^v \mathcal{N} < 3^v2^{\mathcal{T}+1}. $ Then 
\begin{equation*}
	\mathcal{N}= \Big\lfloor \frac{2^{\mathcal{T}-2}}{3} \Big\rfloor = \begin{cases}
		\frac{2^{\mathcal{T}-2}-1}{3}, & \text{if } \mathcal{T}-2 \text{ is even},\\
		\frac{2^{\mathcal{T}-2}-2}{3}, & \text{if } \mathcal{T}-2 \text{ is odd}.
	\end{cases}
\end{equation*} Thus,
\allowdisplaybreaks
\begin{align*}
	\nonumber	R(x)=&  \left(x^{8 \cdot3^v}+1\right) \left(x^{4\cdot3^v}+x^{2\cdot 3^v}+1\right)  \delta(x) + x^{24\cdot3^v} \left(x^{8 \cdot3^v}+1\right) \left(x^{4\cdot3^v}+x^{2\cdot 3^v}+1\right)  \delta(x) + \ldots\\&+x^{24\cdot3^v\mathcal{N}} \left(x^{8 \cdot3^v}+1\right) \left(x^{4\cdot3^v}+x^{2\cdot 3^v}+1\right)  \delta(x)  \mod \left(x^{3^v2^{\mathcal{T}+1}}\right),
\end{align*} where $\deg (\delta(x))< 6\cdot3^v.$
Now,  considering both the cases $\mathcal{T}-2$  odd and $\mathcal{T}-2$  even, and following the same as in Case~I, we get that $	\deg(R(x))\geq 2 \cdot3^v (2^\mathcal{T}-3)$.  Thus, $C_{3}$ is an LCD code. 

This completes the proof.
\end{proof}
The above result, together with Theorems~\ref{theorem24} and \ref{theorem25}, suggests a broader pattern.  Motivated by these results together with extensive computations, we propose the following conjecture.
\begin{conjecture}
	All  binary polycyclic codes associated with $\left(x^{2\cdot3^v}+x^{3^v}+1\right)^{2^\mathcal{T}}$, where  $v\geq 0$, $\mathcal{T}\geq1$, are LCD codes. 
\end{conjecture}
We now determine the parameters of the LCD codes associated with $\left(x^{2\cdot3^v}+x^{3^v}+1\right)^{2^\mathcal{T}}$, where  $v\geq 0$, $\mathcal{T}\geq1$.
By combining Theorems~\ref{theorem20} and \ref{theorem24}, we obtain the following result.
\begin{theorem}\label{theorem27}
	For $0 \leq r \leq \mathcal{T}-1$, the code $C_{2^r} = \Big\langle \left(x^{2\cdot3^v}+x^{3^v}+1\right)^{2^r} \Big\rangle$ is an LCD code with parameters $[3^v2^{\mathcal{T}+1},\; 3^v2^{r+1}\\(2^{\mathcal{T}-r}-1),\; 2].$
	Moreover, its dual $C_{2^r}^\perp$ is an LCD code with parameters
	$
	\left[3^v2^{\mathcal{T}+1},\; 3^v2^{r+1},\; d_{2^r}^\perp\right],
	$
	where $d_{2^r}^\perp$ is given in Theorem~\ref{theorem12}.
\end{theorem}
In particular, for $r=0$, we obtain a more explicit description using Corollary~\ref{corollary3}.

\begin{theorem}\label{theorem28}
	For $\mathcal{T} \geq 1$ and $v \geq 0$, the code
	$	C_1 = \left\langle x^{2\cdot3^v}+x^{3^v}+1 \right\rangle	$ 	and its dual $C_1^\perp$ are LCD codes with parameters
	$
	\left[3^v2^{\mathcal{T}+1},\; 3^v\left(2^{\mathcal{T}+1}-2\right),\; 2\right]
	$
	and
	$
	\left[3^v2^{\mathcal{T}+1},\; 2\cdot3^v,\; d_1^\perp\right],
	$
	respectively, where
	\[
	d_1^\perp =
	\begin{cases}
		\dfrac{2^{\mathcal{T}+2}-2}{3}, & \text{if } \mathcal{T} \text{ is odd},\\[6pt]
		\dfrac{2^{\mathcal{T}+2}-1}{3}, & \text{if } \mathcal{T} \text{ is even}.
	\end{cases}
	\]
\end{theorem}

Combining Theorems~\ref{theorem20} and \ref{theorem25}, we obtain another family of LCD codes.
\begin{theorem}\label{theorem29}
	For $2 \leq r \leq \mathcal{T}$, the code
	$
	C_{2^\mathcal{T}-2^{\mathcal{T}-r}} =
	\Big\langle \left(x^{2\cdot3^v}+x^{3^v}+1\right)^{2^\mathcal{T}-2^{\mathcal{T}-r}} \Big\rangle
	$
	is an LCD code with parameters
	$
	\left[3^v2^{\mathcal{T}+1},\; 3^v2^{\mathcal{T}-r+1},\; d_{2^{\mathcal{T}}-2^{\mathcal{T}-r}}\right].
	$
	Moreover, its dual is also an LCD code with parameters
	$
	[3^v2^{\mathcal{T}+1},\; 3^v(2^{\mathcal{T}+1}-2^{\mathcal{T}-r+1}),\\ d_{2^{\mathcal{T}}-2^{\mathcal{T}-r}}^\perp],
	$
	where
	\[
	d_{2^{\mathcal{T}}-2^{\mathcal{T}-r}} =
	\begin{cases}
		\dfrac{2^{r+2}-1}{3}, & \text{if } r \text{ is even},\\[6pt]
		\dfrac{2^{r+2}-2}{3}, & \text{if } r \text{ is odd},
	\end{cases}
	\]
	and $d_{2^{\mathcal{T}}-2^{\mathcal{T}-r}}^\perp$ is given by Theorem~\ref{theorem13}.
\end{theorem}
\begin{remark}
	Using the bounds on $LCD(n,k)$ from \cite{galvez2018bounds}, it follows that the code $C_{2^r}$ is LCD optimal whenever 
	$
	3^v 2^{\mathcal{T}+1} \geq 2^{3^v2^{r+1}}.
	$ Equivalently, for 
	$
	\mathcal{T} \geq 3^v2^{r+1} - \lceil v \log_2 3 \rceil - 1,
	$
	the code $C_{2^r}$ attains the largest possible minimum distance among all binary LCD codes with the same length and dimension.
	Moreover, the dual code $C_1^\perp$ is LCD optimal for $v=0$. 
\end{remark}

\section{Conclusion}\label{sec7} 
In this paper, we have studied binary polycyclic codes associated with powers of irreducible polynomials, with emphasis on their algebraic structure, Hamming distance, dual codes, and LCD property. We first obtained a complete structural description of these codes and then developed general methods for determining the Hamming distance of the codes and their Euclidean duals. Several examples were included to illustrate the theory and demonstrate the effectiveness of the proposed approach. We also investigated the LCD property of binary polycyclic codes, established necessary and sufficient conditions for such codes to be LCD codes, and constructed several families of binary LCD codes. These constructions further produced many optimal and LCD optimal binary LCD codes. A substantial part of the paper was devoted to polycyclic codes associated with powers of self-reciprocal irreducible trinomials. For this important class, we determined the Hamming distance of all such codes and showed that the corresponding codes are reversible.

In the present work, our attention was restricted to binary polycyclic codes associated with powers of irreducible polynomials. Future studies may involve the extension of these methodologies to non-binary fields, the exploration of polycyclic codes associated with other types of polynomials, and discovering more families of optimal and LCD-optimal codes..

\end{document}